\documentclass[11pt]{article}

\usepackage{fullpage}
\usepackage{amsmath}
\usepackage{amssymb}
\usepackage{amsthm}
\usepackage{thmtools,thm-restate}
\usepackage{graphicx}
\usepackage{fixmath}
\usepackage[english]{babel}
\usepackage[babel=true]{csquotes}
\usepackage[pass]{geometry}
\usepackage{float}
\usepackage{graphics}
\usepackage{subfig}
\usepackage{booktabs}
\usepackage{enumerate}
\usepackage{pxfonts}
\usepackage{color}
\usepackage[round, sort]{natbib}
\usepackage{hyperref}
\usepackage{lipsum}
\usepackage{tikz}
\usepackage{tikz-qtree}
\usepackage{upgreek}
\usepackage{appendix}
\usepackage{soul}
\usepackage{multirow}
\usepackage{caption}
\usepackage[noend]{algpseudocode}
\usepackage{algorithm}
\usepackage{multicol}
\usepackage{bbm}
\usepackage{bbding}
\usepackage[indentafter]{titlesec}
\graphicspath{{./figures/}}
\usetikzlibrary{%
decorations.fractals,%
decorations.shapes,%
decorations.text,%
decorations.pathmorphing,%
decorations.pathreplacing,%
decorations.footprints,%
decorations.markings}
\usetikzlibrary{positioning,shadows,arrows}

\usepackage{enumitem}
\newlist{alphalist}{enumerate}{1}
\setlist[alphalist,1]{label=\textbf{[Rule \arabic*]}}

\newcommand{\ra}[1]{\renewcommand{\arraystretch}{#1}}
\newtheorem{theorem}{Theorem}
\newtheorem{lemma}{Lemma}
\newtheorem{observation}{Observation}
\newtheorem{corollary}{Corollary}
\newtheorem{conjecture}{Conjecture}

\tikzstyle{decision} = [rectangle, draw=none, rounded corners=1mm, fill=blue, drop shadow,text centered, anchor=north, text=white]
    \tikzstyle{chance} = [circle, draw=none, fill=purple, circular drop shadow,text centered, anchor=north, text=white]
    \tikzstyle{decisionGood} = [rectangle, draw=none, rounded corners=1mm, fill=black!60!green, drop shadow,text centered, anchor=north, text=white]
    \tikzstyle{decisionCross} = [rectangle, draw=none, rounded corners=1mm, fill=black!20!orange, drop shadow,text centered, anchor=north, text=white]
     \tikzstyle{arrow} = [->,>=stealth]
     \tikzstyle{arrow1} = [ultra thick,->,dash pattern=on \pgflinewidth off 2pt, >=stealth, color=yellow]
     \tikzstyle{arrow2} = [thick,->,dash pattern=on 3pt off 3pt,>=stealth, color=orange]
     \tikzstyle{arrow3} = [thick,->,>=stealth, color=red]
     \tikzstyle{matrix} = [rectangle, rounded corners, minimum width=2cm, minimum height=2cm,text centered, draw=none, fill=none]

\usepackage{soul}
\makeatletter
\def\BState{\State\hskip-\ALG@thistlm}
\makeatother

\begin{document}
\title{The Stochastic Container Relocation Problem
}
\author{V. Galle	\thanks{Operations Research Center. Email: \url{vgalle@mit.edu}.} 
\and
S. Borjian Boroujeni \thanks{Operations Research Center. Email: \url{setareh@alum.mit.edu}}
\and
V. H. Manshadi \thanks{Yale School of Management. Email: \url{vahideh.manshadi@yale.edu}}
\and
C. Barnhart \thanks{Operations Research Center and Civil \& Environmental Engineering. Email: \url{cbarnhar@mit.edu}}
\and
P. Jaillet \thanks{Operations Research Center and Electrical Engineering \& Computer Science. Email: \url{jaillet@mit.edu}}
}

\date{}
\newgeometry{margin=2cm}
\maketitle

\begin{abstract}
The Container Relocation Problem (CRP) is concerned with finding a sequence of moves of containers that minimizes the number of relocations needed to retrieve all containers, while respecting a given order of retrieval. However, the assumption of knowing the full retrieval order of containers is particularly unrealistic in real operations. This paper studies the stochastic CRP (SCRP), which relaxes this assumption.
A new multi-stage stochastic model, called the \textit{batch model}, is introduced, motivated, and compared with an existing model (the \textit{online model}). The two main contributions are an optimal algorithm called \textit{Pruning-Best-First-Search (PBFS)} and a randomized approximate algorithm called \textit{PBFS-Approximate} with a bounded average error. Both algorithms, applicable in the batch and online models, are based on a new family of lower bounds for which we show some theoretical properties. Moreover, we introduce two new heuristics outperforming the best existing heuristics. Algorithms, bounds and heuristics are tested in an extensive computational section. Finally, based on strong computational evidence, we conjecture the optimality of the ``Leveling'' heuristic in a special ``no information'' case, where at any retrieval stage, any of the remaining containers is equally likely to be retrieved next. 
\end{abstract}

\section*{Introduction}
\label{sec:introduction}
With the growth in international container shipping in maritime ports, there has been an increasing interest in improving operations in container terminals, both on the sea side and land side. The operations on the sea side involve the assignment of quay cranes to ships, the loading of export containers on vessels, and the discharging of import containers from vessels onto internal trucks. Import containers are then transferred to the land side and are stacked in the storage area. Operations on the land side (also called yard operations) include the routing of internal trucks within the yard, the stacking of containers for storage, and the delivery of import containers to external trucks for delivery to another location. This work focuses on the latter problem.

\begin{figure}[h]
\centering
\includegraphics[width=0.45\textwidth]{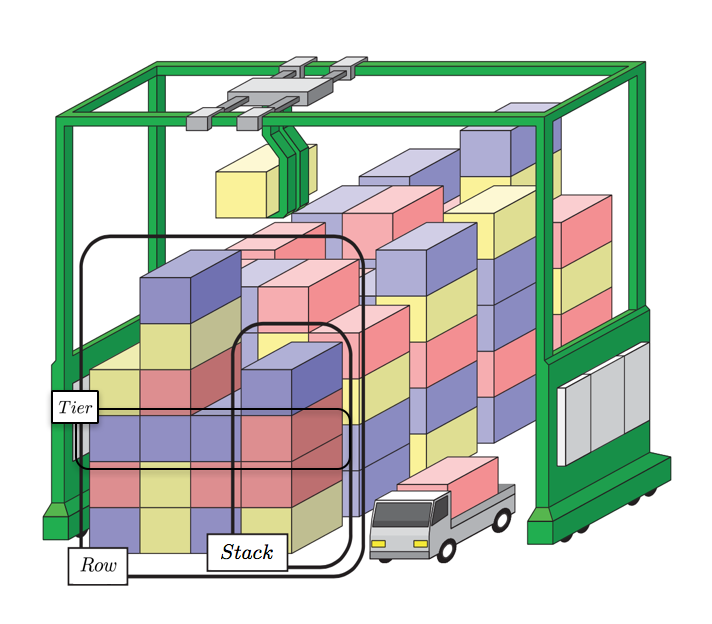}
\caption{Illustration of stacks of containers in a storage yard (figure from \cite{Tanaka})}
\label{fig:PortYard}
\end{figure}

Due to limited space in the storage area, containers are stacked on top of each other. The resulting stacks create rows of containers as shown in Figure \ref{fig:PortYard}. If a container that needs to be retrieved (\textit{target} container) is not located at a top most tier and is covered by other containers, the blocking containers must be relocated to another stack. As a result, during the retrieval process, one or more \textit{relocation} moves are performed by the yard cranes. Such relocations (also called reshuffles) are costly for the port operators and result in delays in the retrieval process. Thus, reducing the number of relocations is one of the main goals of port operators. The \textit{Container Relocation Problem (CRP)} (also known as the Block Relocation Problem) addresses this challenge by minimizing the number of relocations. As this problem is the main discussion of this paper, we provide a formal definition and an extensive literature review.

\paragraph{The container relocation problem}
First, it is commonly known that the time to relocate a container within a row is insignificant compared to the time to relocate a container between two distinct rows. Therefore, in most cases, port operators tend to avoid relocations between rows. The CRP makes the assumption that only relocations within rows are allowed, and problems for different rows should be considered independently. Furthermore, a row usually stores containers of the same type for the sake of stability and simplicity.

Using these facts, the CRP models one row using a two dimensional array of size $(T,S)$, where $S$ is the number of stacks, and $T$ is the maximum height, \textit{i.e.}, the maximum number of containers in a stack limited by the height of the crane. Each element of this array represents a potential slot for a container, and it contains a number only if a container is currently stored in this slot. Stacks are numbered from left ($1$) to right ($S$) and tiers from bottom ($1$) to top ($T$). We refer to this array as a \textit{configuration}. The common assumptions of the CRP are the following:
\begin{itemize}
	\item[] $\mathbf{A_1}$ : The initial configuration has $T$ tiers, $S$ stacks, and $C$ containers. In order for the problem to always be feasible, we suppose that the triplet $(T,S,C)$ satisfies $ 0 \leqslant C \leqslant ST - (T - 1 )$.
	\item[] $\mathbf{A_2}$ : A container can only be retrieved/relocated if it is at the top most tier of its stack, \textit{i.e.}, no other container is blocking it.
	\item[] $\mathbf{A_3}$ : A container can only be relocated if it is blocking the target container. This assumption was suggested by \cite{Caserta12}, and the problem under this assumption is commonly referred to as the \textit{restricted} CRP. Most studies focus on this restricted version, because it is the current practice in many yards, and it helps decrease the dimensionality of the problem, while not losing much optimality (see \cite{Petering}). As is common practice, we will not mention the term ``restricted'' in the rest of the paper even though we always assume $A_3$.
	\item[] $\mathbf{A_4}$ : The cost of relocating a container from a stack does not depend on, to which stack the container is relocated. This allows us to consider the stacks of a configuration as interchangeable. In addition, it motivates the objective of minimizing the number of relocations, since the cost of each relocation can be normalized to $1$. Note that this assumption is not required for all the results stated below, hence our approaches could be easily extended to the case when Assumption $A_4$ does not hold.
	\item[] $\mathbf{A_5}$ : The retrieval order of containers is known, so that each container can be labeled from $1$ to $C$, representing the departure order, \textit{i.e.}, Container $1$ is the first one to be retrieved, and $C$ the last one.
\end{itemize}

The CRP involves finding a sequence of moves to retrieve Containers $1, 2, \ldots, C$ (respecting the order) with a minimum number of relocations. Figure \ref{fig:exampleCRP} provides a simple example of the CRP. The CRP with the above classical assumptions is referred to as \textit{static and full information}: ``Static'' because no new containers arrive during the retrieval process (see Assumption $A_1$) and ``full information'' because we know the full retrieval order at the beginning of the retrieval process (see Assumption $A_5$). This problem was first formulated by \cite{Kim} in a dynamic programming model.

\begin{figure}[h]
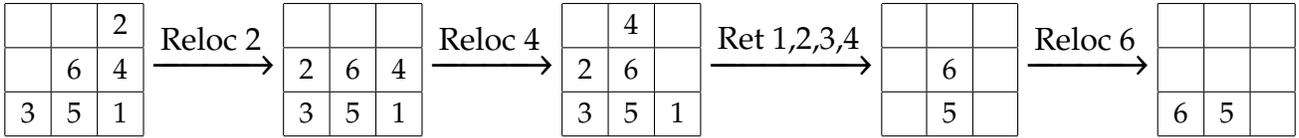

\centering
\begin{tabular}{|c|c|c|} \hline
   &   & 2  \\ \hline
   & 6 & 4  \\ \hline
 3 & 5 & 1 \\ \hline
\end{tabular}
{\LARGE$\xrightarrow{\text{Reloc } 2}$}
\begin{tabular}{|c|c|c|} \hline
   &   &   \\ \hline
 2 & 6 & 4   \\ \hline
 3 & 5 & 1 \\ \hline
\end{tabular}
{\LARGE$\xrightarrow{\text{Reloc } 4}$}
\begin{tabular}{|c|c|c|} \hline
   & 4 &   \\ \hline
 2 & 6 &    \\ \hline
 3 & 5 & 1 \\ \hline
\end{tabular}
{\LARGE$\xrightarrow{\text{Ret } 1,2,3,4}$}
\begin{tabular}{|c|c|c|} \hline
   &  &   \\ \hline
  & 6 &    \\ \hline
 \hspace{5pt} & 5 & \hspace{5pt} \\ \hline
\end{tabular}
{\LARGE$\xrightarrow{\text{Reloc } 6}$}
\begin{tabular}{|c|c|c|} \hline
   &   &   \\ \hline
  &  &    \\ \hline
 6 & 5 & \hspace{5pt} \\ \hline
\end{tabular}
\caption{Configuration for the CRP with 3 tiers, 3 stacks and 6 containers. The optimal solution performs 3 relocations: relocate the container labeled 2 from Stack $3$ to Stack $1$ on the top of the container labeled 3; relocate 4 from 3 to 2 on the top of 6; retrieve 1; retrieve 2; retrieve 3; retrieve 4; relocate 6 from 2 to the empty Stack $1$; retrieve 5; finally, retrieve 6.}
\label{fig:exampleCRP}
\end{figure}

Researchers have tackled the static CRP with full information from two point of views. The first approach aims to find the optimal solution. Primarily, researchers have used Integer Programming (IP) to address this problem. For example, \cite{Caserta12} propose an intuitive formulation of the problem. \cite{Petering} develop a more tractable formulation, that is, however, unable to solve real-sized instances efficiently. \cite{Zehendner15} fix the formulation from \cite{Caserta12} and improve it by removing some variables, tightening some constraints, introducing a new upper bound, and applying a pre-processing step to fix several variables. In all these IP formulations, due to the combinatorial nature of the problem, the number of variables and constraints dramatically increases as the size of the bay grows, and the IP cannot be solved for large instances. In order to bypass this problem, a recent trend has been to look at more efficient ways to explore the branch-and-bound tree, or even decrease its size using the structural properties of the problem. \cite{Unluyurt} and \cite{Exposito} suggest two branch-and-bound approaches with several heuristics based on this idea. Another solution using the $A^*$ algorithm is explored by \cite{Zhu}, and built upon by \cite{Tanaka} and \cite{Borjian15}. Another solution using branch-and-price is presented by \cite{Zehendner14b}.

As the problem is NP-hard (\cite{Caserta12}), an alternative approach is to use quick and efficient heuristics providing sub-optimal solutions. For the sake of conciseness, we only mention some of them that are relevant to this paper. \cite{Caserta12} introduce a 'MinMax' policy that is defined and generalized later in this paper. \cite{Wu10} propose a beam search heuristic, and \cite{Wu12} develop the Group Assignment Heuristic (GAH) also generalized in Section \ref{sec:Bounds}. Finally, in \cite{Kim} and \cite{Zhu}, lower bounds for the CRP are introduced: \cite{Kim} count the number of blocking containers as a straightforward lower bound, and \cite{Zhu} refine this idea by taking into account additional unavoidable relocations using a family of lower bounds.

Finally, there are many related problems to the CRP. The stacking problem is concerned with how to store incoming containers in a configuration given an arrival order of containers. The pre-marshalling problem deals with re-arranging the containers prior to the retrieval process in order to minimize future relocations, but no container is removed in this process. For both problems, general review and classification surveys of the existing literature on the CRP can be found in \cite{Stahlbock}, \cite{Steenken} and \cite{Lehnfeld}. In addition, Assumption $A_1$ can be relaxed, and this leads to the dynamic CRP where stacking and retrieving are done simultaneously as new containers are arriving. For this problem, see \cite{Borjian15_2} and \cite{Akyuz}.

Finally, the main focus of this paper is an extension of the CRP where the full information assumption ($A_5$) is relaxed. Indeed, Assumption $A_5$ is unrealistic given that arrival times of external trucks at the terminal are generally unpredictable due to uncertain conditions. Nevertheless, new technology advancements such as Truck Appointment Systems (TAS's) and GPS tracking can help predict relative truck arrival times. Thus, although the exact retrieval order might not be known, some information on trucks' arrival times might be available. This leads us to introduce a stochastic version of the CRP.

\paragraph{The Stochastic CRP (SCRP)}
A common assumption is that, for each container, there is a time window in which a truck driver will arrive to retrieve it. We refer to a \textit{batch of containers} as the set of containers stacked in the same row and with the same arrival time window. This information can be either inferred using machine learning algorithms, not yet much discussed in the literature, or can be obtained using the appointment time windows in a TAS, which has gained attention over the last decade. The first TAS was implemented by Hong Kong International Terminal (HIT) in 1997. It uses 30-minute time slots, where trucks can register (\cite{Murty}). Another TAS was introduced in New Zealand in 2007. Two other studies, \cite{Giuliano} and \cite{Morais}, evaluate the benefits of TAS, in reducing truck idling time by increasing on-time ratio. More recent information can be found in \cite{WallStreet} and \cite{JOC}.
 
On the modeling side, \cite{Zehendner14} formulate an IP to get the optimal number of slots a TAS should offer for each batch. Very few studies have tackled the SCRP, also referred to as CRP with Time Windows. This problem was first modeled by \cite{Zhao}. In the original model of \cite{Zhao}, each container is assigned to a batch. Batches of containers are ordered such that all containers in a batch must be retrieved before any containers in a later batch. Furthermore, the relative retrieval order of containers within a given batch is assumed to be a random permutation. From now on, we will refer to the model of \cite{Zhao} as the \textit{online model}. In Section \ref{sec:Model}, we discuss in more detail how this model assumes information is revealed. For the online model, \cite{Zhao} develop a myopic heuristic (called RDH) and study, in different settings with two or multiple groups, the value of information using RDH. They conclude that a small improvement in the information system reduces the number of relocations significantly. \Citet{Asperen} use a simulation tool to evaluate the effect of a TAS on many statistics including the ratio of relocations to retrievals. Their decision rules are based on several heuristics including the ``leveling,'' random, or ``traveling distance'' heuristics. More recently, \cite{Ku} also use the online model. They formulate the SCRP under the online model as a finite horizon dynamic programming problem, and suggest a decision tree scheme to solve it optimally. They also introduce a new heuristic called ERI (Expected Reshuffling Index), which outperforms RDH, and they perform computational experiments based on available test instances. We will refer frequently to this work, use some of their techniques, as well as their available test instances to evaluate our algorithms.

In another recent study related to the SCRP, \cite{Zehendner16} study the \textit{Online Container Relocation Problem}, which corresponds to an adversarial model. They prove that the number of relocations performed by the leveling policy can be upper-bounded by a linear function of the number of blocking containers and provide a tight competitive ratio for this policy. Moreover, \cite{Galle} show that the ratio of the expected number of relocations to the expected blocking lower bound converges to one. Finally, \cite{Tierney} study the robust pre-marshalling problem which also considers uncertainty in the retrieval times of containers.

For a general review of techniques on finite horizon Dynamic Programming, we refer the reader to \cite{Bertsekas} and \cite{Sennott}. Table \ref{tab:literatureReview} summarizes the previous literature review.

\paragraph{Contributions of the paper}
The contributions of this paper are:
\begin{enumerate}
\item A new stochastic model, referred to as the \textit{batch model}. This new model uses the same probability distribution as the online model. However, the two models are different in the way new information on the retrieval order is revealed. The batch model is motivated, described and compared with the online model.
\item Lower and upper bounds for the SCRP. We derive a \textbf{new family of lower bounds} for which we show theoretical properties. Furthermore, we develop \textbf{two new fast and efficient heuristics}.
\item \textbf{A novel optimal algorithm scheme based on decision trees and pruning strategies referred to as \textit{Pruning-Best-First-Search} ($PBFS$)}, taking advantage of the properties of the aforementioned lower bounds. The algorithm is explained with pseudocode in Algorithm \ref{algo:PBFS}.
\item \textbf{A second novel algorithm tuned for the case of larger batches referred to as $PBFSA$ (PBFS-Approximate)}. We build upon $PBFS$ and derive a sampling strategy resulting in an approximate algorithm with an expected error that we bound theoretically. The pseudocode of the second algorithm is presented in Algorithm \ref{algo:PBFSA}.
\item We provide \textbf{extensive computational experiments using an existing set of instances}. The first experiment exhibits the efficiency of $PBFS$, our lower bounds and two new heuristics for the batch model based on existing instances, presented by \cite{Ku}, where batches of containers are small (2 containers per batch on average). The second experiment illustrates the advantage of using $PBFSA$ when batches of containers are larger, based on instances obtained by modifying the existing set. In addition, most of our techniques including lower bounds, heuristics, and the $PBFS$ algorithm also apply to the online model. The third experiment shows that, in this model, $PBFS$ outperforms the algorithm introduced in \cite{Ku} in the sense that it is faster for instances that \cite{Ku} could solve, and it can solve problems of larger size. Furthermore, our two new heuristics also outperform the best existing heuristic (ERI) for the online model. Finally, the last experiment is used to test the conjecture about the optimality of the leveling heuristic in the special case of the online model with a unique batch of containers.
\end{enumerate}

\begin{table}[h]
  \centering
\ra{1.2}
\begin{tabular}{@{}ccc|cccc|cccc@{}}\toprule
&  \textit{CRP } &&&& \textit{Static} &&&& \textit{Dynamic} & \\
\midrule
&  &&&& \cite{Caserta12} &&&&  &  \\
&  &&&& \cite{Petering} &&&&  &  \\
&\textit{Full} &&&&  \cite{Zehendner15} &&&& \cite{Borjian15_2} &  \\
&\textit{Information} &&&& \cite{Exposito}   &&&& \cite{Akyuz} &  \\
&  &&&& \cite{Zhu} &&&&  &  \\
&  &&&& \cite{Tanaka} &&&&  &  \\
\midrule
& &&&& \cite{Zhao} &&&&  & \\
&\textit{Partial} &&&& \cite{Ku} &&&&  & \\
&\textit{Information} &&&& \textbf{This paper}  &&&&  &  \\
& &&&&  &&&&  & \\
\bottomrule
\end{tabular}
\caption{\textbf{Optimal solutions} for the different variant of the CRP}
\label{tab:literatureReview}
\end{table}

The rest of the paper is structured as follows: Section \ref{sec:Model} thoroughly describes the batch model, its assumptions and objective, the difference with the online model, and the general theory of decision trees applied to the SCRP. Section \ref{sec:Bounds} gives a good intuition into the problem and defines heuristics and a class of lower bounds for the SCRP used in subsequent sections. Then Sections \ref{sec:PBFS} and \ref{sec:PBFSA} introduce respectively the $PBFS$ and $PBFSA$ algorithms. Computational experiments for both batch and online models are carried through Section \ref{sec::ComputationExperiments}. We conclude the paper by discussing future directions for the SCRP in Section \ref{sec:conclusion}.

\section{SCRP and Decision Trees}
\label{sec:Model}
\subsection{Motivation}

Before stating the general assumptions of the batch model, let us motivate our problem using a typical example. We consider a port with a TAS offering 30-minute time windows during which truck drivers who want to retrieve a container can register to arrive at the port. For the sake of the example, we consider the time window between 9:00 am and 9:30 am. Multiple trucks can be registered in this time window: in this example, presented in Figure \ref{fig:BatchModelTimeLine}, 3 trucks (designated $i_1$, $i_4$ and $i_6$) are registered for this time window. We assume that \textbf{all 3 trucks arrive on time} (between 9:00 am and 9:30 am) and that their containers (similarly designated $i_1$, $i_4$ and $i_6$) form a batch to be retrieved. We display the configuration of interest in Figure \ref{fig:exampleSCRP} ($3$ tiers, $3$ stacks and $6$ containers).

\begin{figure}[h]
\centering
\includegraphics[width=0.8\textwidth]{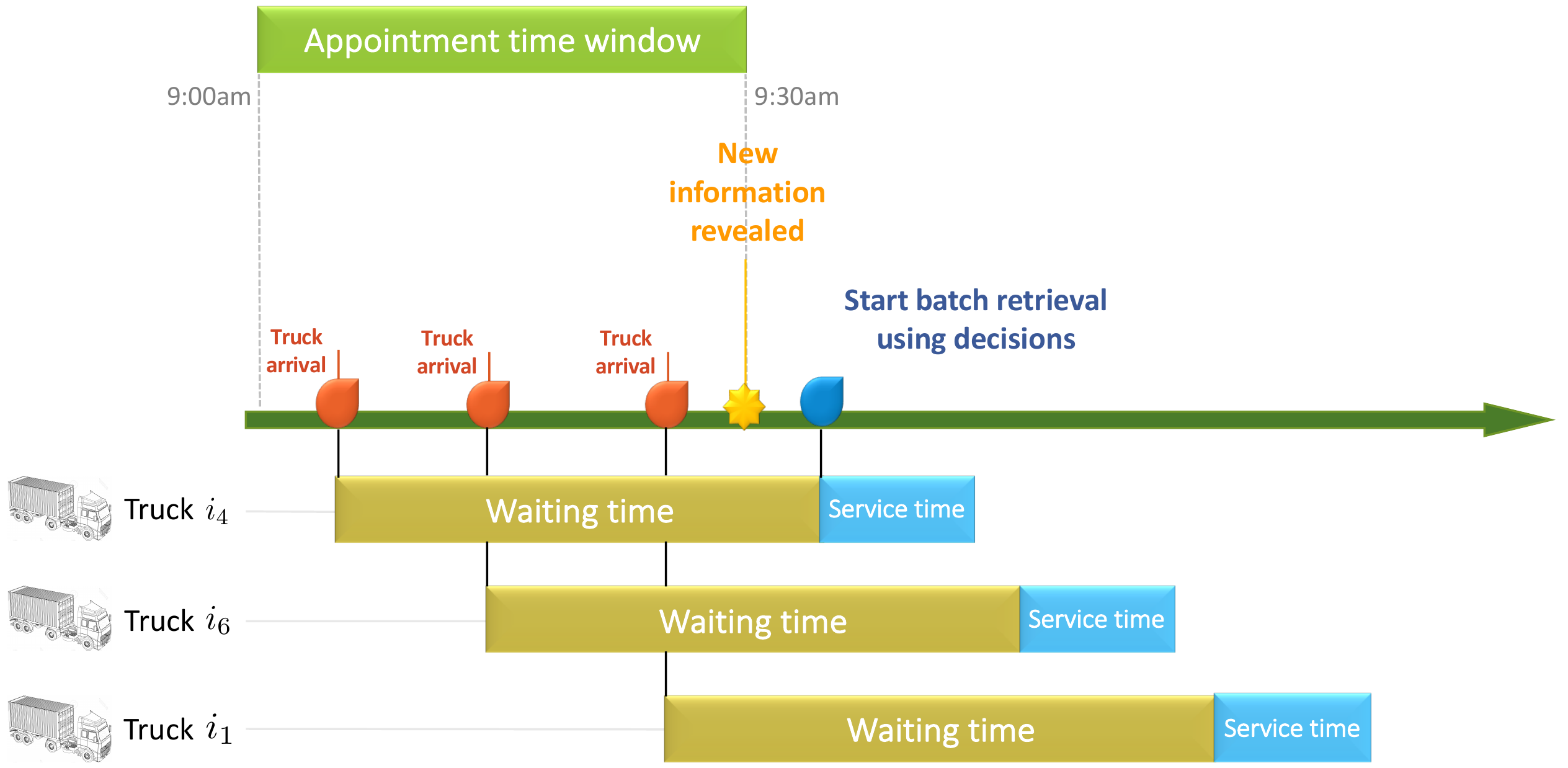}
\caption{Timeline of events for the batch model with three trucks}
\label{fig:BatchModelTimeLine}
\end{figure}

We assume that trucks arrive randomly within the time window, so \textbf{each truck arrival order is equally likely to happen}. In this example, there are 6 potential arrival orders, each with 1/6 chance of occurring.

At 9:00 am, none of the 3 trucks has arrived and \textbf{their relative retrieval order is unknown}. Consequently, these 3 containers are all labelled 1 in Figure \ref{fig:exampleSCRP}a. In Figure \ref{fig:exampleSCRP}b, the IDs of all containers and their locations are depicted.

Between 9:00 am and 9:30 am, trucks arrive in a particular order (\textit{e.g.}, Truck $i_4$ first, then $i_6$, and $i_1$ last). In busy terminals, trucks typically queue up as they wait to be served. Their place in line is based on their arrival order, so the port operator generally retrieves containers based on the arrival order. Processing in this way, \textbf{on a first-come-first-served basis}, avoids issues with truck unions and maintains fairness among drivers. Consequently, we take the retrieval order to be exogenously determined and we do not consider it a potential decision for port operators.

In order to provide a specified level of service to the truck drivers, the terminal operator often sets a target average waiting time. If the appointment time window is about the same as or shorter than the target average waiting time, the operator has information about the retrieval order of containers in the batch before the retrieval of those containers must begin in order to meet the target waiting time. Given this, we make the simplifying assumption in this work that the retrieval of a batch begins at the end of the appointment time window associated with the batch, and \textbf{the retrieval order of all containers in the batch is known as the retrieval of the batch commences}. In our example, the target average waiting time is 30 minutes. At 9:30 am, the retrieval order of the batch ($i_4$, $i_6$, $i_1$) is known and the retrieval of the batch commences soon after. The updated information is depicted in the configuration of Figure \ref{fig:exampleSCRP}c.

The assumption that containers to be retrieved are revealed on a batch basis models the reality that port operators typically know information about all the containers in the same batch before starting to retrieve them.
This is especially true for busy ports that have a TAS. Moreover, we assume that no information about future batches is available when making decisions for the current batch. Similar modeling assumptions have been made in previous works (see \cite{Zhao} and \cite{Ku}).

\begin{figure}[h]
\captionsetup{justification=centering}
\centering
\begin{minipage}{5cm}
\centering
\begin{tabular}{|c|c|c|} \hline
  &   & 1 \\ \hline
  & 5 & 4  \\ \hline
 1 & 5 & 1 \\ \hline
\end{tabular}
\caption*{\ref{fig:exampleSCRP}a Before any truck has arrived (9:00 am)}
\end{minipage}
\qquad
\begin{minipage}{5cm}
\centering
\begin{tabular}{|c|c|c|} \hline
 &   & $i_6$  \\ \hline
  & $i_3$ & $i_5$  \\ \hline
 $i_1$ & $i_2$ & $i_4$ \\ \hline
\end{tabular}
\caption*{\ref{fig:exampleSCRP}b IDs to match containers with trucks}
\end{minipage}
\qquad
\begin{minipage}{5cm}
\centering
\begin{tabular}{|c|c|c|} \hline
  &   & 2 \\ \hline
  & 5 & 4  \\ \hline
 3 & 5 & 1 \\ \hline
\end{tabular}
\caption*{\ref{fig:exampleSCRP}c Before the first container gets retrieved (9:30 am)}
\end{minipage}
\qquad
\caption{SCRP Example. The left configuration is the input to our problem. The configuration in the middle denotes each container with an ID $i_l$ where $l=1,\ldots,6$. The configuration on the right denotes the order of the first batch after it is revealed.}
\label{fig:exampleSCRP}
\end{figure}

The general assumptions we apply to our model are formally stated in Section \ref{subsec:assumptions} ($A^*_5$ and $A^*_6$) and result in the \textit{batch model -- the main focus of this paper}. The goal of the SCRP is to find a sequence of moves minimizing the expected number of relocations needed to empty the initial configuration.

Labels in Figures \ref{fig:exampleSCRP}a and \ref{fig:exampleSCRP}c are defined such that two containers only have the same label if they are in the same batch and their relative order is yet to be revealed. In our example, since Container $i_5$ is the only container in the second batch and is retrieved after the first 3 containers, it is necessarily the fourth container to be retrieved (thus labeled $4$). Containers $i_2$ and $i_3$ are labeled $5$ and when their relative order is revealed, one will be labeled $5$ and the other $6$.

\paragraph{The online model}

The main distinguishing difference between the batch model and the online model introduced by \cite{Zhao} is the information revelation process. 
The online model disregards any within batch information available when it plans the moves to retrieve a container. Hence \textbf{new information is revealed one container at a time.} The online model is especially applicable in less busy ports where the waiting time is significantly shorter than the appointment time windows. In this case, because there is a limited number of trucks waiting, limited information about the future is known. We show through Lemma \ref{lem:compModels} that ignoring information (if available) results in a potential loss of operational efficiency as measured by the expected number of relocations. In addition, most of our batch-based approaches also apply to the online model, and we provide numerical results based on it in Section \ref{sec::ComputationExperiments}.

\subsection{Assumptions, notations and formulation}
\label{subsec:assumptions}

In order to define our problem as a multi-stage stochastic optimization problem (the number of ``stages'' is the number of batches), we need to define \textit{a probabilistic model of the container retrieval order, a scheme for revealing new information about this order, and an objective function.}

\paragraph{Batch model} Let us state the assumptions and objective of the stochastic CRP under the batch model. Assumptions $A^*_1$, $A^*_2$, $A^*_3$ and $A^*_4$ are respectively identical to Assumptions $A_1$, $A_2$, $A_3$ and $A_4$.
\begin{itemize}
	\item[] $\mathbf{A^*_5}$ : (probabilistic model) Given an ordering of batches, the probability distribution of the retrieval order is such that: 1) \textbf{all the containers in a given batch are retrieved before any containers in a later batch}, and 2) \textbf{within each batch of container, the order of the containers is drawn from a uniform random permutation}. This paper focuses mainly on this latter assumption, but our model can be extended to the more general case of any probability distribution on permutations (not necessarily uniform) that respects the order of batches.
	\item[] $\mathbf{A^*_6}$ : (revelation of new information) For each batch $w$, \textbf{the full relative order of containers from the $w^{th}$ batch (\textit{i.e.}, the specific random permutation drawn) is revealed after all containers in batch $1$ through $w-1$ have been retrieved}.
\end{itemize}

Given these assumptions, we want to find the minimum expected number of relocations to retrieve all containers from a given initial configuration. We refer to this problem as the ``\textit{Stochastic CRP}''. Let us introduce some notations:
\begin{itemize}
\item[-] The problem size is given by $\left( T, S, C \right)$, respectively the number of tiers, stacks and containers in the initial configuration.
\item[-] The number of batches of containers in the initial configuration is denoted by $W$. We consider that the batches are ordered from 1 to $W$.
\item[-] For each batch $w \in \{1,\ldots,W\}$, let $C_w$ be the number of containers in $w$. By definition $\displaystyle \sum_{w=1}^W C_w = C$.
\item[-] Each container has two attributes:
\begin{itemize}
\item[$\bullet$] The first attribute, denoted by $\left( c_l \right)_{l \in \{1,\ldots,C\}}$, is the label and is defined as follows: initially, containers in batch $w$ are labeled by $K_w$ such that $\displaystyle K_w = 1 + \sum_{u=1}^{w-1} C_u$, where the sum is empty for $w=1$. Then, for $k \in \{1,\ldots,C\}$, if a container is revealed to be the $k^{th}$ container to be retrieved, its label changes to $k$. Using this labeling, at any point in the retrieval process, two containers only have the same label if there are in the same batch and their relative order is yet to be revealed.
\item[$\bullet$] The second attribute is a unique ID denoted by $\left( i_l \right)_{l \in \{1,\ldots,C\}}$. This ID is only used to identify uniquely containers in the initial configuration (see Figure \ref{fig:exampleSCRP}b) and for the sake of clarity of the following probabilistic model. Note that for $l \in \{1,\ldots,C\}$, if Container $i_l$ is in batch $w$, then $c_l = K_w$ (until the actual retrieval order of $i_l$ is revealed).
\end{itemize}
\item[-] For $k \in \{1,\ldots, C \}$, let $\zeta_k$ be a random variable taking values in $\left( i_l \right)_{l \in \{1,\ldots,C\}}$, such that $\{ \zeta_k = i_l \}$ is the event that Container $i_l$ is the $k^{th}$ container to be retrieved. According to Assumption $A^*_5$, the distribution of $\left( \zeta_k \right)_{k \in \{1,\ldots,C\}}$ is given by
\[ \mathbb{P} \left[ \zeta_k = i_l \right] = \left\{ \begin{array}{l}
\frac{1}{C_w} \text{ , if } w = \min \{ u \in \{1,\ldots,W\} \ | \ K_u \geqslant k\} \text{ and } c_l = K_w
\\
0 \text{ , otherwise}
\end{array}	 \right.\]
In this case, there are a total of $\prod_{w=1}^W \left( C_w ! \right)$ orders with equal probabilities. More generally, we consider the case where the probability of each order within each batch is not necessarily equally likely. However, we still assume that the batches are ordered, thus $\mathbb{P} \left[ \zeta_k = i_l \right]$ can be positive only if $w = \min \{ u | K_u \geqslant k\}$ and $c_l = K_w$. In the practical case where probabilities are not considered to be uniform, a list of potential retrieval orders and their associated probability is given for each batch of containers (based on historical data), thus $\mathbb{P} \left[ \zeta_k = i_l \right]$ can easily be inferred from these probabilities.
\item[-] An action corresponds to a sequence of relocations to retrieve one container. For $k \in \{1,\ldots,C\}$, we denote the action for the $k^{th}$ retrieval by $a_k$, and the feasible set of actions is defined according to Assumptions $A^*_2$ and $A^*_3$.
\item[-] For a given batch $w \in \{1,\ldots,W\}$,
\begin{enumerate}
\item let $y_w$ denote the configuration before the retrieval order of containers in batch $w$ is revealed, \textit{i.e.}, before $\zeta_{K_w},\ldots,\zeta_{K_w + C_w - 1}$ are revealed. Note that $y_1$ corresponds to the initial configuration. We denote $x_{K_w}$ the configuration after the retrieval order of containers in batch $w$ is revealed, and before action $a_{K_w}$ is taken. If $\xrightarrow{\zeta_{K_w},\ldots,\zeta_{K_w + C_w - 1}}$ represents the revelation of the random variables $\zeta_{K_w},\ldots,\zeta_{K_w + C_w - 1}$, we can write $y_w \ \xrightarrow{\zeta_{K_w},\ldots,\zeta_{K_w + C_w - 1}} \ x_{K_w}$. 
\item After the retrieval order of batch $w$ has been revealed, actions to retrieve the revealed containers have to be made. If $C_w > 1$, then $\{ K_w,\ldots,K_w+C_w-2 \} \neq \emptyset$. In this case, for $k \in \{K_w,\ldots,K_w+C_w-2\}$, let $x_{k+1}$ be the configuration after applying action $a_{k}$ to state $x_{k}$ and before action $a_{k+1}$. Therefore, if $\xrightarrow{a_{k}}$ represents the application of action $a_k$, we have $x_{k} \ \xrightarrow{a_{k}} \ x_{k+1}$.
\item The last container to be retrieved in batch $w$ is the $\left( K_w + C_w - 1 \right)$-th container, thus, according to the previous point, $x_{K_w + C_w - 1}$ corresponds to the configuration before $a_{K_w + C_w - 1}$ is taken. After this retrieval, the order of the next batch (batch $w+1$) has to be revealed, and according to the first point, the configuration is $y_{w+1}$. The configuration after retrieving batch $W$ will be empty, thus we define $y_{W+1}$ to be the empty configuration. So if $\xrightarrow{a_{K_w + C_w - 1}}$ represents the application of action $a_{K_w + C_w - 1}$, then we have $x_{K_w + C_w - 1} \ \xrightarrow{a_{K_w + C_w - 1}} \ y_{w+1}$.
\end{enumerate}
In summary, we have
\begin{equation*}
\forall \ w \in \{1,\ldots,W\} \ , \left\{ \begin{array}{ll}
y_w \ \xrightarrow{\zeta_{K_w},\ldots,\zeta_{K_w + C_w - 1}} \ x_{K_w}
\\
x_{k} \ \xrightarrow{a_{k}} \ x_{k+1}, \text{ if } C_w > 1 \ , \ \forall \ k \in \{K_w,\ldots,K_w + C_w - 2\}
\\
x_{K_w + C_w - 1} \ \xrightarrow{a_{K_w + C_w - 1}} \ y_{w+1}.
\end{array}
\right.
\end{equation*}
\item[-] Let the function $r(.)$ be such that $r(x)$ is number of relocations to retrieve the target container in configuration $x$. It is also equal to the number of containers blocking the target container. This function is only defined for configurations in which the target container is identified. Specifically, it is defined for $\left( x_k \right)_{k=1,\ldots,C}$ (but not for $\left( y_w \right)_{w=1,\ldots,W}$). For $k \in \{1,\ldots,C\}$, we refer to $r(x_k)$ as the \textit{immediate cost} for the $k^{th}$ retrieval.
\item[-] Let the function $f(.)$ be such that $f(x)$ is the minimum expected number of relocations required to retrieve all containers from configuration $x$. $f(x)$ is commonly referred to as the \textit{cost-to-go} function of configuration $x$. Note that it is well defined for both $\left( x_k \right)_{k=1,\ldots,C}$ and $\left( y_w \right)_{w=1,\ldots,W}$.
\end{itemize}

By definition, we have:
\begin{equation*}
\forall \ w \in \{1,\ldots,W\}, \ \left\{ \
\begin{array}{ll}
\displaystyle f\left(y_w\right) = \underset{\zeta_{k},\ldots,\zeta_{k+C_w-1}}{\mathbb{E}} \left[  f(x_k) \right] & , \text{ where } k = K_w,
\\
\displaystyle f\left(x_k \right) = r\left(x_k\right) + \min_{a_k} \left\{ f\left(x_{k+1}\right) \right\} & , \ \text{if} \ C_w > 1 \text{ and } \forall \ k \in \{K_w,\ldots,K_w + C_w - 2\},
\\
\displaystyle f\left(x_k \right) = r\left(x_k\right) + \min_{a_k} \left\{ f\left(y_{w+1}\right) \right\} & , \text{ where } k = K_w + C_w - 1,
\end{array}
\right.
\end{equation*}
which can be written as
\begin{equation}
\label{eq:fundrecrusion}
\forall \ w \in \{1,\ldots,W\}, \ \left\{ \
\begin{array}{l}
\displaystyle f\left(y_w\right) = \underset{\zeta_{K_{w}},\ldots,\zeta_{K_{w}+C_w-1}}{\mathbb{E}} \left[  f(x_{K_w}) \right]
,
\\
\displaystyle f\left(x_{K_w}\right) = \min_{a_{K_{w}},\ldots,a_{K_w + C_w - 1}} \left\{ \left( \sum_{k=K_{w}}^{K_w + C_w - 1} r\left(x_k\right) \right) + f\left(y_{w+1}\right) \right\},
\end{array} \right.
\end{equation}
and $f\left(y_{W+1}\right) = 0$.
Therefore, the SCRP is concerned with finding $f\left(y_1\right)$, where by induction:
\begin{eqnarray}
f\left(y_1\right) = \underset{\zeta_{K_{1}},\ldots,\zeta_{K_{1}+C_1-1}}{\mathbb{E}} \left[ \min_{a_{K_{1}},\ldots,a_{K_{1}+C_1-1}} \left\{ \underset{\zeta_{K_{2}},\ldots,\zeta_{K_{2}+C_2-1}}{\mathbb{E}} \left[ \ldots \underset{\zeta_{K_{W}},\ldots,\zeta_{K_{W}+C_W-1}}{\mathbb{E}}  \left[ \min_{a_{K_{W}},\ldots,a_{K_{W}+C_W-1}} \left\{ \sum_{k=1}^{C} r\left(x_k\right) \right\}  \right] \ldots \right] \right\} \right]. \label{eq:DPformulationbatch}
\end{eqnarray}

\paragraph{The online model}
Using our notations, we briefly present the SCRP under the online model. Instead of Assumption $A^*_6$, the online model assumes that for each retrieval $k \in \{1,\ldots,C\}$, only the next target container is revealed (\textit{i.e.} $\zeta_k$). Therefore, we consider the states $\left( y^o_k, x^o_k \right)_{k=1,\ldots,C}$ defined such that $k \in \{1,\ldots,C\}$, $y^o_k \ \xrightarrow{\zeta_{k}} \ x^o_k \ \xrightarrow{a_{k}} \ y^o_{k+1}$, where $y^o_{C+1}$ is the empty configuration.
In this case, if $f^o$ denotes the cost-to-go function, then by definition, we have $f^o\left(y^o_k\right) = \underset{\zeta_k}{\mathbb{E}} \left[ f^o\left(x^o_k\right) \right]$ (with $f^o\left(y^o_{C+1}\right) = 0$), and 
$\displaystyle \forall \ k \in \{1,\ldots,C\}, \ f^o\left(x^o_k\right) = \min_{a_k} \left\{ r\left(x^o_k\right) + f^o\left(y^o_{k+1}\right) \right\}$. By induction, the SCRP under the online model is hence concerned with finding:
\[f^o\left(y^o_1\right) = \underset{\zeta_1}{\mathbb{E}} \left[ \min_{a_1} \left\{ \underset{\zeta_2}{\mathbb{E}} \left[ \ldots \underset{\zeta_C}{\mathbb{E}} \left[ \min_{a_C} \left\{ \sum_{k=1}^{C} r\left(x^o_k\right) \right\}  \right] \ldots \right] \right\} \right].\]

The next lemma compares the batch and the online models theoretically (the proof can be found in the Appendix). It states that it is beneficial in terms of the expected number of relocations to use the batch model compared to the online model, if the first one is applicable. Practically, this suggests that the operator should always use available information.
\begin{restatable}{lemma}{compModels}
\label{lem:compModels}
Let $y$ be a given initial configuration, then we have
\[
f\left(y\right) \leqslant f^o\left(y\right).
\]
\end{restatable}

\subsection{Decision trees}
Multi-stage stochastic optimization problems can be solved using decision trees in which \textit{chance nodes} and \textit{decision nodes} typically alternate. A \textit{chance node} embodies the stochasticity of the model, while a \textit{decision node} models the possible actions of the algorithm. In a decision tree for the SCRP, a node represents a configuration. The root node (denoted by $0$) is the initial configuration, and the leaf nodes are configurations for which we can compute the cost-to-go function.

In our case, we slightly modify the structure of a typical decision tree, in the sense that chance nodes and decision nodes do not necessarily alternate. A chance node is a configuration for which the target container is not known yet, and information needs to be revealed (note that this only occurs at the beginning of each batch). A decision node is a configuration for which the target container is known. Using our notations, chance nodes correspond to $\left( y_w \right)_{w=1,\ldots,W}$ and decision nodes correspond to $\left( x_k \right)_{k=1,\ldots,C}$.

Let $n$ be a node corresponding to a configuration in the decision tree. Thus $f(n)$, the cost-to-go function of $n$, is defined for all nodes $n$, and $r(n)$, the immediate cost function of $n$, is well defined when $n$ is a decision node.
We denote by $\lambda_n$ the level of $n$ in the decision tree, and define it as the number of containers remaining in the configuration. We denote the lowest level of the tree by $\displaystyle \lambda^* = \min_{n \in Tree} \left\{ \lambda_n \right\}$. It corresponds to the level such that if $\lambda_n = \lambda^*$, $f(n)$ can be computed in an efficient way (the empty configuration being an obvious candidate with a cost-to-go of $0$). Moreover,
\begin{itemize}
\item If $n$ is a chance node, then there exists a unique $w \in \{1,\ldots,W\}$ such that $\lambda_n = C - K_w + 1$. We denote by $\Omega_n$ the set of offspring of $n$, each offspring being a decision node corresponding to a realization of the random variables $\zeta_{K_w},\ldots,\zeta_{K_w + C_w - 1}$,  \textit{i.e.}, the full retrieval order of containers in batch $w$. Note that $n$ has a priori $|\Omega_n| = C_w!$ offspring.
\item If $n$ is a decision node, then $r(n)$ is well defined and is equal to the number of containers blocking the target container in configuration $n$ (\textit{i.e.}, the $(C-\lambda_n+1)^{th}$ container to be retrieved). We denote by $\Delta_n$ the set of offspring of $n$, which can either be chance nodes if there exists $w \in \{1,\ldots,W\}$ such that $\lambda_n + 1 = C - K_w + 1$, or decision nodes otherwise. For the sake of simplicity, we compute $\Delta_n$ greedily by considering all feasible combinations of relocations of the $r(n)$ containers blocking the target container in $n$. Note that $|\Delta_n|$ is of the order of $r(n)^{S-1}$, where $S$ is the number of stacks.
\end{itemize}
Equation (\ref{eq:fundrecrusion}) provides the relation to compute the cost-to-go by back-tracking. For all $n$ in the decision tree, we have
\begin{eqnarray}
f(n) = \left\{ \begin{array}{l}
\displaystyle \frac{1}{| \Omega_n |} \sum_{n_i \in \Omega_n} f(n_i) \ \textit{, if n is a chance node},
\\
\displaystyle r(n) + \min_{n_i \in \Delta_n} \left\{ f(n_i) \right\} \ \textit{, if n is a decision node}.
\end{array}
\right.
\label{eq:fundamentalRecursion}
\end{eqnarray}

In the case in which the probability of each permutation (in each batch) is not uniform, we mentioned that in practice, operators provide the probability of potential orders for each batch. Given a chance node $n$ and one of its offspring $n_i \in \Omega_n$, this input probability is exactly the probability that the actual retrieval order is the one revealed in $n_i$. For a given node $n$, we denote these probabilities by $\left(p_{n_i} \right)_{n_i \in \Omega_n}$. In this case, Equation (\ref{eq:fundamentalRecursion}) is replaced by:
\begin{eqnarray*}
f(n) = \left\{ \begin{array}{l}
\displaystyle  \sum_{n_i \in \Omega_n} p_{n_i} f(n_i) \ \textit{, if n is a chance node},
\\
\displaystyle r(n) + \min_{n_i \in \Delta_n} \left\{ f(n_i) \right\} \ \textit{, if n is a decision node},
\end{array}
\right.
\end{eqnarray*}
for all $n$ in the decision tree.

Figures \ref{fig:treeNodes} and \ref{fig:treeBays} provide the description of the decision tree corresponding to the example in Figure \ref{fig:exampleSCRP}, using chance/decision nodes and configurations, respectively. Chance nodes are depicted with a circle, and decision nodes with a square.

\begin{figure}
\begin{center}
\resizebox{17cm}{8.5cm}{
\begin{tikzpicture} [node distance=2cm]
\node (chance0) [chance] {$00$};
\node (dec3) [decision, below of=chance0, xshift=-0.9cm] {$03$};
\node (dec4) [decision, right of=dec3] {$04$};
\node (dec5) [decision, right of=dec4] {$05$};
\node (dec6) [decision, right of=dec5] {$06$};
\node (dec2) [decision, left of=dec3] {$02$};
\node (dec1) [decision, left of=dec2] {$01$};
\node (dec9) [decision, below of=dec1] {$09$};
\node (dec8) [decision, left of=dec9] {$08$};
\node (dec7) [decision, left of=dec8] {$07$};
\node (dec10) [decision, right of=dec9] {$10$};
\node (dec11) [decision, right of=dec10] {$11$};
\node (dec12) [decision, right of=dec11] {$12$};
\node (dec13) [decision, right of=dec12] {$13$};
\node (dec14) [decision, below of=dec6] {$14$};
\node (dec15) [decision, right of=dec14] {$15$};
\node (dec16) [decision, right of=dec15] {$16$};
\node (dec17) [decision, below of=dec8] {$17$};
\node (dec18) [decision, right of=dec17] {$18$};
\node (dec19) [decision, right of=dec18] {$19$};
\node (dec20) [decision, right of=dec19] {$20$};
\node (dec21) [decision, right of=dec20] {$21$};
\node (dec22) [decision, right of=dec21] {$22$};
\node (dec23) [decision, right of=dec22] {$23$};
\node (dec24) [decision, right of=dec23] {$24$};
\node (chance26) [chance, below of=dec20, xshift=0.9cm] {$26$};
\node (chance25) [chance, left of=chance26] {$25$};
\node (chance27) [chance, right of=chance26] {$27$};

\draw [arrow] (chance0) -- (dec1);
\draw [arrow] (chance0) -- (dec2);
\draw [arrow] (chance0) -- (dec3);
\draw [arrow] (chance0) -- (dec4);
\draw [arrow] (chance0) -- (dec5);
\draw [arrow] (chance0) -- (dec6);
\draw [arrow3] (dec1) -- (dec7);
\draw [arrow3] (dec1) -- (dec8);
\draw [arrow3] (dec1) -- (dec9);
\draw [arrow1] (dec2) -- (dec10);
\draw [arrow1] (dec3) -- (dec11);
\draw [arrow1] (dec4) -- (dec12);
\draw [arrow1] (dec5) -- (dec13);
\draw [arrow3] (dec6) -- (dec14);
\draw [arrow3] (dec6) -- (dec15);
\draw [arrow3] (dec6) -- (dec16);
\draw [arrow2] (dec7) -- (dec17);
\draw [arrow3] (dec8) -- (dec17);
\draw [arrow3] (dec8) -- (dec18);
\draw [arrow3] (dec8) -- (dec19);
\draw [arrow2] (dec9) -- (dec19);
\draw [arrow1] (dec10) -- (dec20);
\draw [arrow1] (dec11) -- (dec20);
\draw [arrow2] (dec12) -- (dec21);
\draw [arrow2] (dec12) -- (dec22);
\draw [arrow3] (dec13) -- (dec21);
\draw [arrow3] (dec13) -- (dec22);
\draw [arrow3] (dec13) -- (dec24);
\draw [arrow1] (dec14) -- (dec21);
\draw [arrow1] (dec15) -- (dec22);
\draw [arrow2] (dec16) -- (dec22);
\draw [arrow2] (dec16) -- (dec23);
\draw [arrow1] (dec17) -- (chance25);
\draw [arrow1] (dec18) -- (chance25);
\draw [arrow1] (dec18) -- (chance25);
\draw [arrow1] (dec19) -- (chance26);
\draw [arrow2] (dec20) -- (chance26);
\draw [arrow2] (dec20) -- (chance27);
\draw [arrow2] (dec21) -- (chance25);
\draw [arrow2] (dec21) -- (chance26);
\draw [arrow1] (dec22) -- (chance26);
\draw [arrow1] (dec23) -- (chance25);
\draw [arrow1] (dec24) -- (chance27);
\end{tikzpicture}}
\caption{Decision tree represented with nodes. The colored arrows represent different values of immediate cost, \textit{i.e.}, the number of containers blocking the target container (dotted yellow: 0, dashed orange: 1, solid red: 2)}
\label{fig:treeNodes}
\end{center}
\end{figure}
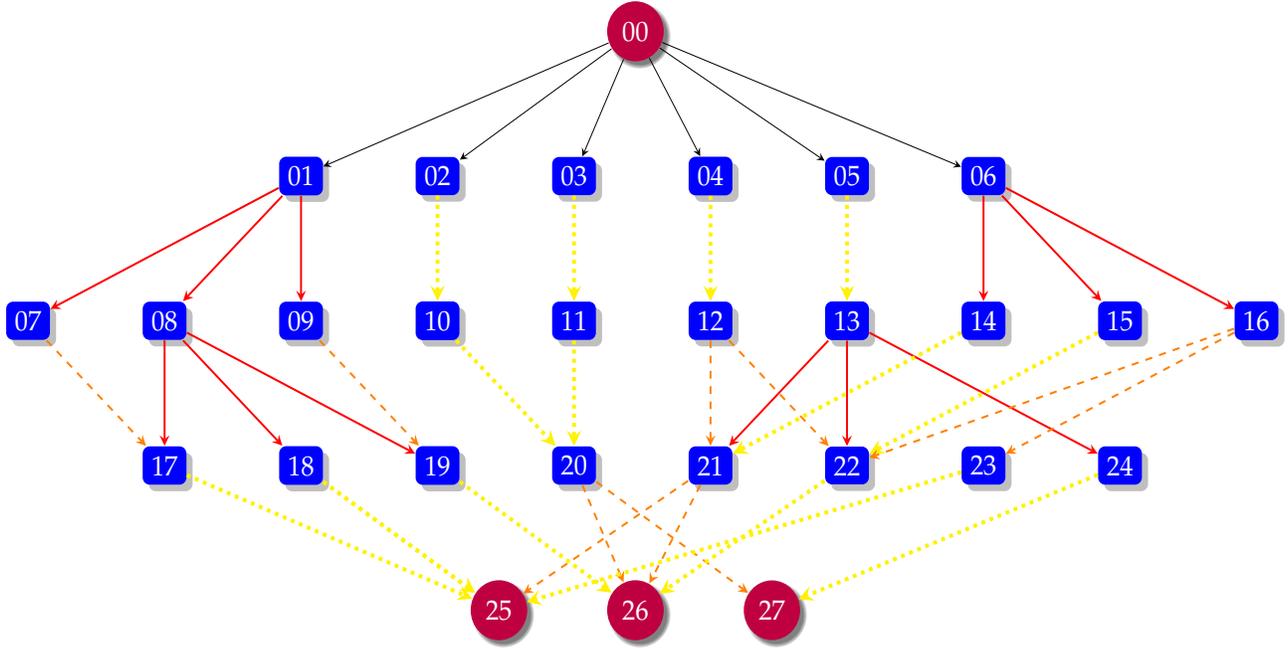

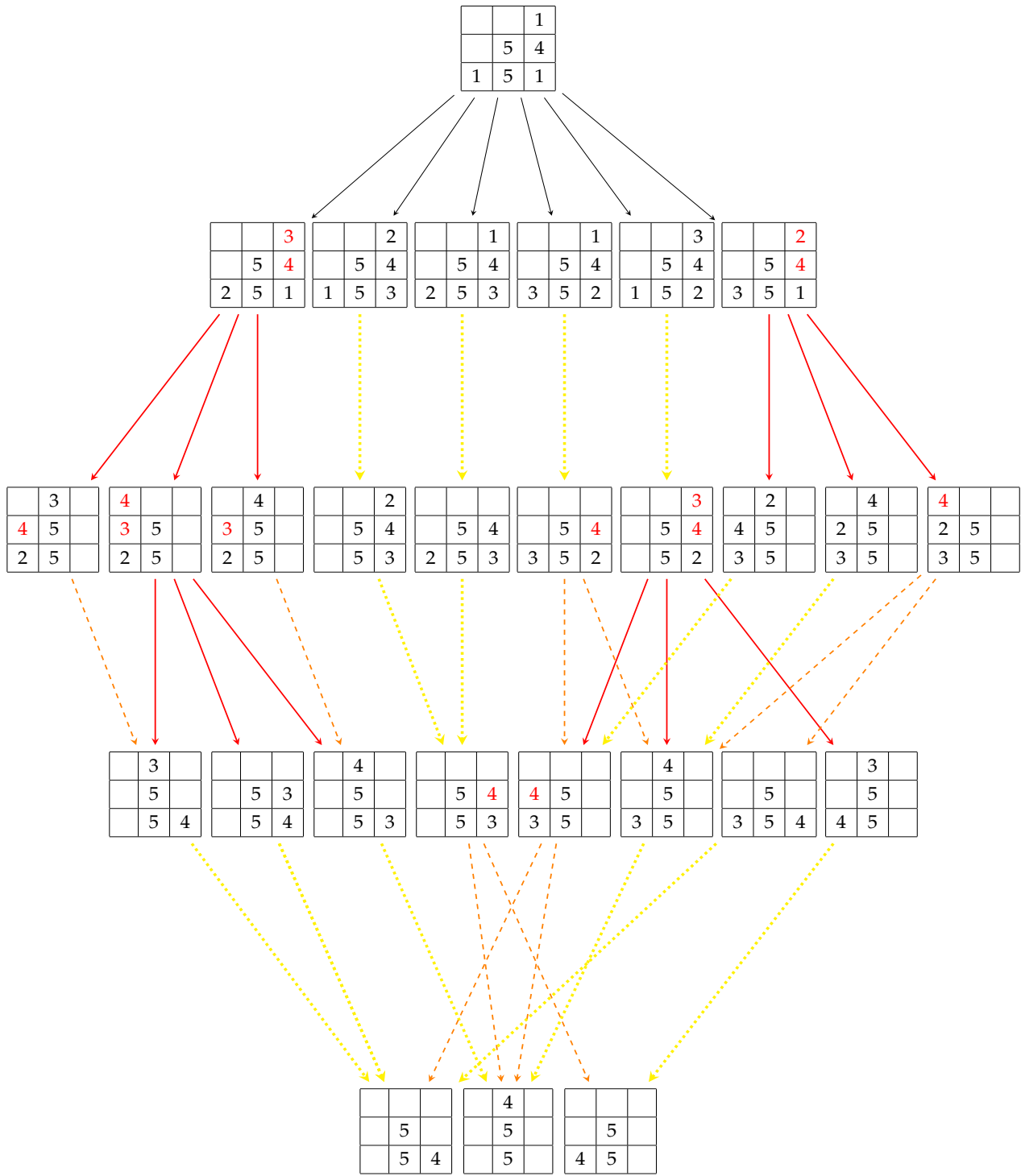
\begin{figure}
\begin{center}
\resizebox{17cm}{19.5cm}{
\begin{tikzpicture}[ node distance=2cm]
\node (chance0) [matrix] {\begin{tabular}{|c|c|c|} \hline
  &   & 1 \\ \hline
  & 5 & 4  \\ \hline
 1 & 5 & 1 \\ \hline
\end{tabular}};
\node (dec3) [matrix, below of=chance0, xshift=-0.9cm, yshift=-2.5cm] {\begin{tabular}{|c|c|c|} \hline
  &   & 1 \\ \hline
  & 5 & 4  \\ \hline
 2 & 5 & 3 \\ \hline
\end{tabular}};
\node (dec4) [matrix, right of=dec3] {\begin{tabular}{|c|c|c|} \hline
  &   & 1 \\ \hline
  & 5 & 4  \\ \hline
 3 & 5 & 2 \\ \hline
\end{tabular}};
\node (dec5) [matrix, right of=dec4] {\begin{tabular}{|c|c|c|} \hline
  &   & 3 \\ \hline
  & 5 & 4  \\ \hline
 1 & 5 & 2 \\ \hline
\end{tabular}};
\node (dec6) [matrix, right of=dec5] {\begin{tabular}{|c|c|c|} \hline
  &   & \textcolor{red}{2} \\ \hline
  & 5 & \textcolor{red}{4}  \\ \hline
 3 & 5 & 1 \\ \hline
\end{tabular}};
\node (dec2) [matrix, left of=dec3] {\begin{tabular}{|c|c|c|} \hline
  &   & 2 \\ \hline
  & 5 & 4  \\ \hline
 1 & 5 & 3 \\ \hline
\end{tabular}};
\node (dec1) [matrix, left of=dec2] {\begin{tabular}{|c|c|c|} \hline
  &   & \textcolor{red}{3} \\ \hline
  & 5 & \textcolor{red}{4}  \\ \hline
 2 & 5 & 1 \\ \hline
\end{tabular}};
\node (dec9) [matrix, below of=dec1, yshift=-3.5cm] {\begin{tabular}{|c|c|c|} \hline
  & 4  &  \\ \hline
 \textcolor{red}{3} & 5 &   \\ \hline
 2 & 5 & \hspace{4pt} \\ \hline
\end{tabular}};
\node (dec8) [matrix, left of=dec9] {\begin{tabular}{|c|c|c|} \hline
 \textcolor{red}{4} &   &  \\ \hline
 \textcolor{red}{3} & 5 &   \\ \hline
 2 & 5 & \hspace{4pt} \\ \hline
\end{tabular}};
\node (dec7) [matrix, left of=dec8] {\begin{tabular}{|c|c|c|} \hline
  & 3 &  \\ \hline
 \textcolor{red}{4} & 5 &   \\ \hline
 2 & 5 & \hspace{4pt} \\ \hline
\end{tabular}};
\node (dec10) [matrix, right of=dec9] {\begin{tabular}{|c|c|c|} \hline
  &  &  2 \\ \hline
  & 5 & 4  \\ \hline
 \hspace{4pt} & 5 & 3 \\ \hline
\end{tabular}};
\node (dec11) [matrix, right of=dec10] {\begin{tabular}{|c|c|c|} \hline
  &  &   \\ \hline
  & 5 & 4  \\ \hline
 2 & 5 & 3 \\ \hline
\end{tabular}};
\node (dec12) [matrix, right of=dec11] {\begin{tabular}{|c|c|c|} \hline
  &  &   \\ \hline
  & 5 & \textcolor{red}{4}  \\ \hline
 3 & 5 & 2 \\ \hline
\end{tabular}};
\node (dec13) [matrix, right of=dec12] {\begin{tabular}{|c|c|c|} \hline
  &  &  \textcolor{red}{3} \\ \hline
  & 5 & \textcolor{red}{4}  \\ \hline
 \hspace{4pt} & 5 & 2 \\ \hline
\end{tabular}};
\node (dec14) [matrix, right of=dec13] {\begin{tabular}{|c|c|c|} \hline
   & 2 &  \\ \hline
 4 & 5 &   \\ \hline
 3 & 5 & \hspace{4pt} \\ \hline
\end{tabular}};
\node (dec15) [matrix, right of=dec14] {\begin{tabular}{|c|c|c|} \hline
   & 4 &  \\ \hline
 2 & 5 &   \\ \hline
 3 & 5 & \hspace{4pt} \\ \hline
\end{tabular}};
\node (dec16) [matrix, right of=dec15] {\begin{tabular}{|c|c|c|} \hline
 \textcolor{red}{4} &  &   \\ \hline
 2 & 5 &   \\ \hline
 3 & 5 & \hspace{4pt} \\ \hline
\end{tabular}};
\node (dec17) [matrix, below of=dec8, yshift=-3.5cm] {\begin{tabular}{|c|c|c|} \hline
   & 3 &  \\ \hline
   & 5 &   \\ \hline
 \hspace{4pt} & 5 & 4 \\ \hline
\end{tabular}};
\node (dec18) [matrix, right of=dec17] {\begin{tabular}{|c|c|c|} \hline
   &  &  \\ \hline
   & 5 & 3 \\ \hline
 \hspace{4pt} & 5 & 4 \\ \hline
\end{tabular}};
\node (dec19) [matrix, right of=dec18] {\begin{tabular}{|c|c|c|} \hline
   & 4 &  \\ \hline
   & 5 &  \\ \hline
 \hspace{4pt} & 5 & 3 \\ \hline
\end{tabular}};
\node (dec20) [matrix, right of=dec19] {\begin{tabular}{|c|c|c|} \hline
   &  &  \\ \hline
   & 5 & \textcolor{red}{4} \\ \hline
 \hspace{4pt} & 5 & 3 \\ \hline
\end{tabular}};
\node (dec21) [matrix, right of=dec20] {\begin{tabular}{|c|c|c|} \hline
   &  &  \\ \hline
 \textcolor{red}{4}  & 5 &  \\ \hline
 3 & 5 & \hspace{4pt} \\ \hline
\end{tabular}};
\node (dec22) [matrix, right of=dec21] {\begin{tabular}{|c|c|c|} \hline
   & 4 &  \\ \hline
   & 5 &  \\ \hline
 3 & 5 & \hspace{4pt} \\ \hline
\end{tabular}};
\node (dec23) [matrix, right of=dec22] {\begin{tabular}{|c|c|c|} \hline
   &  &  \\ \hline
   & 5 &  \\ \hline
 3 & 5 & 4 \\ \hline
\end{tabular}};
\node (dec24) [matrix, right of=dec23] {\begin{tabular}{|c|c|c|} \hline
   & 3 &  \\ \hline
   & 5 &  \\ \hline
 4 & 5 & \hspace{4pt} \\ \hline
\end{tabular}};
\node (chance26) [matrix, below of=dec20, xshift=0.9cm, yshift=-5cm] {\begin{tabular}{|c|c|c|} \hline
   & 4 &  \\ \hline
   & 5 &  \\ \hline
 \hspace{4pt} & 5 & \hspace{4pt} \\ \hline
\end{tabular}};
\node (chance25) [matrix, left of=chance26] {\begin{tabular}{|c|c|c|} \hline
   &  &  \\ \hline
   & 5 &  \\ \hline
 \hspace{4pt} & 5 & 4 \\ \hline
\end{tabular}};
\node (chance27) [matrix, right of=chance26] {\begin{tabular}{|c|c|c|} \hline
   &  &  \\ \hline
   & 5 &  \\ \hline
 4 & 5 & \hspace{4pt} \\ \hline
\end{tabular}};

\draw [arrow] (chance0) -- (dec1);
\draw [arrow] (chance0) -- (dec2);
\draw [arrow] (chance0) -- (dec3);
\draw [arrow] (chance0) -- (dec4);
\draw [arrow] (chance0) -- (dec5);
\draw [arrow] (chance0) -- (dec6);
\draw [arrow3] (dec1) -- (dec7);
\draw [arrow3] (dec1) -- (dec8);
\draw [arrow3] (dec1) -- (dec9);
\draw [arrow1] (dec2) -- (dec10);
\draw [arrow1] (dec3) -- (dec11);
\draw [arrow1] (dec4) -- (dec12);
\draw [arrow1] (dec5) -- (dec13);
\draw [arrow3] (dec6) -- (dec14);
\draw [arrow3] (dec6) -- (dec15);
\draw [arrow3] (dec6) -- (dec16);
\draw [arrow2] (dec7) -- (dec17);
\draw [arrow3] (dec8) -- (dec17);
\draw [arrow3] (dec8) -- (dec18);
\draw [arrow3] (dec8) -- (dec19);
\draw [arrow2] (dec9) -- (dec19);
\draw [arrow1] (dec10) -- (dec20);
\draw [arrow1] (dec11) -- (dec20);
\draw [arrow2] (dec12) -- (dec21);
\draw [arrow2] (dec12) -- (dec22);
\draw [arrow3] (dec13) -- (dec21);
\draw [arrow3] (dec13) -- (dec22);
\draw [arrow3] (dec13) -- (dec24);
\draw [arrow1] (dec14) -- (dec21);
\draw [arrow1] (dec15) -- (dec22);
\draw [arrow2] (dec16) -- (dec22);
\draw [arrow2] (dec16) -- (dec23);
\draw [arrow1] (dec17) -- (chance25);
\draw [arrow1] (dec18) -- (chance25);
\draw [arrow1] (dec18) -- (chance25);
\draw [arrow1] (dec19) -- (chance26);
\draw [arrow2] (dec20) -- (chance26);
\draw [arrow2] (dec20) -- (chance27);
\draw [arrow2] (dec21) -- (chance25);
\draw [arrow2] (dec21) -- (chance26);
\draw [arrow1] (dec22) -- (chance26);
\draw [arrow1] (dec23) -- (chance25);
\draw [arrow1] (dec24) -- (chance27);
\end{tikzpicture}}
\caption{Decision tree represented with configurations. The colored arrows represent different values of immediate cost, \textit{i.e.}, the number of containers blocking the target container (dotted yellow: 0, dashed orange: 1, solid red: 2). Red numbers highlight containers blocking the target container.}
\label{fig:treeBays}
\end{center}
\end{figure}

To illustrate how to use Equation (\ref{eq:fundamentalRecursion}), we derive the calculations using the example in Figure \ref{fig:treeNodes}. Suppose that $f(25) = f(26) = f(27) = 0.5$ are known, then we get
$f(17) = f(18) = f(19) = 0.5$, $f(07) = f(09) = 1.5$ and $f(08) = 2.5$, leading to $f(01) = 3.5$. Similarly, by back-tracking, we can compute $f(02) = f(03) = f(04) = 1.5$ and $f(05) = f(06) = 2.5$ giving us $f(00) = 13/6$.

As the example shows, considering the full decision tree can become intractable even for the small examples, hence quickly impossible for larger problems. As previously mentioned, the number of decision offspring of a chance node scales with $C_w!$, and the number of offspring of a decision node is of the order of $r(n)^{S-1}$. So the size of the tree grows exponentially with the size of the problem. However, there exist general and specific techniques to reduce substantially the size of this tree, as we discuss now:

First, there are suboptimal approaches. One way to approximate $f(n)$, when $n$ is a chance node, is to sample from its offspring. When $\Omega_n$ is large (which can happen in the case of large batches), one might sample a certain number of realizations, resulting in a set of offspring $\Psi_{n} \subset \Omega_n$. By sampling ``enough,'' we show in Section \ref{sec:PBFSA}, that we can ensure guarantees on expectation for such an algorithm. Another popular suboptimal approach is to use techniques from Approximate Dynamic Programming. These techniques can provide good empirical results, but no guarantee on how far from optimality can be obtained. This direction is not discussed in this paper but can be a future direction of work. Finally, another approach is to consider heuristics such as the ones described in the next section, which select a subset of the offspring of decision nodes, and lead to upper bounds on the optimal value $f(0)$.

Second, there exist ways to decrease the size of the tree, while ensuring optimality. One of them is to reduce the number of nodes using the problem structure of the SCRP. In the online setting, \cite{Ku} propose an \textit{``abstraction''} technique, which shrinks significantly the size of the tree. Thanks to Assumption $A^*_4$, we can consider that the stacks of a configuration are interchangeable, making many configurations equivalent in terms of number of relocations. For instance, in Figure \ref{fig:treeNodes}, nodes $20$ and $21$ are identical in terms of number of relocations.

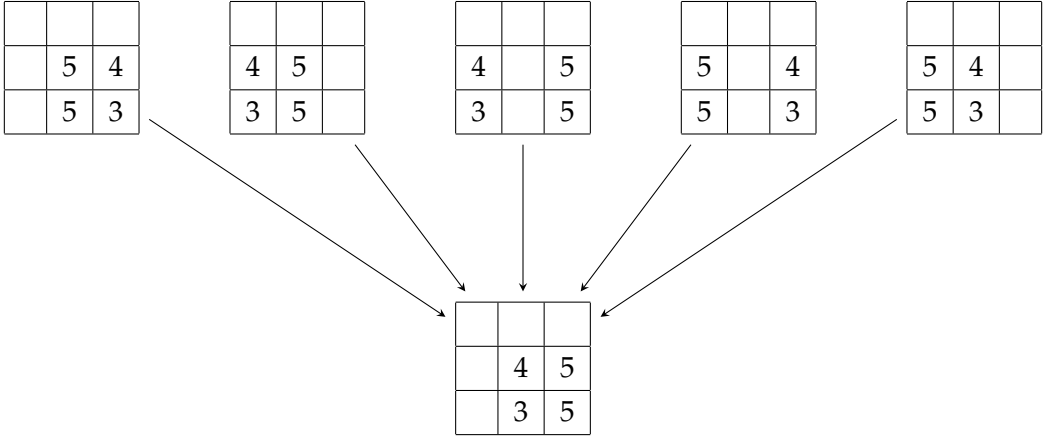
\begin{figure}
\begin{center}
\begin{tikzpicture} [node distance=3cm]
\node (abs1) [matrix] {\begin{tabular}{|c|c|c|} \hline
  &   &  \\ \hline
  & 5 & 4  \\ \hline
\hspace{4pt}  & 5 & 3 \\ \hline
\end{tabular}};
\node (abs2) [matrix, right of=abs1] {\begin{tabular}{|c|c|c|} \hline
  &   &  \\ \hline
4  & 5 &   \\ \hline
3  & 5 & \hspace{4pt} \\ \hline
\end{tabular}};
\node (abs3) [matrix, right of=abs2] {\begin{tabular}{|c|c|c|} \hline
  &   &  \\ \hline
4 &  & 5  \\ \hline
3 & \hspace{4pt} & 5 \\ \hline
\end{tabular}};
\node (abs4) [matrix, right of=abs3] {\begin{tabular}{|c|c|c|} \hline
  &   &  \\ \hline
5 &  & 4  \\ \hline
5 & \hspace{4pt} & 3 \\ \hline
\end{tabular}};
\node (abs5) [matrix, right of=abs4] {\begin{tabular}{|c|c|c|} \hline
  &   &  \\ \hline
5 & 4 &   \\ \hline
5 & 3 & \hspace{4pt} \\ \hline
\end{tabular}};
\node (root) [matrix,below of=abs3, yshift=-1cm] {\begin{tabular}{|c|c|c|} \hline
  &   &  \\ \hline
  & 4 &  5 \\ \hline
 \hspace{4pt} & 3 & 5 \\ \hline
\end{tabular}};

\draw [arrow] (abs1) -- (root);
\draw [arrow] (abs2) -- (root);
\draw [arrow] (abs3) -- (root);
\draw [arrow] (abs4) -- (root);
\draw [arrow] (abs5) -- (root);

\end{tikzpicture}
\caption{``Abstraction'' procedure}
\label{fig:abstraction}
\end{center}
\end{figure}

We describe the ``abstraction'' step with an example in Figure \ref{fig:abstraction}. The five configurations at the top are all equivalent to the configuration at the bottom. The abstraction transformation first ranks the columns by increasing height. For stacks with equal height, it breaks ties by ranking them lexicographically starting from bottom to top. Columns are re-arranged in order to have the first ranked on the left, and the last on the right. \cite{Ku} use a slightly different rule, and add the extra-step of removing empty columns. Using the abstraction procedure, the proposed algorithm should not generate twice nodes with identical abstracted versions. Even though \cite{Ku} introduce this rule for the online model, this ``Abstraction'' step also applies in the batch setting. Throughout the rest of the paper, we will refer in pseudocode to the function $\textproc{Abstract}(n)$, when applying this method to a given node $n$. Finally, we mention that \cite{Ku} also suggest caching strategies that could be added on the top of this simplification step, including caching part of the tree, or using a transportation table.

Finally, the performance of a decision tree based algorithm depends on the exploration strategy of the tree. For the online model, \cite{Ku} use depth-first-search (DFS), and we propose to explore the best-first-search (BFS) approach. Note that BFS requires some kind of measure that we define in Section \ref{sec:PBFS}.

In further sections, we explore two other ways to decrease the size of the tree while still ensuring optimality. The first one is specific to the SCRP, and uses properties of the problem to increase $\lambda^*$. Recall that $\lambda^*$ is the minimum level of the tree at which we can find the optimal expected number of relocations, without further branching. We show that we can set $\lambda^*$ to $\max\{S,C_W\}$, where $S$ is the number of stacks, and $C_W$ is the number of containers in the last batch. Comparatively, \cite{Ku} branch until $\lambda^* = 0$.
The second optimal pruning strategy uses lower bounds in a BFS scheme.

After introducing the batch model for the SCRP (as well as the online model) and some preliminary concepts about decision trees, the next three sections develop approaches to solve the SCRP.

\section{Heuristics and lower bounds}
\label{sec:Bounds}

Before introducing the two main algorithms, we decribe in this section heuristics and lower bounds for the SCRP. Indeed, $PBFS$ and $PBFSA$ build upon some of these bounds. In addition, these algorithms provide good intuition on how to solve the SCRP.
 
Let $n$ be a given configuration with $S$ stacks and $T$ tiers. We say that a Container $c$ is a \textit{blocking container} in $n$ if it is stacked above at least one container which is to be removed before $c$. Note that all the bounds mentioned below apply both in the batch and online models.

\subsection{Heuristics}

For the sake of completeness of this paper, we first describe three existing heuristics used in proofs and/or our computational experiments before describing two novel heuristics.

\subsubsection{Existing heuristics}
\paragraph{Random}
For every relocation of a blocking Container $c$ from Stack $s$, the random heuristic picks any Stack $s' \neq s$ uniformly at random among stacks that are not ``full,'' \textit{i.e.}, stacks containing strictly less than $T$ containers.

\paragraph{Leveling (L)}
For every relocation of a blocking Container $c$ from Stack $s$, L chooses the Stack $s' \neq s$ currently containing the least number of containers, breaking ties arbitrarily by selecting such leftmost stack.

Heuristic L is interesting for several reasons. Most importantly, it is an intuitive and commonly used heuristics in real operations as it uses no more than the height of each stack. It does not require any information about batches or departure times, which means it is robust with respect to the inaccuracy of information. In addition, it is optimal for any configurations with $S$ containers or less (see Section \ref{sec:PBFS}). Finally, we show strong evidence in the last computational experiment (see Section \ref{sec::ComputationExperiments}), that this policy is optimal for the SCRP under the online model with a unique batch (representing the case of no-information on the retrieval order).

\paragraph{Expected Reshuffling Index (ERI)}
This index-based heuristic was introduced by \cite{Ku} for the online model. For every relocation of a blocking Container $c$ from Stack $s$, ERI computes a score called the expected reshuffling index for each Stack $s' \neq s$ that is not full, denoted by $ERI(s',c)$. ERI chooses the Stack $s' \neq s $ with the lowest $ERI(s',c)$. In the case of a tie, the policy breaks it by selecting the highest column among the ones minimizing $ERI(s',c)$. Further ties are arbitrarily broken by selecting the leftmost column verifying the two previous conditions. $ERI(s',c)$ corresponds to the expected number of containers in Stack $s'$ that depart earlier than $c$. Let $H_{s'}$ be the current number of containers in $s'$. If $H_{s'}=0$, then $ERI(s',c) = 0$. Otherwise, let $(c_1,\ldots,c_{H_{s'}})$ be the containers in $s'$, then $\displaystyle ERI(s',c) = \sum_{i=1}^{H_{s'}} \mathbbm{1} \left\{ c_i < c \right\} + \frac{\mathbbm{1} \left\{ c_i = c \right\}}{2}$.

\subsubsection{First new heuristic: Expected Minmax (EM)}
EM considers an idea similar to that of \cite{Caserta12} for the CRP. Let $min(s)$ be the smallest label of a container in $s$ ($min(s) = C + 1$, if Stack $s$ is empty). For every relocation of a blocking Container $c$ from Stack $s$, we select the stack to which we relocate $c$ using the following rules:
\begin{alphalist}
\item If there exists $s' \neq s$ such that $min(s')>c$, let $\displaystyle M = \min_{s' \in \{ 1,\ldots,S \} \setminus s} \lbrace{\min (s'):\min (s')>c}\rbrace$. Select a stack such that $min(s') = M$, breaking ties by choosing from the highest ones, finally taking the leftmost stack if any ties remain.
\item If for all Stacks $s' \neq s, \ min(s') \leqslant c$, let $\displaystyle M = \max_{s' \in \{ 1,\ldots,S \} \setminus s} \lbrace{\min (s')}\rbrace$. Select a stack such that $min(s') = M$. If there are several such stacks, select those with the minimum number of containers labeled $M$. Further ties are again broken by taking the highest ones, and finally choosing the leftmost one arbitrarily.
\end{alphalist}

\begin{figure}[h]
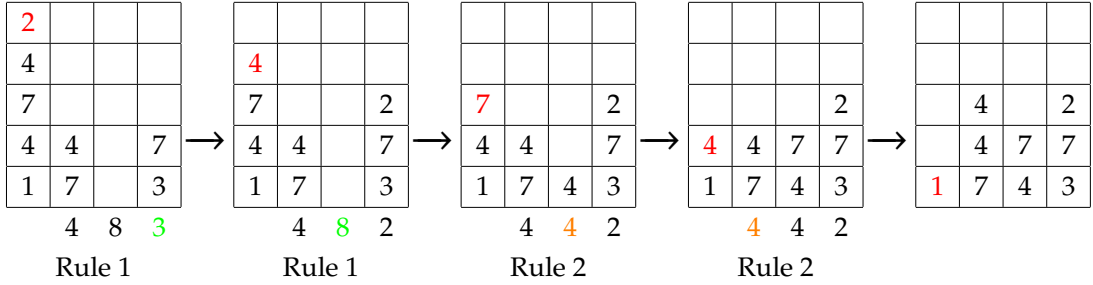

\centering
\resizebox{14.5cm}{2cm}{
\begin{tabular}{|c|c|c|c|} \hline
 \textcolor{red}{2} &   & \hspace{5pt} &   \\ \hline
 4 &   & \hspace{5pt} &   \\ \hline
 7 &   & \hspace{5pt} &   \\ \hline
 4 & 4 & \hspace{5pt} & 7 \\ \hline
 1 & 7 & \hspace{5pt} & 3 \\ \hline
 \multicolumn{1}{c}{} & \multicolumn{1}{c}{4} & \multicolumn{1}{c}{8} & \multicolumn{1}{c}{\textcolor{green}{3}} \\
\multicolumn{4}{c}{Rule 1}
\end{tabular}
{\LARGE$\xrightarrow{}$}
\begin{tabular}{|c|c|c|c|} \hline
   &   & \hspace{5pt} &   \\ \hline
 \textcolor{red}{4} &   & \hspace{5pt} &   \\ \hline
 7 &   & \hspace{5pt} & 2 \\ \hline
 4 & 4 & \hspace{5pt} & 7 \\ \hline
 1 & 7 & \hspace{5pt} & 3 \\ \hline
  \multicolumn{1}{c}{} & \multicolumn{1}{c}{4} & \multicolumn{1}{c}{\textcolor{green}{8}} & \multicolumn{1}{c}{2} \\
\multicolumn{4}{c}{Rule 1}
\end{tabular}
{\LARGE$\xrightarrow{}$}
\begin{tabular}{|c|c|c|c|} \hline
   &   &   &   \\ \hline
   &   &   &   \\ \hline
 \textcolor{red}{7} &   &   & 2 \\ \hline
 4 & 4 &   & 7 \\ \hline
 1 & 7 & 4 & 3 \\ \hline
  \multicolumn{1}{c}{} & \multicolumn{1}{c}{4} & \multicolumn{1}{c}{\textcolor{orange}{4}} & \multicolumn{1}{c}{2} \\
\multicolumn{4}{c}{Rule 2}
\end{tabular}
{\LARGE$\xrightarrow{}$}
\begin{tabular}{|c|c|c|c|} \hline
   &   &   &   \\ \hline
   &   &   &   \\ \hline
   &   &   & 2 \\ \hline
 \textcolor{red}{4} & 4 & 7 & 7 \\ \hline
 1 & 7 & 4 & 3 \\ \hline
  \multicolumn{1}{c}{} & \multicolumn{1}{c}{\textcolor{orange}{4}} & \multicolumn{1}{c}{4} & \multicolumn{1}{c}{2} \\
\multicolumn{4}{c}{Rule 2}
\end{tabular}
{\LARGE$\xrightarrow{}$}
\begin{tabular}{|c|c|c|c|} \hline
   &   &   &   \\ \hline
   &   &   &   \\ \hline
   & 4 &   & 2 \\ \hline
   & 4 & 7 & 7 \\ \hline
 \textcolor{red}{1} & 7 & 4 & 3 \\ \hline
  \multicolumn{1}{c}{} & \multicolumn{1}{c}{} & \multicolumn{1}{c}{} & \multicolumn{1}{c}{} \\
\multicolumn{4}{c}{}
\end{tabular}}
\caption{Decisions of the EM heuristic on an example with 5 tiers, 4 stacks and 9 containers (3 per batch). Under the batch model, the first batch has been revealed and we present the decisions to retrieve the first container made by EM. The container with the red label is the current blocking container. Numbers under the configuration correspond to the stack indices $min(s)$. The green (respectively orange) indices correspond to the selected stack with the corresponding $M$ when Rule 1 (respectively Rule 2) applies.}
\label{fig:exampleEM}
\end{figure}

Rule 1 says: if there is a stack where $\min (s)$ is greater than $c$ ($c$ can almost surely avoid being relocated again), then choose such a stack where $\min (s)$ is minimized, since stacks with larger minimums can be useful for larger blocking containers.

If there is no stack satisfying $\min (s) > c$ (Rule 2), then we have two cases following the same rule. On one hand, if $M=c$, then $M$ is the maximum of the minimum labels of each stack, and $c$ can potentially avoid being relocated again. If there are several stacks that maximize the minimum label, then by selecting the one with the least number of containers labeled $M$, EM minimizes the probability of $c$ being relocated again. 
On the other hand, if $M<c$, $c$ will almost surely be relocated again, then EM chooses the stack where $\min (s)$ is maximized in order to delay the next unavoidable relocation of $c$ as much as possible. We show how EM makes decision on a simple example in Figure \ref{fig:exampleEM}.

\subsubsection{Second new heuristic: Expected Group assignment (EG)}
EM is quite intuitive because it tries to minimize the number of blocking containers after each retrieval. EG aims for the same goal, but uses some more sophisticated rules (althought, as shown in the experiments in Section \ref{sec::ComputationExperiments}, EG does not always provide better solutions that EM). EG is inspired by a heuristic designed by \cite{Wu12} for the complete information case, and we generalize this idea to the SCRP. It is different from ERI and EM because it considers a group of blocking containers together, while ERI and EM consider them one at a time. EG can be decomposed in two main phases for each retrieval. The decisions made by EG on the same example as Figure \ref{fig:exampleEM} are given in Figure \ref{fig:exampleEG}.

\begin{figure}[h]
\centering
\begin{minipage}{\textwidth}
\centering
\resizebox{14.5cm}{2cm}{
\begin{tabular}{|c|c|c|c|} \hline
 2 &   & \hspace{5pt} &   \\ \hline
 4 &   & \hspace{5pt} &   \\ \hline
 \textcolor{red}{7} &   & \hspace{5pt} &   \\ \hline
 4 & 4 & \hspace{5pt} & 7 \\ \hline
 1 & 7 & \hspace{5pt} & 3 \\ \hline
 \multicolumn{1}{c}{} & \multicolumn{1}{c}{4} & \multicolumn{1}{c}{\textcolor{green}{8}} & \multicolumn{1}{c}{3} \\
\multicolumn{4}{c}{}
\end{tabular}
{\LARGE$\xrightarrow{}$}d
\begin{tabular}{|c|c|c|c|} \hline
 2 &   &   &   \\ \hline
 \textcolor{red}{4} &   &   &   \\ \hline
   &  &  &   \\ \hline
 4 & 4 &  & 7 \\ \hline
 1 & 7 & \textcolor{blue}{7} & 3 \\ \hline
  \multicolumn{1}{c}{} & \multicolumn{1}{c}{4} & \multicolumn{1}{c}{$\times$} & \multicolumn{1}{c}{2} \\
\multicolumn{4}{c}{{\fontsize{9.5}{9.5}\selectfont Non-assigned}}
\end{tabular}
{\LARGE$\xrightarrow{}$}
\begin{tabular}{|c|c|c|c|} \hline
 2 &   &   &   \\ \hline
 \textcolor{gray}{4} &   &   &   \\ \hline
   &  &  &   \\ \hline
 \textcolor{red}{4} & 4 &  & 7 \\ \hline
 1 & 7 & \textcolor{blue}{7} & 3 \\ \hline
  \multicolumn{1}{c}{} & \multicolumn{1}{c}{4} & \multicolumn{1}{c}{\textcolor{green}{7}} & \multicolumn{1}{c}{2} \\
\multicolumn{4}{c}{}
\end{tabular}
{\LARGE$\xrightarrow{}$}
\begin{tabular}{|c|c|c|c|} \hline
 \textcolor{red}{2} &   &   &   \\ \hline
 \textcolor{gray}{4} &   &   &   \\ \hline
   &  &  &   \\ \hline
   & 4 & \textcolor{blue}{4} & 7 \\ \hline
 1 & 7 & \textcolor{blue}{7} & 3 \\ \hline
  \multicolumn{1}{c}{} & \multicolumn{1}{c}{4} & \multicolumn{1}{c}{$\times$} & \multicolumn{1}{c}{\textcolor{green}{3}} \\
\multicolumn{4}{c}{}
\end{tabular}
{\LARGE$\xrightarrow{}$}
\begin{tabular}{|c|c|c|c|} \hline
   &   &   &   \\ \hline
 \textcolor{gray}{4} &   &   &   \\ \hline
   &  &  &  \textcolor{blue}{2} \\ \hline
   & 4 & \textcolor{blue}{4} & 7 \\ \hline
 1 & 7 & \textcolor{blue}{7} & 3 \\ \hline
  \multicolumn{1}{c}{} & \multicolumn{1}{c}{} & \multicolumn{1}{c}{} & \multicolumn{1}{c}{} \\
\multicolumn{4}{c}{}
\end{tabular}}
\caption*{\ref{fig:exampleEG}a. First phase: EG assigns acceptable containers in descending order. The container with the red label is the next acceptable container that EG tries to assign to a stack. Containers with blue labels are assigned, gray are unassigned. Below, we show the indices $min(s)$ to apply the first rule of EM ($\times$ means that a container below the considered container has already been assigned to a stack).}
\end{minipage}
\qquad
\begin{minipage}{\textwidth}
\centering
\resizebox{6cm}{2cm}{
\begin{tabular}{|c|c|c|c|} \hline
   &   &   &   \\ \hline
 \textcolor{red}{4} &   &   &   \\ \hline
   &  &  &  \textcolor{blue}{2} \\ \hline
   & 4 & \textcolor{blue}{4} & 7 \\ \hline
 1 & 7 & \textcolor{blue}{7} & 3 \\ \hline
  \multicolumn{1}{c}{} & \multicolumn{1}{c}{\textcolor{orange}{4}} & \multicolumn{1}{c}{0} & \multicolumn{1}{c}{2}
\end{tabular}
{\LARGE$\xrightarrow{}$}
\begin{tabular}{|c|c|c|c|} \hline
   &   &   &   \\ \hline
   &   &   &   \\ \hline
   & \textcolor{blue}{4} &  & \textcolor{blue}{2} \\ \hline
   & 4 & \textcolor{blue}{4} & 7 \\ \hline
 \textcolor{red}{1} & 7 & \textcolor{blue}{7} & 3 \\ \hline
  \multicolumn{1}{c}{} & \multicolumn{1}{c}{} & \multicolumn{1}{c}{} & \multicolumn{1}{c}{}
\end{tabular}}
\caption*{\ref{fig:exampleEG}b. Second Phase, EG assigns all unassigned containers using the index $Gmin(s)$.}
\end{minipage}
\caption{Decisions of the EG heuristic on an example with 5 tiers, 4 stacks and 9 containers (3 in each batch). Under the batch model, the first batch has been revealed and we present the two phase decisions to retrieve the first container made by EG.}
\label{fig:exampleEG}
\end{figure}

The first phase assigns the blocking containers for which there exists $s' \neq s$ such that $min(s')>c$. If this is not the case, the assignments of these containers will be ignored at the first phase. The acceptable containers are assigned in descending order of labels, \textit{i.e.}, the container with highest label is assigned first (breaking ties for the highest one first). In order to assign these acceptable containers, the first phase applies the first of the EM rules. Finally, acceptable containers cannot be assigned to a stack if there is a container below it that was previously assigned to this stack.

The assignment in the second phase for the blocking containers not assigned yet, might lead to additional relocations. These containers are assigned to other stacks in ascending order of labels. The second phase first computes a modified $min(s')$ index for each stack denoted by $Gmin(s')$, which is defined as follows: Let $H_s'$ be the height of Stack $s'$ and $B(s')$ be the subset of containers assigned in the first phase to Stack $s'$,
\begin{eqnarray*}
Gmin(s') = \left\{
\begin{array}{l}
- 1, \ \textit{if } | B(s') | + H_s' = T,
\\
min(s') , \ \textit{if } | B(s') | = 0,
\\
B(s'), \ \textit{if } | B(s') | = 1,
\\
0 , \ \textit{otherwise}.
\end{array}
\right.
\end{eqnarray*}

If a stack is full after we assign the containers during the first phase, then it cannot be selected. If no container was assigned, the index remains as the $min$. If one container was assigned, it is ``artificially'' the new minimum of the stack. Finally, if more than one container was assigned, the index becomes very unattractive by being as low as possible (0).
The second phase is similar to the EM heuristic, but it considers $Gmin$ instead of the $min$ index, breaking ties identically. Note that, after each assignment in the second phase, we update $Gmin$ accordingly for the remaining containers to be assigned. For more details in the complete information case, we refer the reader to \cite{Wu12}.

\subsection{Lower bounds}

After defining heuristics (upper bounds), we are now concerned with defining valid lower bounds for the SCRP. More specifically, we care about computing lower bounds for decision nodes in the decision tree defined before. Note that the computation of lower bounds easily extends to chance nodes.

\subsubsection{Blocking lower bound}

Suppose that the departure order is known, like in the CRP. The following lower bound was introduced by \cite{Kim} and it is based on the following simple observation. If a container is blocking in $n$, then it must be relocated at least once. Thus, the optimal number of relocations is lower bounded by the number of blocking containers.

In the SCRP, the retrieval order is a random variable, so the fact that a container is blocking is also random. Let us denote the expected number of blocking containers in $n$ by $b(n)$. Therefore, by taking the expectation on the retrieval order of the previous fact, which holds for every retrieval order, we have the following observation.

\begin{observation}
For all configurations $n$, $f(n)$ is the minimum expected number of relocations to empty $n$, and $b(n)$ is the expected number of blocking containers, then
\begin{eqnarray*}
f(n) \geqslant b(n).
\end{eqnarray*}
\end{observation}

Lemma \ref{lemma:FormulaBlockingLowerBound} shows one way to compute the expected number of blocking containers for one stack, and $b(n)$ is the sum of the expected number of blocking containers of each stack of $n$. Mathematically, let $b_s(n)$ be the expected number of blocking containers in Stack $s$ of $n$, we have $\displaystyle b(n) = \sum_{s=1}^S b_s(n)$.

\begin{lemma}
\label{lemma:FormulaBlockingLowerBound}
Let $n$ be a single stack configuration with $T$ tiers, and $H \geqslant 0$ containers ($H \leqslant T$). If $H = 0$, we have 
\begin{equation*}
b(n) = 0.
\end{equation*}
If $H \geqslant 1$, we denote the label of containers by $(c_i)_{i=1,\ldots,H}$, where $c_1$ is the container at the bottom and $c_H$ at the top (see Figure \ref{fig:FormulaBlockingLowerBound}), then we have:
\begin{eqnarray*}
\displaystyle b(n) = H - \sum_{h=1}^H \frac{\mathbbm{1} \left\{ \displaystyle c_h = \min_{i=1,\ldots,h} \left\{ c_i \right\} \right\}}{ \displaystyle \sum_{i=1}^h \mathbbm{1} \left\{ c_h = c_i \right\}},
\end{eqnarray*}
where $\mathbbm{1} \left\{ A \right\}$ is the indicator function of $A$.
\end{lemma}

\begin{figure}[h]
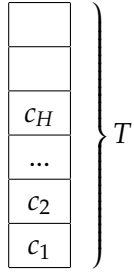

\centering
\begin{eqnarray*}
\left.
\begin{tabular}{|c|c|c|} \hline
\\ \hline
\\ \hline
$c_H$  \\ \hline
... \\ \hline
$c_2$  \\ \hline
$c_1$ \\ \hline
\end{tabular} \hspace{8pt}
\right\} T
\end{eqnarray*}
\caption{Example of a single stack configuration}
\label{fig:FormulaBlockingLowerBound}
\end{figure}

\begin{proof}
Clearly, if $H = 0$, then $b(n) = 0$. If $H \geqslant 1$, then by definition, we have
\begin{eqnarray*}
b(n) = \mathbb{E} \left[ \sum_{h=1}^H \mathbbm{1} \left\{ c_h \textit{ is a blocking container} \right\} \right] = \sum_{h=1}^H \mathbb{P} \left[ c_h \textit{ is a blocking container} \right].
\end{eqnarray*}
Let us fix $h \in \{1,\ldots,H\}$, and compute the probability that $c_h$ is blocking. We consider two cases:
\begin{itemize}
\item If $\displaystyle c_h > \min_{i=1,\ldots,h} \left\{ c_i \right\}$, then $c_h$ is almost surely blocking.
\item Otherwise $ \displaystyle c_h = \min_{i=1,\ldots,h} \left\{ c_i \right\}$, and there are $ \displaystyle \sum_{i=1}^h \mathbbm{1} \left\{ c_h = c_i \right\} - 1$ containers below $c_h$ with the same label (or batch). Since each departure sequence between containers of the same batch is equally likely, the probability that $c_h$ is blocking is equal to $\frac{ \sum_{i=1}^h \mathbbm{1} \left\{ c_h = c_i \right\} - 1}{ \sum_{i=1}^h \mathbbm{1} \left\{ c_h = c_i \right\}} = 1 - \frac{ 1}{ \sum_{i=1}^h \mathbbm{1} \left\{ c_h = c_i \right\}}$.
\end{itemize}
Consequently, we get
\begin{align*}
\mathbb{P} \left[ c_h \textit{ is a blocking container} \right] =& \ 1 \times \mathbbm{1} \left\{ \displaystyle c_h > \min_{i=1,\ldots,h} \left\{ c_i \right\} \right\} + \left( 1 - \frac{ 1}{ \sum_{i=1}^h \mathbbm{1} \left\{ c_h = c_i \right\}} \right) \times \mathbbm{1} \left\{ \displaystyle c_h = \min_{i=1,\ldots,h} \left\{ c_i \right\} \right\}
\\
=& \ 1 - \frac{\mathbbm{1} \left\{ \displaystyle c_h = \min_{i=1,\ldots,h} \left\{ c_i \right\} \right\}}{\sum_{i=1}^h \mathbbm{1} \left\{ c_h = c_i \right\}}.
\end{align*}
We sum the above expression for $h=1,\ldots,H$ to conclude the proof.
\end{proof}

Therefore, one can compute the blocking lower bound as follows: let $H^s$ be the number of containers in Stack $s$, and $\left( c_1^s, \ldots, c_{H^s}^s \right)$ be the containers in Stack $s$ listed from bottom to top, then
\begin{equation}
\label{eq:blockinglowerboundtotal}
b(n) = \sum_{\substack{s=1,\ldots,S \\ H^s \geqslant 1}} \left( H^s - \sum_{h=1}^{H^s}  \frac{\mathbbm{1} \left\{ \displaystyle c^s_h = \min_{i=1,\ldots,h} \left\{ c^s_i \right\} \right\}}{ \displaystyle \sum_{i=1}^h \mathbbm{1} \left\{ c^s_h = c^s_i \right\}} \right).
\end{equation}

\paragraph{Non-uniform case} In the case where probabilities are not uniform across retrieval orders, we still consider a similar lower bound. For each Container $c^s_h$, let $q_{c^s_h}$ be the probability that $c^s_h$ is the first container to be retrieved among the ones with the with same batch, and positioned below in its stack. Equation (\ref{eq:blockinglowerboundtotal}) extends to give:
\[
b(n) = \sum_{\substack{s=1,\ldots,S \\ H^s \geqslant 1}} \left( H^s - \sum_{h=1}^{H^s} q_{c^s_h} \mathbbm{1} \left\{ \displaystyle c^s_h = \min_{i=1,\ldots,h} \left\{ c^s_i \right\} \right\} \right).
\]

\subsubsection{Look-ahead lower bounds}
Note that the blocking lower bound $b$ is only taking into account the current configuration. However, some relocations lead necessarily to an additional relocation. We refer to such relocations as ``bad.'' Formally, let $s$ be a stack of a configuration, and $\min (s)$ be the smallest label of a container in $s$. Recall that, if $s$ is empty, we set $\min (s) = C + 1$. We say that the relocation of Container $c$ from Stack $s$ is ``bad'' if $ \displaystyle c > \max_{s'=1,\ldots,S \ , \ s' \neq s} \left\{ \min (s') \right\}$. We propose to construct a lower bound that anticipates ``bad'' relocations.

The basic idea is based on a similar one used by \cite{Zhu} for the CRP. We consider the $1^{st}$ look-ahead lower bound denoted by $b_1(n)$. By definition, we take $b_1(n) = b(n) + d_1(n)$, where $b(n)$ is the blocking lower bound, and $d_1(n)$ is the expected number of unavoidable ``bad'' relocations while performing the first removal.
We compute the term $d_1(n)$ by considering all realizations of the first target container. For each realization, we compute the number of unavoidable ``bad'' relocations, and average them.
Formally, for a given configuration $n$, consider $U_n$ the set of potential next target container in $n$ (which can be a singleton if it is known already), \textit{i.e.}, $\displaystyle U_n = \left\{ c \left| c = \min_{s=1,\ldots,S} \left( min(s) \right) \right. \right\}$. Based on the definition of a bad relocation, we compute the number of unavoidable ``bad'' moves for each $u \in U_n$ denoted by $\beta(n,u)$, and we take:
\begin{eqnarray*}
d_{1}(n) = \frac{1}{|U_n|} \sum_{u \in U_n} \beta(n,u),
\end{eqnarray*}
or $d_{1}(n) = \sum_{u \in U_n} \ p_{n,u} \beta(n,u)$, where $p_{n,u}$ is the probability that $u$ is the next target container in $n$ if the probabilities considered are not uniform (which can be computed using $\left(p_{n_i}\right)_{n_i \in \Omega_n}$ if $n$ is a chance node).

\begin{figure}[h]
\centering
\begin{tabular}{|c|c|c|} \hline
  &   & 4 \\ \hline
  &  & 3  \\ \hline
 1 & 3 & 1 \\ \hline
\end{tabular}
\caption{Example for look-ahead lower bounds}
\label{fig:exampled1}
\end{figure}

For example, in Figure \ref{fig:exampled1}, the presented configuration denoted by $n$ is such that $b(n) = 2$. Now consider a container $u \in U_n$: if $u$ is the container labeled $1$ in Stack $1$, then there is no blocking container, so $\beta(n, u) = 0$; if $u$ is the other container labeled $1$, the relocation of the container labeled $4$ from Stack $3$ is necessarily a bad relocation, since $min(1) = 1 < 4$ and $min(2) = 3 < 4$, but it is not the case for the blocking container labeled by $3$, hence $\beta(n,u) = 1$. Therefore, $d_1(1) = 0.5 (0+1) = 0.5$, and $b_1(n) = 2 + 0.5 = 2.5$, hence giving a lower bound closer to the optimal solution than $b(n)$.
Note that, if $n$ has an empty stack, then $\beta(n,u) = 0$ for all $u \in U_n$, and hence $d_1(n) = 0$. 
\\

We can refine this idea, by trying to find unavoidable ``bad'' relocations for the second removal. In this case, the configuration depends on the first removal, and the decisions that have been made accordingly. For the sake of clarity, consider that the first target container has been revealed, and denote it $u_1$. After retrieving $u_1$, only containers blocking $u_1$ have changed from their initial position. It can be very challenging to detect future unavoidable ``bad'' moves for these containers. In order to bypass this issue, we consider that all containers blocking $u_1$ are also removed, resulting in a configuration without $u_1$ and its blocking containers. Given this new configuration denoted by $n(u_1)$, we can compute the expected number of unavoidable bad moves $d_1(n(u_1))$. Since $u_1$ is actually random, we have to consider each scenario with their associated probability, and compute a new configuration where blocking containers are retrieved with the target container. We denote the result $d_2(n)$, and it is a lower bound on the expected number of unavoidable bad relocations for the first two removals starting at $n$. Finally, our $2^{nd}$ look-ahead lower bound is given by $b_2(n) = b(n) + d_2(n)$. 

\begin{algorithm}[h]
\caption{Lower bound on the number of unavoidable bad relocations for the k first removals}
\label{algo:lowerbound}
\begin{algorithmic}[1]
\Procedure{[$d_k(n)$] = UnavoidableBadReloc $(n, \ k)$}{}
 \If{$k=0$ or $n$ has an empty stack or $n$ is empty} $d_k(n) = 0$
 \Else\ let $U_n = \left\{ \textit{containers with minimum label in n} \right\}$
 	\If{$k=1$} $d_k(n) = \frac{1}{|U_n|} \sum_{u \in U_n} \beta(n,u)$
 	\Else
 		\For{$u \in U_n$}
 			\State\ Let $n(u)$ be the configuration $n$ without $u$ and all containers blocking $u$
 			\State\ Compute recursively $d_{k-1} \left( n(u) \right)$ = UnavoidableBadReloc $\left( n(u), \ k-1 \right)$
 		\EndFor\
 		\State\ Compute $d_k(n) = \frac{1}{| U_n |} \sum_{u \in U_n} \beta(n,u) + d_{k-1} \left( n(u) \right)$
 	\EndIf\
 \EndIf\
\EndProcedure\
\end{algorithmic}
\end{algorithm}

This idea can easily be generalized for $k \geqslant 2$ by induction with $b_k(n) = b(n) + d_k(n)$. Here $k$ is the number of removals that the lower bound considers to compute the expected number of unavoidable bad relocations (see pseudocode of Algorithm \ref{algo:lowerbound}). We mention that we only use the $1^{st}$ and $2^{nd}$ look ahead lower bounds in our computational experiments. However, note that, as $k$ grows, the computational complexity clearly increases, whereas experiments reveal that the marginal increase of the lower bound, \textit{i.e.}, $b_{k+1}(n) - b_{k}(n) \geqslant 0$, decreases.

\section{PBFS, a new optimal algorithm for the SCRP}
\label{sec:PBFS}

Building upon lower bounds introduced in the previous section, this section introduces, studies and proves the optimality of one of the main contributions of this paper, the $PBFS$ Algorithm.

\subsection{PBFS algorithm}

We start by giving the pseudocode of our algorithm, and we derive its optimality in Lemmas \ref{lemma:lowerBoundPrune} and \ref{lem:optpruning}. PBFS takes two inputs, the configuration $n$ for which we aim to compute $f(n)$, and a valid lower bound $l$. This algorithm uses a combination of four features to return $f(n)$. The first one is the BFS exploration of the tree based on a given lower bound $l$. We first compute $f$ for the ``most promising nodes,'' because nodes with small lower bounds are more likely to result in small $f$. The second technique is stopping to compute $f$ recursively after level $\lambda^* = \max\{S,C_W\}$, by calculating it either using $b$, or the $A^*$ algorithm defined later. The third one is pruning with a lower bound revealing the sub-optimality of some nodes without actually computing $f$. Finally, it also uses the abstraction technique described previously.

\begin{algorithm}[h!]
\caption{\textit{PBFS} Algorithm}
\label{algo:PBFS}
\begin{algorithmic}[1]
\Procedure{[$f(n)$] = $PBFS \left(n, \ l \right)$}{}
 \If{$\lambda_n \leqslant S$ ($n$ has less than $S$ containers)} $f(n) = b(n)$
 \Else
 	\If{ $n$ is a \emph{chance node}} start with $\Psi^{PBFS}_n = \{ \}$
 		\For{ $n_i \in \Omega_n$} $n_i \gets \textproc{Abstract}(n_i)$
 			\If{there exists $m = n_i$ already in $\Psi^{PBFS}_n$} $p^n_m \gets p^n_m + \frac{1}{| \Omega_n |}$
 			\ElsIf{ there exists $m = n_i$ already in the decision tree} add $m$ to $\Psi^{PBFS}_n$ and $p^n_m = \frac{1}{| \Omega_n |}$
 			\Else{ add $n_i$ to $\Psi^{PBFS}_n$, $p^n_{n_i} = \frac{1}{| \Omega_n |}$ and compute $f(n_i) = PBFS \left(n_i, \ l \right)$}
 			\EndIf\
 		\EndFor\
 		\State\ Compute $\displaystyle f(n) = \sum_{n_i \in \Psi^{PBFS}_n} p^n_{n_i} f(n_i)$
 	 \Else{ $n$ is a \emph{decision node}}
 	 	\If{$\lambda_n = C_W$ (the full retrieval order is known)} $f(n) = A^*(n)$
 	 	\Else\ construct $\Delta_n$ by considering all feasible sets of decisions to deliver the target container
 	 		\State\ Compute $l(n_i)$ for each $n_i \in \Delta_n$
 	 		\State\ Sort $\left( n_{(1)},n_{(2)},\ldots,n_{ \left( | \Delta_n | \right) } \right)$ in non-decreasing order of $l(.)$
 	 		\State\ Compute $f(n_{(1)}) = PBFS \left(n_{(1)}, \ l \right)$
 	 		\State\ Start with $\Gamma^{PBFS}_n = \{ n_{(1)} \}$ and $k = 2$
 	 		\While{$k \leqslant | \Delta_n |$ and $\displaystyle l(n_{(k)}) < \min_{j=1,\ldots,k-1} \left\{ f(n_{(j)}) \right\}$} $n_{(k)} \gets \textproc{Abstract}(n_{(k)})$
 	 			\If{ there exists $m = n_{(k)}$ already in the decision tree} add $m$ to $\Gamma^{PBFS}_n$
 	 			\Else{ add $n_{(k)}$ to $\Gamma^{PBFS}_n$ and compute $f(n_{(k)}) = PBFS \left(n_{(k)}, \ l \right)$}
 	 			\EndIf\
 	 			\State\ Update $k = k + 1$
 	 		\EndWhile\
 	 		\State\ $\displaystyle f(n) = r(n) + \min_{n_i \in \Gamma^{PBFS}_n} \left\{ f(n_i) \right\}$
 	 		\EndIf\
 	 \EndIf\
 \EndIf\
\EndProcedure\
\end{algorithmic}
\end{algorithm}

\subsubsection{Decreasing the size of decision tree by increasing $\boldsymbol{\lambda^*}$ to $\boldsymbol{\max\{S,C_W\}}$}

\paragraph{If $\boldsymbol{C_W \leqslant S}$, then compute $\mathbf{f(n)}$ using $\mathbf{b(n)}$}
Recall that, for every relocation, heuristic L chooses the stack with the least number of containers, breaking ties arbitrarily by choosing the leftmost one. Note that $L$ always provides a valid upper bound for the SCRP. So if we denote the resulting expected number of relocations to empty configuration $n$ using $L$ by $f_L(n)$, then we have $f_L(n) \geqslant f(n)$.

\begin{lemma}
\label{lemma:lowerBoundPrune}
Let $n$ be a configuration with $S$ stacks, $T$ tiers, and $C$ containers such that $C \leqslant S$, then we have
\begin{eqnarray*}
f_L(n) = b(n) = f(n)
\end{eqnarray*}
\end{lemma}

\begin{proof}
Consider a retrieval order of containers from $n$ that has a non-zero probability of occurring. If there are no blocking containers, then the lemma clearly holds. Otherwise, let $c$ be one of the blocking containers for this retrieval order, and consider the first removal for which $c$ has to be relocated. For this removal, there are at most $S$ containers in the configuration, hence there exists at least one empty stack to relocate $c$. Since heuristic $L$ chooses always empty columns if one exists, $L$ would move $c$ to one of the existing empty stacks. Note that in this case, $c$ will never be blocking again, and hence never be relocated again. This observation holds for any blocking containers, thus $L$ relocates each blocking container at most once.

Since this fact holds for any retrieval orders with non-zero probability, by taking expectation on the retrieval order, we have $f_L(n) \leqslant b(n)$, thus $f_L(n) \leqslant b(n) \leqslant f(n) \leqslant f_L(n)$, which concludes the proof.
\end{proof}

Lemma \ref{lemma:lowerBoundPrune} states that for configurations with $S$ containers or less, the $L$ heuristic is optimal for the SCRP. This implies that, for nodes at level $S$, we have access to the cost-to-go function using $b(n)$, as well as an optimal solution (provided by heuristic $L$). Hence PBFS can stop branching at $\lambda^* = S$ (line 2 of Algorithm \ref{algo:PBFS}).

\paragraph{If $\boldsymbol{C_W > S}$, then compute $\mathbf{f(n)}$ using the $\boldsymbol{A^*}$ algorithm}
If $n$ is a decision node at level $C_W$, the full order of retrieval is known, and computing $f(n)$ reduces to solving a classical CRP, so we can leverage the existence of efficient solutions to the classical CRP such as the $A^*$ algorithm, and take $\lambda^* = C_W$. Throughout the rest of the paper, $A^*$ refers to the improved version of this algorithm presented in \cite{Borjian15} and we denote the optimal number of relocations obtained by $A^*(n)$ (line 11 of Algorithm \ref{algo:PBFS}).
\\

Combining with the two previous observations, we can take $\lambda^* = \max\{S,C_W\}$.

\subsubsection{Decreasing the size of decision tree by pruning using lower bounds}
We would also like to reduce the size of the tree before level $\lambda^*$. For a decision node $n$, $PBFS$ considers only a subset $\Gamma^{PBFS}_{n}$ of all the offspring $\Delta_n$ (line 21 of Algorithm \ref{algo:PBFS}). Our goal is to set $\Gamma^{PBFS}_{n}$ in order to still guarantee optimality.

\begin{figure}[h]
\begin{center}
\resizebox{13cm}{5.5cm}{
\begin{tikzpicture} [node distance=2cm]
\node (dec0) [decision] {$n$};
\node (dec3) [decisionGood, below of=dec0, xshift=-1.2cm] {$n_{(\left|\Gamma_n^{PBFS}\right|)}$};
\node (dec2)[left of=dec3] {$\ldots$};
\node (dec1) [decisionGood, left of=dec2] {$n_{(1)}$};
\node (dec4) [decisionCross, below of=dec0, xshift=1.2cm] {$n_{\left(\left|\Gamma_n^{PBFS}\right|+1\right)}$};
\node (dec5)[right of=dec4] {$\ldots$};
\node (dec6) [decisionCross, right of=dec5] {$n_{\left(\left|\Delta_n\right|\right)}$};
\node (dec30) [below of=dec3,yshift=1cm] {$l\left(n_{\left(\left|\Gamma_n^{PBFS}\right|\right)}\right)$};
\node (dec20) [below of=dec2,yshift=1cm] {$\leqslant \ldots \leqslant$};
\node (dec10) [below of=dec1,yshift=1cm] {$l\left(n_{(1)}\right)$};
\node (dec40) [below of=dec4,yshift=1cm] {$l\left(n_{\left(\left|\Gamma_n^{PBFS}\right|+1\right)}\right)$};
\node (dec50) [below of=dec5,yshift=1cm] {$\leqslant \ldots \leqslant$};
\node (dec60) [below of=dec6,yshift=1cm] {$l\left(n_{\left(\left|\Delta_n\right|\right)}\right)$};
\node (dec300) [below of=dec30,yshift=1cm] {$f\left(n_{\left(\left|\Gamma_n^{PBFS}\right|\right)}\right)$};
\node (dec200)[below of=dec20,yshift=1cm] {$\ldots$};
\node (dec100) [below of=dec10,yshift=1cm] {$f\left(n_{(1)}\right)$};
\node (dec400) [below of=dec40,yshift=1cm] {$\times$};
\node (dec500)[below of=dec50,yshift=1cm] {$\ldots$};
\node (dec600) [below of=dec60,yshift=1cm] {$\times$};
\draw [arrow] (dec0) -- (dec1);
\draw [arrow] (dec0) -- (dec3);
\draw [arrow] (dec0) -- (dec4);
\draw [arrow] (dec0) -- (dec6);
\end{tikzpicture}}
\caption{Illustration of the pruning rule. First, offspring are ordered by non-decreasing lower bounds. Then we start computing the objective function starting at $n_{(1)}$. We stop computing the objective functions once the pruning rule is reached. In the figure above, green nodes are nodes in $\Gamma_n^{PBFS}$ \textit{i.e.} $f(.)$ has been computed. Orange nodes are nodes in $\Delta_n \setminus \Gamma_n^{PBFS}$ \textit{i.e.} $f(.)$ does not need to be computed which is represented here by $\times$.}
\label{fig:pruning}
\end{center}
\end{figure}
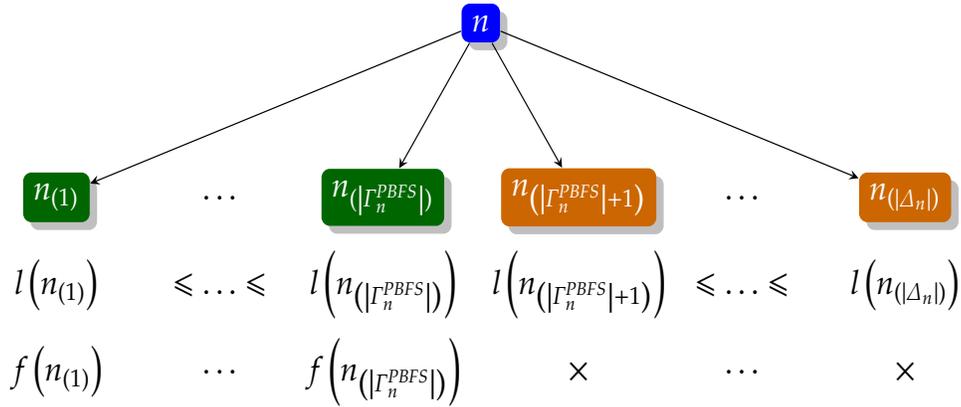

First, PBFS generates all nodes $n_i \in \Delta_n$ by considering all feasible sets of decisions to deliver the target container in $n$ (line 12 of Algorithm \ref{algo:PBFS}), and for each of them, compute a lower bound $l(n_i)$, where $l$ is the input lower bound (line 13 of Algorithm \ref{algo:PBFS}). Let $\left( n_{(1)},n_{(2)},\ldots,n_{ \left( | \Omega_n | \right) } \right)$ be the list of offspring of $n$ sorted by non-decreasing lower bound (line 14 of Algorithm \ref{algo:PBFS}). The algorithm considers first $n_{(1)}$, adds it to $\Gamma^{PBFS}_n$ and computes $f(n_{(1)})$ recursively (line 15-16 of Algorithm \ref{algo:PBFS}). Then for $k \geqslant 2$, we consider $n_{(k)}$'s sequentially, and check if $\displaystyle l(n_{(k)}) < \min_{j=1,\ldots,k-1} \left\{f(n_{(j)}) \right\}$ (line 17 of Algorithm \ref{algo:PBFS}). If so, add $n_{(k)}$ to $\Gamma^{PBFS}_n$ and compute $f(n_{(k)})$ recursively. If not, we stop branching on all nodes $n_{(k)},\ldots,n_{ \left( | \Omega_n | \right) }$. An illustration of the pruning rule is shown in Figure \ref{fig:pruning} and the next lemma shows the optimality of this rule.

\begin{lemma}
\label{lem:optpruning}
	Let $n$ be a decision node in the decision tree, and $\Gamma^{PBFS}_{n}$ be the subset of nodes considered for this node in Algorithm \ref{algo:PBFS}, and constructed as aforementioned, then we have
\begin{eqnarray*}
	\min_{m_i \in \Gamma^{PBFS}_{n}} \left\{f(m_i) \right\} = \min_{n_i \in \Delta_{n}} \left\{f(n_i) \right\}.
\end{eqnarray*}
\end{lemma}

\begin{proof}
Let $\left( n_{(1)},n_{(2)},\ldots,n_{ \left( | \Omega_n | \right) } \right)$ be the list of offspring of $n$, sorted by non-decreasing lower bounds. We consider two cases. 
\begin{itemize}
\item If $\Gamma^{PBFS}_n = \Delta_n$, the statement clearly holds.
\item Otherwise, there exists $k \leqslant | \Delta_n |$ such that $\displaystyle l(n_{(k)}) \geqslant \min_{j=1,\ldots,k-1} \left\{ f(n_{(j)}) \right\}$, and $\Gamma^{PBFS}_n = \left\{ n_{(1)}, \ldots, n_{(k-1)} \right\}$. Note that we have $\forall k' \geqslant k$, $f(n_{(k')}) \geqslant \displaystyle l(n_{(k')}) \geqslant l(n_{(k)}) \geqslant \min_{j=1,\ldots,k-1} \left\{f(n_{(j)}) \right\}$. Hence $\displaystyle \min_{n_i \in \Delta_{n}} \left\{f(n_i) \right\} = \min_{j=1,\ldots,k-1} \left\{f(n_{(j)}) \right\} = \min_{m_i \in \Gamma^{PBFS}_{n}} \left\{f(m_i) \right\}$.
\end{itemize}

\end{proof}

We claim that increasing $\lambda^*$ to $\max\{S,C_W\}$ together with pruning in a Best-First-Search scheme, dramatically help in the efficiency of PBFS while keeping the guarantee of optimality. In the case of small batches, the PBFS algorithm appears to be efficient (see Section \ref{sec::ComputationExperiments}). However, this algorithm faces the issue that $|\Omega_n| = C_w !$ if $n$ is a chance node. So if $\mathbf{C_w}$ \textbf{is large, typically} $\mathbf{C_w \geqslant 4}$, the number of nodes to consider gets too large. We tackle this issue by considering a near-optimal algorithm in the next section.

As a final remark, batches should be as small as possible if information is only at stake. Indeed, smaller batches correspond to an efficient information system since more information is known about the retrieval order. But the size of batches is restricted by two intrinsic constraints:
\begin{enumerate}
\item \textbf{Batches should be at least larger than a certain size.} Indeed, a terminal offers time slots for trucks to register, and these slots cannot be too small (in terms of time), as trucks would most certainly not arrive during their appointed slot due to traffic or other uncertain factors. Therefore, given the minimum time of a slot, the terminal will allow at least a certain number of trucks to register for each slot, \textit{i.e.}, the minimum batch size.
\item \textbf{Batches cannot be too large} in order to have the batch model applicable, since in this model, the appointment time windows are supposed to be the same as or shorter than the target waiting time. As the target waiting time is limited, there is a limited number of containers that can be retrieved in a certain batch.
\end{enumerate}
This leads us to consider an alternative to $PBFS$ (see Section \ref{sec:PBFSA}) in the case of larger batches.

\section{PBFSA, near-optimal algorithm with guarantees for large batches}
\label{sec:PBFSA}

\begin{algorithm}[h!]
\caption{\textit{PBFSA} Algorithm}
\label{algo:PBFSA}
\begin{algorithmic}[1]
\Procedure{[$\tilde{f}(n)$] = $PBFSA \left(n, \ l , \ \epsilon \right)$}{}
 \If{$\lambda_n \leqslant S$} $\tilde{f}(n) = b(n)$
 \Else
 	\If{ $n$ is a \emph{chance node}} start with $\Psi^{PBFSA}_n = \{ \}$. Let $w_{min}$ be such that $\lambda_n = C - K_{w_{min}} + 1$
 	  \State\ Compute $ \delta_n = \min \left\{ w \in \left\{w_{min},\ldots,W \right\} \left| \sum_{u=w_{min}}^w C_{u} \geqslant \lambda_n - \lambda^* \right. \right\}$ to get $\displaystyle \epsilon_n = \frac{\epsilon}{\delta_n}$
 	  \State\ Compute $f_{max}(n)$ and $f_{min}(n)$ to get $\displaystyle N_n(\epsilon_n) = \frac{\pi \left( f_{max}(n) - f_{min}(n) \right)^2}{2\epsilon_n^2}$
 	  \If{$N_n(\epsilon_n) \leqslant C_{w_{min}}!$}
 	    \For{ $i = 1,\ldots,N_n(\epsilon_n)$} 
 		  \State\ Sample a random permutation, get corresponding $n_i \in \Omega_n$ and $n_i \gets \textproc{Abstract}(n_i)$
 		  \If{ there is $m = n_i$ already in $\Psi^{PBFSA}_n$} $\ p^n_{m} \gets  p^n_{m} + \frac{1}{N_n(\epsilon_n)}$
 		  \ElsIf{ there is $m = n_i$ already in the decision tree} add $m$ to $\Psi^{PBFSA}_n$, $p^n_m = \frac{1}{N_n(\epsilon_n)}$
 		  \Else{ add $n_i$ to $\Psi^{PBFSA}_n$, $p^n_{n_i} = \frac{1}{N_n(\epsilon_n)}$ and compute $\tilde{f}(n_i) = PBFSA \left(n_i, \ l, \ \epsilon - \epsilon_n \right)$}
 		  \EndIf\
 		\EndFor\
 	  \Else
 	    \For{ $n_i \in \Omega_n$} $n_i \gets \textproc{Abstract}(n_i)$
 			\If{there exists $m = n_i$ already in $\Psi^{PBFSA}_n$} $p^n_m \gets p^n_m + \frac{1}{| \Omega_n |}$
 			\ElsIf{there exists $m = n_i$ already in decision tree} add $m$ to $\Psi^{PBFSA}_n$ and $p^n_m = \frac{1}{| \Omega_n |}$
 			\Else{ add $n_i$ to $\Psi^{PBFSA}_n$, $p^n_{n_i} = \frac{1}{| \Omega_n |}$ and compute $\tilde{f}(n_i) = PBFSA \left(n_i, \ l, \ \epsilon - \epsilon_n \right)$}
 			\EndIf\
 		\EndFor\
 	  \EndIf\
 	  \State\ Compute $\displaystyle \tilde{f}(n) = \sum_{n_i \in \Psi^{PBFSA}_n} p^n_{n_i} \tilde{f}(n_i)$
    \Else{ $n$ is a \emph{decision node}}
    	  \If{$\lambda_n \leqslant C_W$} $\tilde{f}(n) = A^*(n)$
 	 	\Else\ Construct $\Delta_n$ by considering all feasible sets of decisions to deliver the target container
 	 		\State\ Compute $l(n_i)$ for each $n_i \in \Delta_n$
 	 		\State\ Sort $\left( n_{(1)},n_{(2)},\ldots,n_{ \left( | \Delta_n | \right) } \right)$ in non-decreasing order of $l(.)$
 	 		\State\ Compute $\tilde{f}(n_{(1)}) = PBFSA \left(n_{(1)}, \ l, \ \epsilon \right)$
 	 		\State\ Start with $\Gamma^{PBFSA}_n = \{ n_{(1)} \}$ and $k = 2$
 	 		\While{$k \leqslant | \Delta_n |$ and $\displaystyle l(n_{(k)}) < \min_{j=1,\ldots,k-1} \left\{ \tilde{f}(n_{(j)}) \right\}$} $n_{(k)} \gets \textproc{Abstract}(n_{(k)})$
 	 			\If{ there exists $m = n_{(k)}$ already in the decision tree} add $m$ to $\Gamma^{PBFSA}_n$
 	 			\Else{ add $n_{(k)}$ to $\Gamma^{PBFSA}_n$ and compute $\tilde{f}(n_{(k)}) = PBFSA \left(n_{(k)}, \ l, \ \epsilon \right)$}
 	 			\EndIf\
 	 			\State\ Update $k = k + 1$
 	 		\EndWhile\
 	 		\State\ $\displaystyle \tilde{f}(n) = r(n) + \min_{n_i \in \Gamma^{PBFSA}_n} \left\{ \tilde{f}(n_i) \right\}$
 	 	\EndIf\
    \EndIf\
  \EndIf\
\EndProcedure\
\end{algorithmic}
\end{algorithm}

Building upon $PBFS$ introduced in the previous section, this section describes the randomized algorithm $PBFSA$ and shows some theoretical guarantees on expectation. This new algorithm is identical to $PBFS$ except when computing the value function of a chance node (lines 4 to 17 of Algorithm \ref{algo:PBFSA}). In order to decrease the number of decision offspring to consider for each chance node, we sample a certain number of \textit{i.i.d.} permutations, and only consider the decision nodes associated with these permutations as illustrated in Figure \ref{fig:sampling}. The facts that the objective function is bounded and that our problem has a finite number of sampling stages allow us to independently sample nodes in order to approximate the objective function. Using concentration inequalities, we can chose the number of samples needed to control the approximation error.

\begin{figure}[h]
\begin{center}
\resizebox{15cm}{6.5cm}{
\begin{tikzpicture} [node distance=2cm]
\node (dec01) [above of=dec0,yshift=-0.5cm] {\begin{tabular}{|c|c|c|} \hline
   &  & 1 \\ \hline
   & 5 & 4 \\ \hline
 1 & 5 & 1 \\ \hline
\end{tabular}};
\node (dec0) [chance] {$n$};
\node (dec4)[below of=dec0] {$\ldots$};
\node (dec3) [decisionGood, left of=dec4] {$n_{N_n\left( \epsilon_n \right)}$};
\node (dec30) [below of=dec3,yshift=1cm] {$f\left( n_{N_n\left( \epsilon_n \right)} \right)$};
\node (dec2)[left of=dec3] {$\ldots$};
\node (dec1) [decisionGood, left of=dec2] {$n_{1}$};
\node (dec10) [below of=dec1,yshift=1cm] {$f\left( n_{1} \right)$};
\node (dec5) [decisionCross, right of=dec4] {$n_{i}$};
\node (dec50) [below of=dec5,yshift=1cm] {$\times$};
\node (dec6)[right of=dec5] {$\ldots$};
\node (dec7)[decisionCross,right of=dec6] {$n_{\left|\Omega_n\right|}$};
\node (dec70) [below of=dec7,yshift=1cm] {$\times$};
\draw [arrow] (dec0) -- (dec1);
\draw [arrow] (dec0) -- (dec3);
\draw [arrow] (dec0) -- (dec5);
\draw [arrow] (dec0) -- (dec7);
\end{tikzpicture}}
\caption{Illustration of the sampling rule. In this figure, the smallest batch is batch 1 therefore $w_{min} = 1$, and there are 6 containers thus $\lambda_n = 6$. These values allow us to compute the number of samples required $N_n(\epsilon_n)$. If $N_n(\epsilon_n)$ in less than the total number of offspring $\left|\Omega_n\right| = C_{w_{min}}! = 3!$, then we only compute $f(.)$ for sampled nodes. $\Psi_n^{PBFSA}$ represents the subset of sampled nodes collored green and for which $f(.)$ needs to be computed. Note that $\left|\Psi_n^{PBFSA}\right| = N_n(\epsilon_n)$. Orange nodes are nodes in $\Omega_n \setminus \Psi_n^{PBFSA}$ \textit{i.e.} there were not sampled and $f(.)$ does not need to be computed which is represented here by $\times$. Finally, the approximate value of $f(n)$ is the average of the objective values over all sampled nodes.}
\label{fig:sampling}
\end{center}
\end{figure}
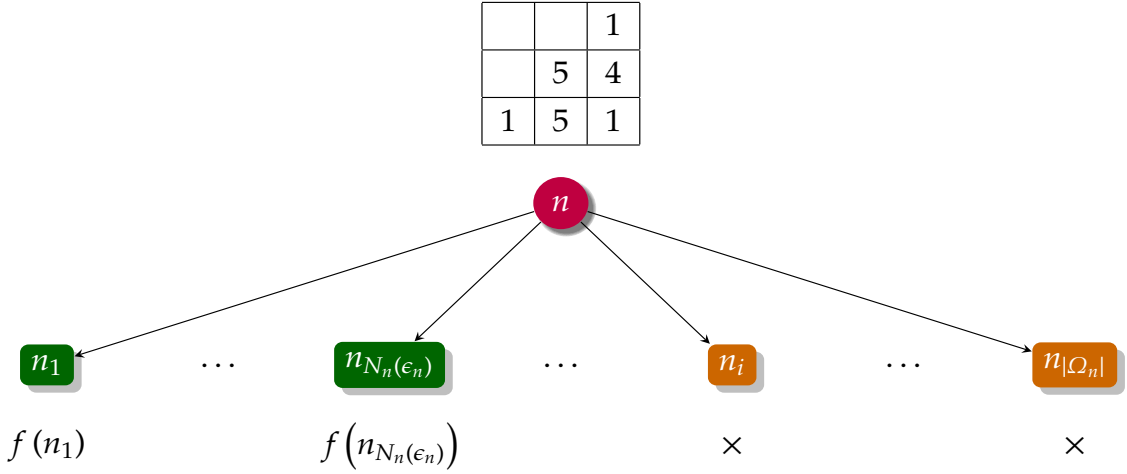

Formally, let $n$ be a given chance node, recall that we compute $\displaystyle f(n) = \frac{1}{| \Omega_n |} \sum_{n_i \in \Omega_n} f(n_i)$, where each $n_i \in \Omega_n$ represents one retrieval order (a random permutation) of batch $w$ (if $\lambda_n = C - K_w + 1$). Let $\Psi_n \subset \Omega_n$ be defined as the resulting subset of the $N_n(\epsilon_n) \left( \in \mathbb{N} \right)$ offspring drawn \textit{i.i.d.}, where $N_n$ is a function of $\epsilon_n$, itself a function of $n$ and $\epsilon > 0$ (a target error). The goal is to define $N_n(\epsilon_n)$, such that $\displaystyle \tilde{f}(n) = \frac{1}{| \Psi_n |} \sum_{m \in \Psi_n} f(m)$ is a ``good'' approximation of $f(n)$, \textit{i.e.}, $\left| \tilde{f}(n) - f(n) \right|$ is bounded by $\epsilon$ on average.

$PBFSA$ takes three input arguments, the configuration $n$ for which we want to evaluate $f$, a valid lower-bound $l$ and an upper bound $\epsilon > 0$ on the total expected ``error'' ensured by the algorithm.
It outputs $\tilde{f}(n)$, which is a randomized approximation of $f(n)$. Because of the samplings performed in line 9 in Algorithm \ref{algo:PBFSA}, the output of $PBFSA$ is random. The average error incurred by the algorithm is 
$\mathbb{E} \left[  \left| \tilde{f}(n) - f(n) \right| \right]$, where the expectation is taken over the aforementioned samplings. Our main result (Lemma \ref{lemma:mainResult}) states that $PBFSA$ ensures $\mathbb{E} \left[  \left| \tilde{f}(n) - f(n) \right| \right] \leqslant \epsilon$, in other words, $PBFSA$ guarantees an average error of at most $\epsilon$.

\begin{lemma}
\label{lemma:mainResult}
Let $n$ be a configuration with $\lambda_n \geqslant 0$ containers, $l$ be a valid lower bound function, and $\epsilon > 0$. If $\tilde{f}(n) = PBFSA(n, l, \epsilon)$, then
\begin{equation*}
\mathbb{E} \left[  \left| \tilde{f}(n) - f(n) \right| \right] \leqslant \epsilon.
\end{equation*}
\end{lemma}

\subsection{Hoeffding`s inequality applied to the SCRP}

In order to prove this result, we use Hoeffding`s inequality to compute the number of samples to ensure probabilistic guarantees. We first state the well-known inequality, and a direct corollary.

\begin{theorem}[Hoeffding`s inequality]
Let $X \in \left[ x_{min}, x_{max} \right] $ be a real-valued bounded random variable with mean value $\mathbb{E} \left[ X \right]$. Let $N \in \mathbb{N}$ and $\left( X_1,\ldots, X_{N} \right)$ be $N$ \textit{i.i.d.} samples of $X$. If $\displaystyle \overline{X} = \frac{1}{N} \sum_{i=1}^{N} X_i$, then we have
\begin{equation}
\label{eq:Hoeff1}
\forall \delta > 0 \ , \ \mathbb{P} \left( \overline{X} - \mathbb{E} \left[ X \right] > \delta \right) \ \leqslant\ exp \left( \frac{-2 N \delta^2}{\left(x_{max} - x_{min} \right)^2} \right),
\end{equation}
and
\begin{equation}
\label{eq:Hoeff2}
\forall \delta > 0 \ , \ \mathbb{P} \left( \overline{X} - \mathbb{E} \left[ X \right] < - \delta \right) \ \leqslant\ exp \left( \frac{-2 N \delta^2}{\left(x_{max} - x_{min} \right)^2} \right).
\end{equation}
\end{theorem}

\begin{corollary}
\label{corr:Hoeffding}
Let $X \in \left[ x_{min}, x_{max} \right] $ be a real-valued bounded random variable with mean value $\mathbb{E} \left[ X \right]$. Let $N \in \mathbb{N}$ and $\left( X_1,\ldots, X_{N} \right)$ be $N$ \textit{i.i.d.} samples of $X$. If $\displaystyle \overline{X} = \frac{1}{N} \sum_{i=1}^{N} X_i$, then $\forall \epsilon > 0$ such that $\displaystyle N \geqslant \frac{\pi \left( x_{max} - x_{min} \right)^2}{2\epsilon^2}$, we have 
\begin{equation}
\label{eq:Hoeffdingcor1}
\mathbb{E} \left[ \left( \overline{X} - \mathbb{E} \left[ X \right] \right)^+ \right] \leqslant \frac{\epsilon}{2},
\end{equation}
\begin{equation}
\label{eq:Hoeffdingcor2}
\mathbb{E} \left[ \left( \overline{X} - \mathbb{E} \left[ X \right] \right)^- \right] \leqslant \frac{\epsilon}{2},
\end{equation}
where $x^+ = \max\{x,0\}$ (resp. $x^- = -\min\{x,0\}$) is the positive (resp. negative) part of $x$.
\end{corollary}

\begin{proof}
For the first result, define $\Delta = \left( \overline{X} - \mathbb{E} \left[ X \right] \right)^+ = \left( \overline{X} - \mathbb{E} \left[ X \right] \right) \mathbbm{1} \left\{ \overline{X} - \mathbb{E} \left[ X \right] > 0 \right\}$. Note that $\Delta$ is a non-negative random variable, and $\forall \delta > 0$, $\left\{ \Delta > \delta \right\} = \left\{ \overline{X} - \mathbb{E} \left[ X \right] > \delta \right\}$. Let $F_{\Delta}$ denote the cumulative distribution function of $\Delta$, thus, using Equation (\ref{eq:Hoeff1}), $1 - F_{\Delta}(\delta) = \mathbb{P} \left( \Delta > \delta \right) = \mathbb{P} \left( \overline{X} - \mathbb{E} \left[ X \right] > \delta \right) \leqslant  exp \left( \frac{-2 N \delta^2}{\left(x_{max} - x_{min} \right)^2} \right)$, which gives
\begin{align*}
\mathbb{E} \left[ \Delta \right] = & \ \int_{\delta=0}^{\infty} \left( 1 - F_{\Delta}(\delta) \right) d\delta
 \leqslant \int_{\delta=0}^{\infty} exp \left( \frac{-2 N \delta^2}{\left(x_{max} - x_{min} \right)^2} \right) d\delta \\
 & \ = \frac{\left(x_{max} - x_{min} \right)}{\sqrt{2N}} \int_{u=0}^{\infty} exp(-u^2)du
 = \frac{\sqrt{\pi} \left(x_{max} - x_{min} \right)}{2\sqrt{2N}} \leqslant \frac{\epsilon}{2}.
\end{align*}

The proof of the second result is identical to the first one if we consider $\Delta' = \left( \overline{X} - \mathbb{E} \left[ X \right] \right)^- = \left( \mathbb{E} \left[ X \right] - \overline{X} \right) \mathbbm{1} \left\{ \mathbb{E} \left[ X \right] - \overline{X} > 0 \right\}$ and notice that $\forall \delta > 0$, $\left\{ \Delta' > \delta \right\} = \left\{ \overline{X} - \mathbb{E} \left[ X \right] < - \delta \right\}$, hence $1 - F_{\Delta'}(\delta) \leqslant exp \left( \frac{-2 N \delta^2}{\left(x_{max} - x_{min} \right)^2} \right)$ using Equation (\ref{eq:Hoeff2}).
\end{proof}

\paragraph{Computing $\boldsymbol{f_{min}}$ and $\boldsymbol{f_{max}}$}
In order to use Hoeffding`s inequality, we need to define lower ($f_{min}$) and upper ($f_{max}$) bound functions, such that for each chance node $n$, $\displaystyle f_{min}(n) \leqslant \min_{n_i \in \Omega_n} \left\{ f(n_i) \right\}$ and $\displaystyle f_{max}(n) \geqslant \max_{n_i \in \Omega_n} \left\{ f(n_i) \right\}$.
\begin{lemma}
Let $n$ be a chance node, if
\begin{equation}
f_{min}(n) = \min_{n_i \in \Omega_n} \left\{ b(n_i) \right\},
\end{equation}
and
\begin{equation}
f_{max}(n) = \min \left\{ \left(\left( \lambda_n - S \right)(T-1)\right)^+ + \left(\min \left\{ S, \lambda_n \right\}-1\right) \ , \ \left( 2 \left\lceil\frac{\lambda_n}{S}\right\rceil - 1 \right)\max_{n_i \in \Omega_n} \left\{ b(n_i) \right\} \right\},
\end{equation}
then
\begin{equation*}
f_{min}(n) \leqslant \min_{n_i \in \Omega_n} \left\{ f(n_i) \right\} \textit{ and } f_{max}(n) \geqslant \max_{n_i \in \Omega_n} \left\{ f(n_i) \right\}.
\end{equation*}
\end{lemma}

\begin{proof}
Since $b(n_i) \leqslant f(n_i)$, then we have $\displaystyle f_{min}(n) = \min_{n_i \in \Omega_n} \left\{ b(n_i) \right\} \leqslant \min_{n_i \in \Omega_n} \left\{ f(n_i) \right\}$.

By definition, $f_{max}(n)$ is the minimum of two valid upper bounds. The first one comes from a basic observation. If there are $\lambda_n$ containers remaining to be retrieved in $n$, consider two cases:
\begin{itemize}
\item If $\lambda_n > S$, take the r$^{th}$ retrieval. If $S < r \leqslant \lambda_n$, then in order to perform this retrieval, there are at most $T-1$ containers blocking the target container so at most $T-1$ relocations are needed. When there are $S$ or less containers remaining, each container (except the lowest one) is at most relocated once hence we need at most $S - 1$ relocations. Combining these two facts, the maximum number of relocations is bounded by $\left( \lambda_n - S \right)(T-1) + (S-1) = \left(\left( \lambda_n - S \right)(T-1)\right)^+ + \left(\min \left\{ S, \lambda_n \right\}-1\right)$.
\item If $\lambda_n \leqslant S$, we know that $f(n_i) = b(n_i)\leqslant \lambda_n - 1 = \left(\left( \lambda_n - S \right)(T-1)\right)^+ + \left(\min \left\{ S, \lambda_n \right\}-1\right)$.
\end{itemize}
This shows the validity of the first upper bound.

For the second upper bound, \cite{Zehendner16} prove that, in the online case with a unique batch, the number of relocations performed by the leveling heuristic ($L$) is at most $\left( 2 \left\lceil\frac{\lambda_n}{S}\right\rceil - 1 \right) B$, where $B$ is the number of blocking containers. Since $L$ is not using any information about batches (only the height of stacks), this result holds for both batch and online models with any number of batches.
Let $n_i \in \Omega_n$, using this result and taking expectation over the retrieval order of containers not unveiled in $n_i$ yet, we have $f(n_i) \leqslant f_L(n_i) \leqslant \left( 2 \left\lceil\frac{\lambda_n}{S}\right\rceil - 1 \right) b(n_i)$. By taking the maximum over all $n_i \in \Omega_n$, the latter inequality results in the second upper bound.
\end{proof}

Notice that the previous lemma involves computing $\displaystyle \min_{n_i \in \Omega_n} \left\{ b(n_i) \right\}$ and $\displaystyle \max_{n_i \in \Omega_n} \left\{ b(n_i) \right\}$. The following corollary provides an efficient way of computing these values.
\begin{lemma}
\label{lemma:minbmaxb}
Let $n$ be a chance node, and $w_{min} \in \{1,\ldots,W\}$ be such that $\displaystyle \lambda_n = C - K_{w_{min}} + 1$ (\textit{i.e.}, the minimum batch in $n$). For each Stack $s$ of $n$ with $H^s \geqslant 1$ containers, let $(c^s_h)_{h=1,\ldots,H^s}$ be the containers in $s$, where $c^s_1$ is the container at the bottom and $c^s_{H^s}$ at the top (see Figure \ref{fig:FormulaBlockingLowerBound}, for the case $H=H^s$). Finally, consider $C^s_{w_{min}} = \left| \left\{ c^s_h = K_{w_{min}}, \ h=1,\ldots,H^s \right\} \right|$. Then we have
\begin{eqnarray}
\displaystyle \min_{n_i \in \Omega_n} \left\{ b(n_i) \right\} = \sum_{\substack{s=1,\ldots,S \\ H^s \geqslant 1}} \left( H^s - C^s_{w_{min}} - \sum_{\substack{h=1,\ldots,H^s \\ c^s_h \neq K_{w_{min}}}}  \frac{\mathbbm{1} \left\{ \displaystyle c^s_h = \min_{i=1,\ldots,h} \left\{ c^s_i \right\} \right\}}{ \displaystyle \sum_{i=1}^h \mathbbm{1} \left\{ c^s_h = c^s_i \right\}} \right),
\end{eqnarray}
and
\begin{eqnarray}
\displaystyle \max_{n_i \in \Omega_n} \left\{ b(n_i) \right\} = \sum_{\substack{s=1,\ldots,S \\ H^s \geqslant 1}} \left( H^s - \sum_{\substack{h=1,\ldots,H^s \\ c^s_h \neq K_{w_{min}}}}  \frac{\mathbbm{1} \left\{ \displaystyle c^s_h = \min_{i=1,\ldots,h} \left\{ c^s_i \right\} \right\}}{ \displaystyle \sum_{i=1}^h \mathbbm{1} \left\{ c^s_h = c^s_i \right\}} \right).
\end{eqnarray}
\end{lemma}

\begin{proof}
Let $n$ be a chance node, and $n_i \in \Omega_n$ be one of its decision offspring, such that all containers in batch $w_{min}$ have been revealed. Recall that $\displaystyle b(n_i) = \sum_{s=1,\ldots,S} b^s(n_i)$, where $b^s(n_i)$ is the expected number of blocking containers in Stack $s$. First, for each Stack $s$ such that $H^s = 0$, $b^s(n_i) = 0$. Hence $\displaystyle b(n_i) = \sum_{\substack{s=1,\ldots,S \\ H^s \geqslant 1}} b^s(n_i)$.

For each Stack $s$ such that $H^s \geqslant 1$, consider the containers in this stack $(c^s_i)_{i=1,\ldots,H^s}$. Since all containers labeled $K_{w_{min}}$, \textit{i.e.}, from batch $w_{min}$, are known in $n_i$, we can write
\[ b^s(n_i) = \sum_{h=1,\ldots,H^s} \mathbb{P} \left[ c^s_h \textit{ is blocking in } n_i \right] = \sum_{\substack{h=1,\ldots,H^s \\ c^s_h = K_{w_{min}}}} \mathbbm{1} \left\{ c^s_h \textit{ is blocking in } n_i \right\} \ + \sum_{\substack{h=1,\ldots,H^s \\ c^s_h \neq K_{w_{min}}}} \mathbb{P} \left[ c^s_h \textit{ is blocking in } n_i \right].\]

Fix $h \in \{1,\ldots,H^s\}$ and $c^s_h \neq K_{w_{min}}$, then the proof of Lemma \ref{lemma:FormulaBlockingLowerBound} uses the fact that $\mathbb{P} \left[ c^s_h \textit{ is blocking in } n_i \right] = 1 - \frac{\mathbbm{1} \left\{ \displaystyle c^s_h = \min_{i=1,\ldots,h} \left\{ c^s_i \right\} \right\}}{\sum_{i=1}^h \mathbbm{1} \left\{ c^s_h = c^s_i \right\}}$. Finally, it is clear that $\displaystyle 0 \leqslant \sum_{\substack{h=1,\ldots,H^s \\ c^s_h = K_{w_{min}}}} \mathbbm{1} \left\{ c^s_h \textit{ is blocking in } n_i \right\} \leqslant C^s_{w_{min}}$. Therefore, we can get the corresponding formulae.
As a final remark, note that each of these bounds are tight. Indeed, consider the offspring of $n$, in which all containers in batch $w_{min}$ are in the decreasing (resp. increasing) order of retrieval from top to bottom, then this offspring has no (resp. $C^s_{w_{min}}$) blocking container(s). 
\end{proof}

\paragraph{Non-uniform case} Similar to the blocking lower bound, we can extend Lemma \ref{lemma:minbmaxb} to the case where probabilities are not uniform across retrieval orders. Recall that $q_{c^s_h}$ denotes the probability that $c^s_h$ is the first one to be retrieved among the ones positioned below in its stack and with the same label. Then we have
\[
\min_{n_i \in \Omega_n} \left\{ b(n_i) \right\} = \sum_{\substack{s=1,\ldots,S \\ H^s \geqslant 1}} \left( H^s - C^s_{w_{min}} - \sum_{\substack{h=1,\ldots,H^s \\ c^s_h \neq K_{w_{min}}}} q_{c^s_h} \mathbbm{1} \left\{ \displaystyle c^s_h = \min_{i=1,\ldots,h} \left\{ c^s_i \right\} \right\} \right),
\]
and
\[
\max_{n_i \in \Omega_n} \left\{ b(n_i) \right\} = \sum_{\substack{s=1,\ldots,S \\ H^s \geqslant 1}} \left( H^s - \sum_{\substack{h=1,\ldots,H^s \\ c^s_h \neq K_{w_{min}}}} q_{c^s_h} \mathbbm{1} \left\{ \displaystyle c^s_h = \min_{i=1,\ldots,h} \left\{ c^s_i \right\} \right\} \right).
\]

Now we can prove Lemma \ref{lemma:mainResult}.
\begin{proof}[Proof of Lemma \ref{lemma:mainResult}]
The proof is by induction on $\lambda_n$. Throughout the proof, we use the same notations as the ones introduced in Algorithm \ref{algo:PBFSA}. We say that $\tilde{f}$ verifies Conditions $(A)$ and $(B)$ at node $n$, if it verifies respectively the first and second inequalities below:
\begin{equation*}
\mathbb{E} \left[ \left( \tilde{f}(n) - f(n) \right)^+ \right] \leqslant \frac{\epsilon}{2} \text{ and } \mathbb{E} \left[ \left( \tilde{f}(n) - f(n) \right)^- \right] \leqslant \frac{\epsilon}{2}.
\end{equation*}
Note that if $\tilde{f}$ verifies Conditions $(A)$ and $(B)$ at node $n$, then $\mathbb{E} \left[ \left| \tilde{f}(n) - f(n) \right| \right] \leqslant \epsilon$, which would prove the lemma. Given $\epsilon > 0$, and $l$ a valid lower bound, the induction hypothesis is:
\[
\text{If } \tilde{f}(n) = PBFSA(n, l, \epsilon) \text{, then } \tilde{f} \text{ verifies Conditions } (A) \text{ and } (B) \text{ at node } n.
\]

First, if $\lambda_n \leqslant S$, then $\tilde{f}(n) = b(n) = f(n)$, and therefore, $\tilde{f}$ verifies Conditions $(A)$ and $(B)$ at node $n$. In this case, $\tilde{f}(n)$ is actually deterministic since no sampling is performed by $PBFSA$.

From now on, consider $n$ such that $\lambda_n > S$.

\textbf{If $\boldsymbol{n}$ is a decision node} such that $\tilde{f}(n) = PBFSA(n, l, \epsilon)$ and $\lambda_n > S$. First, if $S < \lambda_n \leqslant C_W$, then $\tilde{f}(n) = A^*(n) = f(n)$, hence $\tilde{f}$ verifies Conditions $(A)$ and $(B)$ at $n$.

If $\lambda_n > \max \{ S,C_W \}$, consider $\tilde{n} = \underset{n_i \in \Gamma^{PBFSA}_n}{argmin} \left\{ \tilde{f}(n_i) \right\}$ and $n^* = \underset{n_i \in \Delta_n}{argmin} \left\{ f(n_i) \right\}$. Note that $\tilde{f}(n) - f(n) = \tilde{f}(\tilde{n}) - f(n^*)$ almost surely (\textit{a.s.}), and by definition, $\tilde{n}$ and $n^*$ are both such that $\lambda_{\tilde{n}} = \lambda_{n^*} = \lambda_n - 1 < \lambda_n$. Consider the following measurable event:
\begin{equation}
\mathcal{E} = \left\{ \tilde{f}(n) - f(n) = \tilde{f}(\tilde{n}) - f(n^*) > 0 \right\}.
\end{equation}
\begin{itemize}
\item Conditioned on $\mathcal{E}$, we have $\left( \tilde{f}(n) - f(n) \right)^- = 0$ \textit{a.s.}, thus
\begin{equation}
\label{eq:mathcalE1}
\mathbb{E} \left[ \left. \left( \tilde{f}(n) - f(n) \right)^- \ \right| \ \mathcal{E} \right] = 0.
\end{equation}
Now let us show that conditioned on $\mathcal{E}$, $n^* \in \Gamma_n^{PBFSA}$ \textit{a.s.}; we suppose by contradiction that \textit{a.s.} $n^* \notin \Gamma_n^{PBFSA}$. If $k = \left| \Gamma_n^{PBFSA} \right|$, then $k < \left| \Delta_n \right|$ \textit{a.s.}, and $\displaystyle \min_{j=1,\ldots,k} \left\{ \tilde{f}(n_{(j)}) \right\} \leqslant l\left(n_{(k+1)}\right)$ \textit{a.s.} By definition $\displaystyle \tilde{f}\left( \tilde{n} \right) = \min_{j=1,\ldots,k} \left\{ \tilde{f}(n_{(j)}) \right\}$ so $\tilde{f}\left( \tilde{n} \right) \leqslant l\left(n_{(k+1)}\right)$ \textit{a.s.} Since $n^* \notin \Gamma_n^{PBFSA}$, then there exists $k^* \in \left\{k+1,\ldots,\left| \Delta_n \right| \right\}$ such that $n^* = n_{\left( k^* \right)}$. Since $\left( n_{(i)}\right)_{i \in \{1,\ldots,|\Delta_n|\}}$ are ordered by non-decreasing $l(.)$, we have $l\left(n_{(k+1)}\right) \leqslant l\left(n_{\left( k^* \right)}\right) = l\left(n^*\right)$. Therefore $\tilde{f}\left( \tilde{n} \right) \leqslant l\left(n^*\right)$ \textit{a.s.}; but, conditioned on $\mathcal{E}$, $\tilde{f}\left( \tilde{n} \right) > f\left( n^* \right) \geqslant l\left(n^*\right)$ \textit{a.s.}, which leads to a contradiction.
Thus conditioned on $\mathcal{E}$, $n^* \in \Gamma_n^{PBFSA}$ \textit{a.s.}. Therefore, we have $\tilde{f}(n^*) = PBFSA(n^*, l, \epsilon)$. By induction, $\tilde{f}$ verifies Condition $(A)$ at node $n^*$, thus
\begin{equation}
\label{eq:mathcalE2}
\mathbb{E} \left[ \left( \tilde{f}(n^*) - f(n^*) \right)^+ \right] \leqslant \frac{\epsilon}{2}.
\end{equation}
Finally, since $\tilde{n} = \underset{n_i \in \Gamma^{PBFSA}_n}{argmin} \left\{ \tilde{f}(n_i) \right\}$ and $n^* \in \Gamma_n^{PBFSA}$, then $\tilde{f}(\tilde{n}) \leqslant \tilde{f}(n^*)$ \textit{a.s.}, so we have $\tilde{f}(\tilde{n}) - f(n^*) \leqslant \tilde{f}(n^*) - f(n^*)$ \textit{a.s.} Consequently we have $\left( \tilde{f}(n) - f(n) \right)^+ = \left( \tilde{f}(\tilde{n}) - f(n^*) \right)^+ \leqslant \left( \tilde{f}(n^*) - f(n^*) \right)^+$ \textit{a.s.}, resulting in
\begin{equation}
\label{eq:mathcalE3}
\mathbb{E} \left[ \left. \left( \tilde{f}(n) - f(n) \right)^+ \ \right| \ \mathcal{E}  \right] \leqslant \mathbb{E} \left[ \left. \left( \tilde{f}(n^*) - f(n^*) \right)^+ \ \right| \ \mathcal{E}  \right].
\end{equation}
\item Conditioned on $\overline{\mathcal{E}}$, we have $\left( \tilde{f}(n) - f(n) \right)^+ = 0$ \textit{a.s.}, thus
\begin{equation}
\label{eq:mathcalE4}
\mathbb{E} \left[ \left. \left( \tilde{f}(n) - f(n) \right)^+ \ \right| \ \overline{\mathcal{E}} \right] = 0.
\end{equation}
Moreover, by definition $\tilde{n} \in \Gamma_n^{PBFSA}$ \textit{a.s.}, and $\tilde{f}(\tilde{n}) = PBFSA(\tilde{n}, l, \epsilon)$, thus the induction hypothesis can be applied to $\tilde{n}$. In particular, we have
\begin{equation}
\label{eq:mathcalE5}
\mathbb{E} \left[ \left( \tilde{f}(\tilde{n}) - f(\tilde{n}) \right)^- \right] \leqslant \frac{\epsilon}{2}.
\end{equation}
Finally, it is clear that $f(\tilde{n}) \geqslant f(n^*)$, then $\tilde{f}(\tilde{n}) - f(n^*) \geqslant \tilde{f}(\tilde{n}) - f(\tilde{n})$ \textit{a.s.}, which is equivalent to $\left( \tilde{f}(n) - f(n) \right)^- = \left( \tilde{f}(\tilde{n}) - f(n^*) \right)^- \leqslant \left( \tilde{f}(\tilde{n}) - f(\tilde{n}) \right)^-$ \textit{a.s.}, resulting in
\begin{equation}
\label{eq:mathcalE6}
\mathbb{E} \left[ \left. \left( \tilde{f}(n) - f(n) \right)^- \ \right| \ \overline{\mathcal{E}}  \right] \leqslant \mathbb{E} \left[ \left. \left( \tilde{f}(\tilde{n}) - f(\tilde{n}) \right)^- \ \right| \ \overline{\mathcal{E}}  \right].
\end{equation}
\end{itemize}
Finally, note the following observation: Let $Y \geqslant 0$ \textit{a.s.}, and $\mathcal{F}$ be measurable, then we have 
\begin{equation*}
\mathbb{E} \left[ \left. Y \ \right| \ \mathcal{F}  \right] \mathbb{P} \left( \mathcal{F} \right) \leqslant \mathbb{E} \left[ Y \right] \text{ and } \mathbb{E} \left[ \left. Y \ \right| \ \overline{\mathcal{F}}  \right] \mathbb{P} \left( \overline{\mathcal{F}} \right) \leqslant \mathbb{E} \left[ Y \right].
\end{equation*}
Now we can derive
\begin{equation*}
\mathbb{E} \left[ \left( \tilde{f}(n) - f(n) \right)^+ \right] = \mathbb{E} \left[ \left. \left( \tilde{f}(n) - f(n) \right)^+ \ \right| \ \mathcal{E}  \right] \mathbb{P} \left( \mathcal{E} \right) \leqslant \mathbb{E} \left[ \left. \left( \tilde{f}(n^*) - f(n^*) \right)^+ \ \right| \ \mathcal{E}  \right] \mathbb{P} \left( \mathcal{E} \right) \leqslant \mathbb{E} \left[ \left( \tilde{f}(n^*) - f(n^*) \right)^+ \right] \leqslant \frac{\epsilon}{2},
\end{equation*}
where the first equality comes from Equation (\ref{eq:mathcalE4}), the first inequality uses Equation (\ref{eq:mathcalE3}), the second one holds thanks to $\left( \tilde{f}(n^*) - f(n^*) \right)^+ \geqslant 0$ \textit{a.s.}, and the last one is Equation (\ref{eq:mathcalE2}).
Therefore, $\tilde{f}$ verifies Condition $(A)$ at node $n$.

Similarly, we have
\begin{equation*}
\mathbb{E} \left[ \left( \tilde{f}(n) - f(n) \right)^- \right] = \mathbb{E} \left[ \left. \left( \tilde{f}(n) - f(n) \right)^- \ \right| \ \overline{\mathcal{E}}  \right] \mathbb{P} \left( \overline{\mathcal{E}} \right) \leqslant \mathbb{E} \left[ \left. \left( \tilde{f}(\tilde{n}) - f(\tilde{n}) \right)^- \ \right| \ \overline{\mathcal{E}}  \right] \mathbb{P} \left( \overline{\mathcal{E}} \right) \leqslant \mathbb{E} \left[ \left( \tilde{f}(\tilde{n}) - f(\tilde{n}) \right)^- \right] \leqslant \frac{\epsilon}{2},
\end{equation*}
where the first equality comes from Equation (\ref{eq:mathcalE1}), the first inequality uses Equation (\ref{eq:mathcalE6}), the second one holds thanks to $\left( \tilde{f}(\tilde{n}) - f(\tilde{n}) \right)^- \geqslant 0$ \textit{a.s.}, and the last one is Equation (\ref{eq:mathcalE5}).
Therefore, $\tilde{f}$ verifies Condition $(B)$ at node $n$.

Therefore, we have proven that if $n$ is a decision node with $\lambda_n > S$, $\tilde{f}$ verifies both Conditions $(A)$ and $(B)$ at node $n$, which proves the lemma for decision nodes.

\textbf{If $\boldsymbol{n}$ is a chance node} such that $\tilde{f}(n) = PBFSA(n, l, \epsilon)$, and $\lambda_n > S$. Let us define $\displaystyle \overline{f}(n) = \sum_{n_i \in \Psi_n^{PBFSA}} p^n_{n_i}f(n_i)$, and show that
\begin{equation}
\label{eq:errorDecision}
\mathbb{E} \left[ \left( \tilde{f}(n) - \overline{f}(n) \right)^+ \right] \leqslant \frac{\epsilon - \epsilon_n}{2} \text{ and } \mathbb{E} \left[ \left( \tilde{f}(n) - \overline{f}(n) \right)^- \right] \leqslant \frac{\epsilon - \epsilon_n}{2}.
\end{equation}

Recall that $\forall n_i \in \Psi_n^{PBFSA}$, $\lambda_{n_i} = \lambda_n$, and $n_i$ are decision nodes such that $\tilde{f}(n_i) = PBFSA(n_i, l, \epsilon - \epsilon_n)$. Therefore, using the previous result, we know that $\displaystyle \mathbb{E} \left[ \left( \tilde{f}(n_i) - f(n_i) \right)^+ \right] \leqslant \frac{\epsilon - \epsilon_n}{2}$ and $\displaystyle \mathbb{E} \left[ \left( \tilde{f}(n_i) - f(n_i) \right)^- \right] \leqslant \frac{\epsilon - \epsilon_n}{2}$. We derive the following calculations:
\begin{align*}
\mathbb{E} \left[ \left( \tilde{f}(n) - \overline{f}(n) \right)^+ \right]  \ = \ & \mathbb{E} \left[ \left( \sum_{n_i \in \Psi_n^{PBFSA}} p^n_{n_i}\left( \tilde{f}(n_i) - f(n_i) \right) \right)^+ \right] \leqslant \mathbb{E} \left[ \sum_{n_i \in \Psi_n^{PBFSA}} p^n_{n_i} \left( \tilde{f}(n_i) - f(n_i) \right)^+ \right] \\
= \ & \sum_{n_i \in \Psi_n^{PBFSA}} p^n_{n_i} \mathbb{E} \left[ \left( \tilde{f}(n_i) - f(n_i) \right)^+ \right] \leqslant \sum_{n_i \in \Psi_n^{PBFSA}} p^n_{n_i} \frac{\epsilon - \epsilon_n}{2} = \frac{\epsilon - \epsilon_n}{2}.
\end{align*}
Similarly, we have
\begin{align*}
\mathbb{E} \left[ \left( \tilde{f}(n) - \overline{f}(n) \right)^- \right]  \ = \ & \mathbb{E} \left[ \left( \sum_{n_i \in \Psi_n^{PBFSA}} p^n_{n_i}\left( \tilde{f}(n_i) - f(n_i) \right) \right)^- \right] \leqslant \mathbb{E} \left[ \sum_{n_i \in \Psi_n^{PBFSA}} p^n_{n_i} \left( \tilde{f}(n_i) - f(n_i) \right)^- \right] \\
= \ & \sum_{n_i \in \Psi_n^{PBFSA}} p^n_{n_i} \mathbb{E} \left[ \left( \tilde{f}(n_i) - f(n_i) \right)^- \right] \leqslant \sum_{n_i \in \Psi_n^{PBFSA}} p^n_{n_i} \frac{\epsilon - \epsilon_n}{2} = \frac{\epsilon - \epsilon_n}{2},
\end{align*}
which proves Equation (\ref{eq:errorDecision}).

If $N_n(\epsilon_n) > C_{w_{min}}!$, then $f(n) = \overline{f}(n)$ so $\tilde{f}(n) - f(n) = \tilde{f}(n) - \overline{f}(n)$ \textit{a.s.}, and since $\frac{\epsilon - \epsilon_n}{2} \leqslant \frac{\epsilon}{2}$, Equation (\ref{eq:errorDecision}) implies that $\tilde{f}$ verifies Conditions $(A)$ and $(B)$ at node $n$.

Otherwise, we have $N_n(\epsilon_n) \leqslant C_{w_{min}}!$. Since $\Psi_n^{PBFSA}$ is constructed using $\displaystyle N_n(\epsilon_n) = \frac{\pi \left( f_{max}(n) - f_{min}(n) \right)^2}{2\epsilon_n^2}$ samples, thus by using Corollary \ref{corr:Hoeffding}, we have
\begin{equation}
\label{eq:errorChance}
\mathbb{E} \left[ \left( \overline{f}(n) - f(n) \right)^+ \right] \leqslant \frac{\epsilon_n}{2} \text{ and } \mathbb{E} \left[ \left( \overline{f}(n) - f(n) \right)^- \right] \leqslant \frac{\epsilon_n}{2}
\end{equation}
By combining Equations (\ref{eq:errorDecision}) and (\ref{eq:errorChance}), we have
\begin{eqnarray*}
\mathbb{E} \left[ \left( \tilde{f}(n) - f(n) \right)^+ \right] \leqslant \mathbb{E} \left[ \left( \tilde{f}(n) - \overline{f}(n) \right)^+ \right] + \mathbb{E} \left[ \left( \overline{f}(n) - f(n) \right)^+ \right] \leqslant \frac{\epsilon - \epsilon_n}{2} + \frac{\epsilon_n}{2} = \frac{\epsilon}{2}, \\ 
\mathbb{E} \left[ \left( \tilde{f}(n) - f(n) \right)^- \right] \leqslant \mathbb{E} \left[ \left( \tilde{f}(n) - \overline{f}(n) \right)^- \right] + \mathbb{E} \left[ \left( \overline{f}(n) - f(n) \right)^- \right] \leqslant \frac{\epsilon - \epsilon_n}{2} + \frac{\epsilon_n}{2} = \frac{\epsilon}{2},
\end{eqnarray*}
which shows that $\tilde{f}$ verifies Conditions $(A)$ and $(B)$ at node $n$ and concludes the proof.
\end{proof}

\section{Computational experiments}
\label{sec::ComputationExperiments}

Having introduced lower and upper bounds, $PBFS$, $PBFSA$, and theoretical guarantees in previous sections, we present several experimental results in this section to understand the effectiveness of our algorithms for the SCRP. For clarity, we refer to the set of instances from \cite{Ku} as the \textit{existing dataset}. We present 4 sets of experiments:
\begin{enumerate}
\item Based on instances from the existing dataset, which have relatively small batches, we test the $PBFS$ algorithm, as well as the two new heuristics and our lower bounds.
\item We slightly modify the existing dataset to obtain the \textit{modified dataset}, in order to obtain instances with relatively larger batches. We test the efficiency of $PBFSA$ on this modified dataset.
\item Based on the existing dataset, we show that $PBFS$ improves on the algorithm proposed in \cite{Ku} for the online model. Moreover, the two new heuristics ($EM$ and $EG$) outperform the $ERI$ algorithm on expectation for the majority of the instances of the dataset.
\item We change the existing dataset by considering that all containers belong to a unique batch. We show strong computational evidence to support Conjecture \ref{conj1}, which states that the leveling policy is optimal for the SCRP under the online model with a unique batch.
\end{enumerate}
All experiments are performed on a MacBook Pro with 2.2 GHz Intel Core i7 processor, 8.00 GB of RAM and the programming language is MATLAB 2016a. Finally, all results and instances used in this section are available at \url{https://github.com/vgalle/StochasticCRP}.

\paragraph{Implementation of heuristics}

\begin{enumerate}
\item Computing the number of relocations using $b$ when there are $S$ containers or less: In the retrieval process, \textbf{when there are $\boldsymbol{S}$ containers or less remaining in the configuration, the expected number of relocations performed by $\boldsymbol{ERI}$, $\boldsymbol{EM}$, $\boldsymbol{EG}$ and $\boldsymbol{L}$ is computed using $\boldsymbol{b}$}. This is motivated by the following observation: $ERI$, $EM$, $EG$ and $L$ are optimal when there are $S$ containers or fewer remaining in the configuration, and Lemma \ref{lemma:lowerBoundPrune} shows that the optimal expected number of relocations in this case is equal to $b$. Therefore, for all heuristics (except Random), instead of running simulations until there are no containers left, we stop when there are $S$ containers left and compute the expected number of relocations using $b$ instead.
\item Estimate the expected number of relocations using sampling: In order to estimate the exact objective value for a given heuristic, one would have to consider all possible retrieval scenarios. Instead, for each heuristic \textbf{unless specified otherwise, we report the average over 5000 samples (of retrieval orders) for each instance}, where samples are uniformly drawn at random.
\end{enumerate}

\paragraph{Existing dataset description}
The full description of the dataset can be found in \cite{Ku} and the original data set is available at \url{http://crp-timewindow.blogspot.com}. Note that:
\begin{itemize}
\item Configuration sizes vary from $\mathbf{T = 3,\ldots,6}$ tiers, and $\mathbf{S = 5,\ldots,10}$ stacks.
\item Two occupancy rates are considered, \textbf{50 and 67 percent}. The occupancy rate ($\mu \in [0,1]$) is defined such that the initial number of containers is $C = round \left( \mu \times S \times T \right)$, where $round(x)$ rounds $x$ to the closest integer. Therefore, a given triplet $(T,S,\mu)$ is equivalent to a given triplet $(T,S,C)$, and note that if $C = round \left( 0.67 \times S \times T \right)$, the condition $ 0 \leqslant C \leqslant ST - (T - 1 )$ is satisfied.
\item Given a configuration size ($T$ and $S$) and an occupancy rate ($\mu$) resulting in a given initial number of containers ($C$), the dataset includes \textbf{30 different initial configurations}.
\item For all 1'440 instances, the ratio between the number of batches and $C$ is taken to be around half, \textit{i.e.}, there are on average two containers per batch, which is the smallest size for a batch.
\end{itemize}

\textbf{In all our experiments, the time limit is set to an hour, and the $1^{st}$ look-ahead lower bound $b_1$ is used as input for both $PBFS$ and $PBFSA$}. All instances are solved by heuristics and lower bounds within seconds or less.

\subsection{Experiment 1: Batch model with small batches}

\begin{table}[h]
  \centering
  \ra{1}
  \resizebox{\textwidth}{!}{
\begin{tabular}{l l l l l l l l l l l l l l l}\toprule
  && $ T$ & & \multicolumn{2}{l}{3} & & \multicolumn{2}{l}{4} & & \multicolumn{2}{l}{5} & & \multicolumn{2}{l}{6}
\\ \cline{5-6}  \cline{8-9} \cline{11-12} \cline{14-15} $S$ & & Fill rate & &  50 percent & 67 percent & & 50 percent & 67 percent & & 50 percent & 67 percent & & 50 percent & 67 percent \\
\hline
5 & & $C$ & & 8 & 10 & & 10 & 13 & & 13 & 17 & & 15 & 20 \\
  & & Solved & & \CheckmarkBold & \CheckmarkBold & & \CheckmarkBold & \CheckmarkBold  & & \CheckmarkBold  & \textcolor{red}{28/30}  & & \CheckmarkBold  & \textcolor{red}{15/30}  \\
    & & Time (s) & & 0.01 & 0.02 & & 0.03 & 0.12  & & 0.17  &  & & 5.17  &  \\ \addlinespace
    \hline
6 & & $C$ & & 9 & 12 & & 12 & 16 & & 15 & 20 & & 18 & 24 \\ 
  & & Solved & & \CheckmarkBold & \CheckmarkBold & & \CheckmarkBold & \CheckmarkBold  & & \CheckmarkBold  & \textcolor{red}{25/30}  & & \CheckmarkBold & \textcolor{red}{14/30}  \\
  & & Time (s) & & 0.01 & 0.03 & & 0.04 & 0.86 & & 2.90 &  & & 15.94 &  \\ \addlinespace
  \hline
7 & & $C$ & & 11 & 14 & & 14 & 19 & & 18 & 23 & & 21 & 28 \\ 
  & & Solved & & \CheckmarkBold & \CheckmarkBold & & \CheckmarkBold & \CheckmarkBold & & \CheckmarkBold  & \textcolor{red}{24/30}  & & \textcolor{red}{23/30}  & \textcolor{red}{5/30}  \\
  & & Time (s) & & 0.02 & 0.04 & & 0.04 & 0.83 & & 1.37 &  & &  &  \\ \addlinespace
  \hline
8 & & $C$ & & 12 & 16 & & 16 & 21 & & 20 & 27 & & 24 & 32 \\
  & & Solved & & \CheckmarkBold & \CheckmarkBold & & \CheckmarkBold & \CheckmarkBold  & & \CheckmarkBold  & \textcolor{red}{20/30}  & & \textcolor{red}{22/30}  & \textcolor{red}{5/30}  \\
  & & Time (s) & & 0.01 & 0.06 & & 0.16 & 10.04  & & 6.84 &  & &  &  \\ \addlinespace
  \hline
9 & & $C$ & & 14 & 18 & & 18 & 24 & & 23 & 30 & & 27 & 36 \\
  & & Solved & & \CheckmarkBold & \CheckmarkBold & & \CheckmarkBold & \CheckmarkBold  & & \textcolor{red}{29/30} & \textcolor{red}{10/30}  & & \textcolor{red}{19/30} & \textcolor{red}{2/30}  \\
  & & Time (s) & & 0.03 & 0.10 & & 0.37 & 8.84  & &   &  & &  &  \\ \addlinespace
  \hline
10 & & $C$ & & 15 & 20 & & 20 & 27 & & 25 & 34 & & 30 & 40 \\
  & & Solved & & \CheckmarkBold & \CheckmarkBold & & \CheckmarkBold & \textcolor{red}{28/30} & & \textcolor{red}{29/30} & \textcolor{red}{12/30}  & & \textcolor{red}{22/30} & \textcolor{red}{2/30}  \\
  & & Time (s) & & 0.03 & 0.10 & & 0.54 &  & & &  & &  &  \\
 \addlinespace
 \hline
\end{tabular}}
\caption{Instances solved by $PBFS$ in the batch model with small batches.}
\label{tab:solve_1}
\end{table}

Table \ref{tab:solve_1} gives a summary of the results as follows: \CheckmarkBold indicates that all 30 instances are solved optimally by $PBFS$. In this case, the average solution time to solve these instances is given in seconds. Otherwise, the number of instances solved optimally is provided in red and in the form x/30. This table shows the efficiency of $PBFS$ as it can solve all instances except two, for $T=3$ and $T=4$. Most importantly, the average time to solve these instances is under 10 seconds for these problem sizes. Since many ports today have a maximum tier requirement of 4 and need fast solutions, $PBFS$ could be used in practice in the case of small batches. However, for $T=5$ and $6$, $PBFS$ cannot solve all instances optimally in a timely manner. This suggests that, as the problem grows slightly, some instances become very hard to solve, which should not be a surprise, knowing the NP-hardness of the problem. In order to avoid such situations in real operations, heuristics can be used to provide a ``good'' sub-optimal solution (good in the sense of being not too far from optimality). Therefore, we want to evaluate the performance of these heuristics in order to know which one should be used in real operations.

We measure the performance of heuristics and the tightness of lower bounds in Tables \ref{tab:UBLB_50_1} and \ref{tab:UBLB_67_1}. 
Concerning lower bounds, $b$ encompasses a significant number of relocations. Adding unavoidable ``bad'' relocations in $b_1$ and $b_2$, improves slightly the lower bound. But experiments seem to confirm that $b_2(n) - b_1(n) \leqslant b_1(n) - b(n)$ holds, supporting our intuition that the relative increase of lower bounds $b_k(n) - b_{k-1}(n)$ decreases with $k$.

Concerning heuristics, $EG$ and $EM$ clearly outperform $ERI$ as they result in lower expected numbers of relocations. When we have access to the optimal solutions, both heuristics are on average at most 2\% more than the optimal solution. We expect this behavior to be similar for larger instances, however we only have access to lower bounds to evaluate their performances. In this case, heuristics are on average at most 11\% more than $b_2$, hence at most 11\% from the optimal solution (even though we believe that this number is very conservative, as our lower bounds are not ``tight'').
Therefore, both $EG$ and $EM$ appear to be good solutions for the SCRP under the batch model with small batches. In this case, we recommend using $EM$ for its simplicity of implementation and understandability.

\subsection{Experiment 2: Batch model with larger batches}
\subsubsection{Modifying existing instances}
For the sake of reproducibility, we use the existing set of instances, but slightly modify it to consider larger batches. In order to create these instances, for each original instance $n$, consider $n'$ with the same containers in the same configuration. But, if $w$ is the batch of a container $c$ in $n$, then we take the batch of $c$ in $n'$ to be $w' = \left\lceil \frac{w}{\gamma} \right\rceil$, where $\gamma > 1$, \textit{i.e.}, we merge $\gamma$ batches together. In these experiments, we take $\gamma = 2$, which implies that batches have an average size of 4.

\begin{table}[h]
  \centering
  \ra{1}
  \resizebox{\textwidth}{!}{
\begin{tabular}{l l l l l l l l l l l l l l l}\toprule
  && $ T$ & & \multicolumn{2}{l}{3} & & \multicolumn{2}{l}{4} & & \multicolumn{2}{l}{5} & & \multicolumn{2}{l}{6}
\\ \cline{5-6}  \cline{8-9} \cline{11-12} \cline{14-15} $S$ & & Fill rate & &  50 percent & 67 percent & & 50 percent & 67 percent & & 50 percent & 67 percent & & 50 percent & 67 percent \\
\hline
5 & & $C$ & & 8 & 10 & & 10 & 13 & & 13 & 17 & & 15 & 20 \\
  & & Solved & & \CheckmarkBold & \CheckmarkBold & & \CheckmarkBold & \CheckmarkBold & & \CheckmarkBold & \textcolor{red}{21/30} & & \CheckmarkBold  &  \textcolor{red}{3/30} \\
  & & Time (s) & & 0.08 & 0.29 & & 0.14 & 4.55 & & 3.20 &  & & 72.70 &  \\
\addlinespace
\hline
6 & & $C$ & & 9 & 12 & & 12 & 16 & & 15 & 20 & & 18 & 24 \\ 
  & & Solved & & \CheckmarkBold & \CheckmarkBold & & \CheckmarkBold & \CheckmarkBold & & \CheckmarkBold  & \textcolor{red}{18/30} & & \textcolor{red}{27/30} & \textcolor{red}{1/30} \\
  & & Time (s) & & 0.08 & 0.47 & & 0.25 & 126.37 & & 14.74 &  & &  & \\
\addlinespace
\hline
7 & & $C$ & & 11 & 14 & & 14 & 19 & & 18 & 23 & & 21 & 28 \\ 
  & & Solved & & \CheckmarkBold & \CheckmarkBold & & \CheckmarkBold & \textcolor{red}{29/30} & & \CheckmarkBold & \textcolor{red}{9/30} & & \textcolor{red}{14/30} & \textcolor{red}{0/30}  \\
  & & Time (s) & & 0.13 & 0.71 & & 0.58 &  & & 17.74 &  & &  & \\
\addlinespace
\hline
8 & & $C$ & & 12 & 16 & & 16 & 21 & & 20 & 27 & & 24 & 32 \\
  & & Solved & & \CheckmarkBold & \CheckmarkBold & & \CheckmarkBold & \textcolor{red}{28/30} & & \textcolor{red}{29/30} & \textcolor{red}{6/30} & & \textcolor{red}{17/30} & \textcolor{red}{1/30} \\
  & & Time (s) & & 0.08 & 1.67 & & 1.26 &  & &  &  & &  &  \\
\addlinespace
\hline
9 & & $C$ & & 14 & 18 & & 18 & 24 & & 23 & 30 & & 27 & 36 \\
  & & Solved & & \CheckmarkBold & \CheckmarkBold & & \CheckmarkBold & \textcolor{red}{26/30} & & \textcolor{red}{25/30} & \textcolor{red}{5/30} & & \textcolor{red}{15/30} & \textcolor{red}{0/30}  \\
  & & Time (s) & & 0.13 & 1.49 & & 1.47 &  & &  &  & &  & \\
\addlinespace
\hline
10 & & $C$ & & 15 & 20 & & 20 & 27 & & 25 & 34 & & 30 & 40 \\
  & & Solved & & \CheckmarkBold & \CheckmarkBold & & \CheckmarkBold & \textcolor{red}{22/30} & & \textcolor{red}{29/30} & \textcolor{red}{7/30} & & \textcolor{red}{14/30} & \textcolor{red}{0/30} \\
  & & Time (s) & & 0.17 & 0.79 & & 3.10 &  & & &  & &  & \\
  \addlinespace
 \hline
\end{tabular}}
\caption{Instances solved by $PBFSA$ in the batch model with larger batches.}
\label{tab:solve_2}
\end{table}

\subsubsection{Target error $\epsilon$}
In order to set our target error, we consider the following. Let $n_0$ be a given instance, and set $\epsilon = \frac{b(n_0)}{2}$. In this case, we know that $\epsilon \leqslant \frac{f(n_0)}{2}$, which implies that we are making an error of at most 50\%. Note that this error is very conservative due to 2 major things: first, $b(n_0)$ is not necessarily representative of $f(n_0)$, specially if $n_0$ has many containers. Second, the number of samples given by Hoeffding's inequality is also very conservative, probably making our approximation substantially more accurate than what we can theoretically prove.

\subsubsection{Results}

Results are summarized in Tables \ref{tab:solve_2}. Similarly to Table \ref{tab:solve_1}, \CheckmarkBold indicates that all 30 instances are solved approximately by $PBFSA$ within the given expected error. In this case, the average solution time to solve these instances is given in seconds. Otherwise, the number of instances solved is provided in red and in the form x/30. This table shows that $PBFSA$ presents several advantages. First, it solves most of instances with $T=4$ and $S \leqslant 9$ approximately within 2 minutes, while we note that $PBFS$ was not able to solve most of these. Moreover, as can be seen in Tables \ref{tab:UBLB_50_2} and \ref{tab:UBLB_67_2}, $PBFSA$ still outperforms the best heuristics despite the fact that we only set the theoretical guarantee to 50\% of optimality. Together, these two advantages show the practicality of $PBFSA$ for problem sizes typically encountered in real ports. Moreover, we note that increasing the batch size appears to make the problem significantly more complicated to solve as we can solve optimally larger instances in Experiment 1 (see Table \ref{tab:solve_1}). Finally, we remark that similar conclusions of Experiment 1 can be drawn for lower and upper bounds (see Tables \ref{tab:UBLB_50_2} and \ref{tab:UBLB_67_2}).

\subsection{Experiment 3: Online model and comparison with \cite{Ku}}
Table \ref{tab:solve_3} gives a summary similar to the two previous experiments. In addition, we report the results of \cite{Ku} who take a \textit{time limit of eight hours} for each instance. In this table, \checkmark {\fontsize{8}{8}\selectfont (\checkmark)} indicates that all 30 instances are solved optimally by both $PBFS$ and \cite{Ku}. In this case, 
the average solution time in seconds to solve these instances is given for $PBFS$ and for \cite{Ku} in parenthesis. \CheckmarkBold indicates that all 30 instances are solved optimally only by $PBFS$ but not \cite{Ku}. In this case, only the average solution time to solve these instances with $PBFS$ is given in seconds. Otherwise, the number of instances solved by $PBFS$ is provided in red and in the form x/30.

\begin{table}[h]
  \centering
  \ra{1}
  \resizebox{\textwidth}{!}{
\begin{tabular}{l l l l l l l l l l l l l l l}\toprule
  && $ T$ & & \multicolumn{2}{l}{3} & & \multicolumn{2}{l}{4} & & \multicolumn{2}{l}{5} & & \multicolumn{2}{l}{6}
\\ \cline{5-6}  \cline{8-9} \cline{11-12} \cline{14-15} $S$ & & Fill rate & &  50 percent & 67 percent & & 50 percent & 67 percent & & 50 percent & 67 percent & & 50 percent & 67 percent \\
\hline
5 & & $C$ & & 8 & 10 & & 10 & 13 & & 13 & 17 & & 15 & 20 \\ 
 & & Solved & & \checkmark {\fontsize{8}{8}\selectfont (\checkmark)}  & \CheckmarkBold & & \checkmark {\fontsize{8}{8}\selectfont (\checkmark)} & \CheckmarkBold & & \checkmark {\fontsize{8}{8}\selectfont (\checkmark)} & \textcolor{red}{28/30} & & \CheckmarkBold & \textcolor{red}{18/30} \\
 & & Time (s) & & 0.01 {\fontsize{8}{8}\selectfont (0.02)}  & 0.02 & & 0.01 {\fontsize{8}{8}\selectfont (2.51)} & 0.09 & & 0.16 {\fontsize{8}{8}\selectfont (2483.30)} &  & & 3.74 &  \\
\addlinespace
\hline
 6 & & $C$ & & 9 & 12 & & 12 & 16 & & 15 & 20 & & 18 & 24 \\ 
  & & Solved & & \checkmark {\fontsize{8}{8}\selectfont (\checkmark)} & \CheckmarkBold & & \checkmark {\fontsize{8}{8}\selectfont (\checkmark)} & \CheckmarkBold & & \CheckmarkBold & \textcolor{red}{25/30} & & \CheckmarkBold & \textcolor{red}{15/30} \\
  & & Time & & 0.01 {\fontsize{8}{8}\selectfont (0.01)} & 0.04 & & 0.06 {\fontsize{8}{8}\selectfont (139.08)} & 0.81 & & 1.85 & & & 14.92 & \\
\addlinespace
\hline
 7 & & $C$ & & 11 & 14 & & 14 & 19 & & 18 & 23 & & 21 & 28 \\ 
  & & Solved & & \checkmark {\fontsize{8}{8}\selectfont (\checkmark)} & \CheckmarkBold & & \checkmark {\fontsize{8}{8}\selectfont (\checkmark)} & \CheckmarkBold & & \CheckmarkBold & \textcolor{red}{24/30}  & & \textcolor{red}{23/30}  & \textcolor{red}{5/30} \\
  & & Time (s) & & 0.02 {\fontsize{8}{8}\selectfont (0.33)} & 0.04 & & 0.04 {\fontsize{8}{8}\selectfont (207.62)} & 0.67 & & 1.38 &  & &  & \\
\addlinespace 
\hline
 8 & & $C$ & & 12 & 16 & & 16 & 21 & & 20 & 27 & & 24 & 32 \\
  & & Solved & & \checkmark {\fontsize{8}{8}\selectfont (\checkmark)} & \CheckmarkBold & & \CheckmarkBold & \CheckmarkBold & & \CheckmarkBold & \textcolor{red}{20/30} & &  \textcolor{red}{22/30} & \textcolor{red}{5/30} \\
  & & Time (s) & & 0.01 {\fontsize{8}{8}\selectfont (0.33)} & 0.05 & & 0.10 & 8.29 & & 5.85 &  & &  &  \\
\addlinespace
\hline
 9 & & $C$ & & 14 & 18 & & 18 & 24 & & 23 & 30 & & 27 & 36 \\ 
  & & Solved & & \checkmark {\fontsize{8}{8}\selectfont (\checkmark)} & \CheckmarkBold & & \CheckmarkBold & \CheckmarkBold & & \textcolor{red}{29/30} & \textcolor{red}{10/30} & & \textcolor{red}{19/30}  & \textcolor{red}{2/30} \\
  & & Time (s) & & 0.02 {\fontsize{8}{8}\selectfont (32.24)} & 0.09 & & 0.38 & 7.26 & &   &  & &  & \\
\addlinespace
\hline
 10 & & $C$ & & 15 & 20 & & 20 & 27 & & 25 & 34 & & 30 & 40 \\ 
  & & Solved & & \checkmark {\fontsize{8}{8}\selectfont (\checkmark)} & \CheckmarkBold & & \CheckmarkBold & \textcolor{red}{28/30}  & & \textcolor{red}{29/30} & \textcolor{red}{12/30}   & & \textcolor{red}{16/30}  & \textcolor{red}{2/30} \\
  & & Time (s) & & 0.03 {\fontsize{8}{8}\selectfont (58.85)} & 0.08 & & 0.52 &  & &  &  & &  & \\
  \addlinespace
 \hline
\end{tabular}}
\caption{Instances solved by $PBFS$ and \cite{Ku} in the online model with small batch.}
\label{tab:solve_3}
\end{table}

Results show strong evidence that our solution is improving significantly the best existing results for the SCRP under the online model, given that we solve many larger instances optimally. Furthermore, it also outperforms the most recent algorithm in solution time for the problem sizes it can solve. It appears that, for problems for which we can solve all (or almost all) instances, most instances are ``easy'' to solve as the algorithm finds a solution within seconds. However, as in the batch model, there exists some instances for which the optimal solution still requires an exponential number of nodes, which makes our algorithm not tractable.

In Tables \ref{tab:UBLB_50_3} and \ref{tab:UBLB_67_3}, we also report in parenthesis the averages for ERI and Random found by \cite{Ku}. The results for random are consistent. However, we find significant better results for our implementation of $ERI$. This is unexpected since the only difference between the two implementations is the use of lower bound $b$, when the configuration has less than $S$ containers remaining. Nevertheless, $ERI$ should also be optimal in this case, as it reduces to heuristic $L$. So this should not affect the expected number of relocations, and we cannot explain this difference.
Finally, we point out that the results are quite similar to those of Experiment 1. Indeed, the existing data set has relatively small batches (on average 2 containers), which inherently makes the two models, batch and online, very close to each other.

\subsection{Experiment 4: Online model with a unique batch}

\begin{table}[h]
  \centering
  \ra{1}
  \resizebox{\textwidth}{!}{
\begin{tabular}{l l l l l l l l l l l l l l l}\toprule
  && $ T$ & & \multicolumn{2}{l}{3} & & \multicolumn{2}{l}{4} & & \multicolumn{2}{l}{5} & & \multicolumn{2}{l}{6}
\\ \cline{5-6}  \cline{8-9} \cline{11-12} \cline{14-15} $S$ & & Fill rate & &  50 percent & 67 percent & & 50 percent & 67 percent & & 50 percent & 67 percent & & 50 percent & 67 percent \\
\hline
5 & & $C$ & & 8 & 10 & & 10 & 13 & & 13 & 17 & & 15 & 20 \\ 
  & & $PBFS$ & & 2.08 & 3.33 & & 3.54 & 6.53 & & 6.56 & 12.05 & & 9.13 & 17.28 \\
  & & $L.$ & & 2.08 & 3.33 & & 3.54 & 6.52 & & 6.57 & 12.06 & & 9.14 & 17.29 \\
  \addlinespace
\hline
6 & & $C$ & & 9 & 12 & & 12 & 16 & & 15 & 20 & & 18 & 24 \\ 
  & & $PBFS$ & & 2.10 & 4.04 & & 4.23 & 8.02 & & 7.01 & 13.48 & & 10.73 & 20.13 \\
  & & $L.$ & & 2.10 & 4.04 & & 4.23 & 8.02 & & 7.01 & 13.48 & & 10.73 & 20.13 \\
  \addlinespace
  \hline
7 & & $C$ & & 11 & 14 & & 14 & 19 & & 18 & 23 & & 21 & 28 \\ 
  & & $PBFS$ & & 2.69 & 4.60 & & 4.82 & 9.55 & & 8.58 & 14.75 & & 12.22 & 23.14 \\
  & & $L.$ & & 2.69 & 4.61 & & 4.82 & 9.55 & & 8.58 & 14.74  & & 12.22 & 23.14 \\
  \addlinespace
  \hline
8 & & $C$ & & 12 & 16 & & 16 & 21 & & 20 & 27 & & 24 & 32 \\ 
  & & $PBFS$ & & 2.61 & 5.19 & & 5.51 & 9.96 & & 9.12 & 17.75 & & 13.83 & - \\
  & & $L.$ & & 2.62 & 5.19 & & 5.51 & 9.95 & & 9.12 & 17.75 & & 13.83 & 26.04  \\
  \addlinespace
  \hline
9 & & $C$ & & 14 & 18 & & 18 & 24 & & 23 & 30 & & 27 & 36 \\ 
  & & $PBFS$ & & 3.31 & 5.72 & & 6.10 & 11.58 & & 10.89 & 19.15 & & - & - \\
  & & $L.$ & & 3.31 & 5.72 & & 6.10 & 11.58 & & 10.89 & 19.14 & & 15.40 & 28.84 \\
  \addlinespace
  \hline
10 & & $C$ & & 15 & 20 & & 20 & 27 & & 25 & 34 & & 30 & 40 \\ 
   & & $PBFS$ & & 3.36 & 6.38 & & 6.68 & 12.98 & & 11.13 & 22.07 & & - & - \\
   & & $L.$ & & 3.36 & 6.38 & & 6.67 & 12.99 & & 11.13 & 22.06 & & 16.87 & 31.77 \\
   \addlinespace
 \hline
\end{tabular}}
\caption{Instances solved with $PBFS$ and heuristic $L$ in the online model with a unique batch.}
\label{tab:expe4}
\end{table}

In this experiment, we consider the existing data set, but assign all containers into a unique batch ($W = 1$). We consider the SCRP under online model, where containers are revealed one at a time. Note that, in this case, each container is equally likely to be retrieved, and it is equivalent to know no-information about containers relative retrieval order. For each instance, we solve it twice: first using $PBFS$, and then using heuristic $L$, for which we sample \textit{10000 scenarios} (this is different from the 5000 samples considered in previous experiments). We report the results in Table \ref{tab:expe4}. In this table, for each problem size, we report the expected optimal number of relocations averaged over 30 instances. `-' means that all 30 instances could not be solved optimally with $PBFS$ within the given time limit of an hour. Note that, the exepected number of relocations using heuristic $L$ reported in this experiment might be less than the one of $PBFS$; this is only due to the fact that we are sampling.
Intuitively, $L$ should be the optimal solution in this setting, and this experiment shows strong evidence that the next conjecture holds.
\begin{conjecture}
\label{conj1}
Consider a configuration $n$ with a unique batch. Let $f^o(n)$ be the minimum expected number of relocations to empty $n$ under the online model, and $f^{o,L}(n)$ be the expected number of relocations performed by the Leveling heuristic under the online model, then
\begin{equation}
f^o(n) = f^{o,L}(n).
\end{equation}
\end{conjecture}

This conjecture could also be made in the dynamic case, when containers arrive to be stacked. These results would have a strong implication in the port operations: \textit{leveling configurations is the optimal policy when minimizing relocations, and no-information is given in advance}.

\section{Discussion}
\label{sec:conclusion}

Managing relocation moves is one of the main challenges in the storage yard of container terminals, because it has a direct effect on the costs and efficiency of yard operations. The Container Relocation Problem, notorious for its computational intractability, addresses this issue. In this paper, we extend the CRP to the more practical case in which the retrieval order of containers is not known far in advance. First, we introduce a new stochastic model, called the batch model, show the applicability of this model and compare it theoretically with the existing model of \cite{Zhao}. Then, we derive lower bounds and fast and efficient heuristics for the SCRP. Subsequently, we develop two novel algorithms ($PBFS$ and $PBFSA$) to solve the stochastic CRP in different settings. Efficiencies of all algorithms are supported through computational experiments, for which all results are made available online. Finally, using our solution methods and based on extensive experiments, we conjecture the optimality of the simple Leveling heuristic in the online stochastic setting. More generally, the methods developed in this paper apply to multi-stage stochastic optimization problems, where the number of stages is finite, the set of feasible actions at each stage is finite, the objective function is bounded and bounds on the objective function can be easily computed.

Future work could include the proof of Conjecture \ref{conj1}. Another important future work is the optimal design of time windows for a TAS. On one hand, small time windows imply more information on the retrieval sequence, hence higher operational efficiency of port operators. On the other hand, large time windows insure higher flexibility for truck drivers and a high rate of on-time arrivals. In order to balance this trade-off, one would need to quantify two important metrics with respect to the expected number of relocations: the ``value of information'' and the assignment of containers to ``wrong'' batches. Finally, in the grand scheme of port operations, the study of stacking and retrieving simultaneously, as well as the extension in the row dimension of blocks is important for future studies of operations to take into account.

\section*{Acknowledgements}
The authors would like to thank anonymous reviewers and the editor for their suggestions that led to the significant improvements of this paper.
 
\bibliographystyle{apalike}
\bibliography{refs} 

\appendix

\section{Theoretical and computational comparison of the batch and the online models}

\subsection{Theoretical comparison: proof of Lemma \ref{lem:compModels}}
\compModels*
\begin{proof}
We prove this lemma by induction on the number of batches $W$.
The lemma clearly holds if $y$ is empty (\textit{i.e.} $W=0$). Now consider $W \geqslant 1$ and $C_1 \geqslant 1$. For the sake of clarity of the proof, we define the following notation:
\begin{equation}
\label{eq:newProc}
\forall \ d \in \{1,\ldots,C_1\} \ , \left\{ \begin{array}{ll}
y \ \xrightarrow{\zeta_{1},\ldots,\zeta_{d}} \ x^d_{1}
\\
x^d_{k} \ \xrightarrow{a_{k}} \ x^d_{k+1}, \text{ if } d > 1 \ , \ \forall \ k \in \{1,\ldots,d-1\}
\\
x^d_k \ \xrightarrow{a_{k}} \ y^d_{k-d+2} \ , \ \forall \ k \in \{d,\ldots,C\}
\\
y^d_{k-d+1} \ \xrightarrow{\zeta_{k}} \ x^d_k \ , \ \forall \ k \in \{d+1,\ldots,C\}.
\end{array}
\right.
\end{equation}
These notations corresponds to the following process: the first $d$ containers to be retrieved are all revealed at once. Then decisions to retrieve these $d$ containers are made. Afterwards, each of the $C-d$ remaining containers is revealed one at a time (as in the online model). Under this revelation process, the minimum expected number of relocation is given by
\[
f^d\left(y\right) = \underset{\zeta_{1},\ldots,\zeta_{d}}{\mathbb{E}} \left[ \min_{a_{1},\ldots,a_{d}} \left\{ \sum_{k=1}^d r\left(x^d_k\right) + f^o\left(y^d_{2}\right) \right\} \right], \ \forall \ d \in \{1,\ldots,C_1\}.
\]
Moreover, using the recursion formula from the online model, we have
\[
f^o\left(y^d_{2}\right) =  \underset{\zeta_{d}}{\mathbb{E}} \left[ \min_{a_{d}} \left\{ r\left(x^{d}_{d+1}\right)  + f^o\left(y^{d}_{3}\right) \right\} \right].
\]
In particular, by definition of the online model, we have $f^o\left(y\right) = f^{1}\left(y\right)$.

Using these relations, let us prove that
\begin{equation}
\label{eq:ineqComp}
f^{d}\left(y\right) \leqslant f^{d-1}\left(y\right), \ \forall \ d \in \{2,\ldots,C_1\}.
\end{equation}
Let $d \in \{2,\ldots,C_1\}$, we have
\begin{align}
f^{d}\left(y\right) \ = \  & \underset{\zeta_{1},\ldots,\zeta_{d}}{\mathbb{E}} \left[ \min_{a_{1},\ldots,a_{d}} \left\{ \sum_{k=1}^{d} r\left(x^{d}_k\right) + f^o\left(y^{d}_{2}\right) \right\} \right] \nonumber
\\
= \  & \underset{\zeta_{1},\ldots,\zeta_{d-1}}{\mathbb{E}} \left[ \underset{\zeta_{d}}{\mathbb{E}} \left[ \min_{a_{1},\ldots,a_{d-1}} \left\{ \min_{a_{d}} \left\{ \sum_{k=1}^{d} r\left(x^{d}_k\right) + f^o\left(y^{d}_{2}\right) \right\} \right\} \right] \right]
\label{eq:beforeIneq}
\\
\leqslant \ & \underset{\zeta_{1},\ldots,\zeta_{d-1}}{\mathbb{E}} \left[ \min_{a_{1},\ldots,a_{d-1}} \left\{ \underset{\zeta_{d}}{\mathbb{E}} \left[ \min_{a_{d}} \left\{ \sum_{k=1}^{d} r\left(x^{d-1}_k\right) + f^o\left(y^{d-1}_{3}\right) \right\} \right] \right\} \right]
\label{eq:afterIneq}
\\
= \ & \underset{\zeta_{1},\ldots,\zeta_{d-1}}{\mathbb{E}} \left[ \min_{a_{1},\ldots,a_{d-1}} \left\{ \sum_{k=1}^{d-1} r\left(x^{d-1}_k\right) + \underset{\zeta_{d}}{\mathbb{E}} \left[ \min_{a_{d}} \left\{ r\left(x^{d-1}_{d}\right)  + f^o\left(y^{d-1}_{3}\right) \right\} \right] \right\} \right] \label{eq:afterEq}
\\
= \ & \underset{\zeta_{1},\ldots,\zeta_{d-1}}{\mathbb{E}} \left[ \min_{a_{1},\ldots,a_{d-1}} \left\{ \sum_{k=1}^{d-1}  r\left(x^{d-1}_k\right) + f^o\left(y^{d-1}_{2}\right) \right\} \right] = f^{d-1}\left(y\right), \nonumber
\end{align}
where the equality (\ref{eq:afterEq}) holds since $x_k^{d-1}$ for $k \in \{1,\ldots,d-1\}$ does not depend on $a_d$ and $\zeta_{d}$. Finally, the inequality holds because we have $\mathbb{E} \left[ \min \left\{ Z_1,\ldots,Z_m \right\} \right] \leqslant \min \left\{ \mathbb{E} \left[ Z_1,\ldots,Z_m \right] \right\}$ for any $Z_1,\ldots,Z_m$ random variables. Note that we changed $x_k^d$ in $x_k^{d-1}$ and $y_2^{d}$ in $y_3^{d-1}$. This change is made in order to stay consistent with the definition of Equation (\ref{eq:newProc}). Indeed, the order between the expectations and the minimums in Equation (\ref{eq:beforeIneq}) implies that the process of the first $d$ retrievals corresponds to
\[
y \xrightarrow{\zeta_{1},\ldots,\zeta_{d}} \ x^d_{1} \xrightarrow{a_{1}} \ x^d_{2} \xrightarrow{a_{2}} \ldots \xrightarrow{a_{d-1}} x_d^d \xrightarrow{a_{d}} y_2^{d},
\] 
while the order between the expectations and the minimums in Equation (\ref{eq:afterIneq}) corresponds to the following process for the first $d$ retrievals:
\[
y \xrightarrow{\zeta_{1},\ldots,\zeta_{d-1}} \ x^{d-1}_{1} \xrightarrow{a_{1}} \ x^{d-1}_{2} \xrightarrow{a_{2}} \ldots \xrightarrow{a_{d-1}} y_2^{d-1} \xrightarrow{\zeta_{d}} x_d^{d-1} \xrightarrow{a_{d}} y_3^{d-1}.
\]

Recall Equation (\ref{eq:fundrecrusion}) and apply it with $w=1$ (note that $K_1=1$ thus $K_1 + C_1 - 1 = C_1$) to get
\[
f\left(y\right) = \underset{\zeta_{1},\ldots,\zeta_{C_1}}{\mathbb{E}} \left[ \min_{a_{1},\ldots,a_{C_1}} \left\{ \sum_{k=1}^{C_1} r\left(x_k\right) + f\left(y_{2}\right) \right\} \right].
\]
By induction, for all configuration $y_2$ with $W-1$ batches we have $f\left(y_{2}\right) \leqslant f^o\left(y_{2}\right)$, thus
\[
f\left(y\right) = \underset{\zeta_{1},\ldots,\zeta_{C_1}}{\mathbb{E}} \left[ \min_{a_{1},\ldots,a_{C_1}} \left\{ \sum_{k=1}^{C_1} r\left(x_k\right) + f\left(y_{2}\right) \right\} \right] \leqslant \underset{\zeta_{1},\ldots,\zeta_{C_1}}{\mathbb{E}} \left[ \min_{a_{1},\ldots,a_{C_1}} \left\{ \sum_{k=1}^{C_1} r\left(x^{C_1}_k\right) + f^o\left(y^{C_1}_{2}\right) \right\} \right] = f^{C_1}\left(y\right),
\]
where we replaced $x_k$ by $x^{C_1}_k$ and $y_{2}$ by $y^{C_1}_{2}$ because, on the right-hand-side of the inequality, the revelation process after the first $C_1$ containers is the online model.
Finally, since $f^o\left(y\right) = f^{1}\left(y\right)$, by applying Equation (\ref{eq:ineqComp}) for each value of $d \in \{C_1,\ldots,2\}$, we complete the proof as
\[
f\left(y\right) \leqslant f^{C_1}\left(y\right) \leqslant f^{C_1-1}\left(y\right) \leqslant \ldots \leqslant f^2\left(y\right) \leqslant f^1\left(y\right) = f^o\left(y\right).
\]
\end{proof}

As a final remark, Lemma \ref{lem:compModels} is tight in the general setting. Indeed, there exists an initial configuration $y$ for which $f(y) = f^o(y)$. For instance consider the configuration in Figure \ref{fig:exampleSCRP}a, then we have $f(y) = f^o(y) = 13/6$.

\subsection{Computational comparison}
There also exist configurations for which $f(y) < f^o(y)$. The difference between these two values represents the value of taking into account available information (if possible).

In order to show a positive difference, we could have compared Experiments 1 and 3. However, since the average batch size is 2, these experiments do not show a positive difference between both models. Another possibility would have been to use the instances of Experiment 2. However, as we previously mentioned, such instances are hard to solve optimally and not approximately.

Instead, we consider another set of simpler instances randomly generated: 100 instances with $T=4$ tiers, $S=4$ stacks and $C = 12$ containers. Each instance has $W=3$ batches and each batch has $C_w = 4$ containers (for $w=1,2,3$). For each of these 100 instances, we solve it under the batch and the online models. The code and detailed results are available at \url{https://github.com/vgalle/StochasticCRP}. We are especially interested about $\displaystyle \frac{f^o(.)-f(.)}{f(.)} \times 100$, which the \% difference between the batch and the online models.

On average over the 100 instances, the optimal expected number of relocations under the batch model is $6.526$ and under the online model is $6.61$, hence giving a difference of $0.084$. We observe here that this difference represents more than $1.287\%$ of the optimal solution under the batch model, which is quite significant considering the fact that heuristic $EM$ experimentally lies within $2\%$ above the optimal solution. In addition, we noticed that for 25 of these instances, this difference was more than $2\%$ and the maximum was about $4\%$ (see Figure \ref{fig:DistributionDifference}a).

\begin{figure}[h]
\captionsetup{justification=centering}
\centering
\begin{minipage}{7.5cm}
\centering
\includegraphics[width=0.75\textwidth]{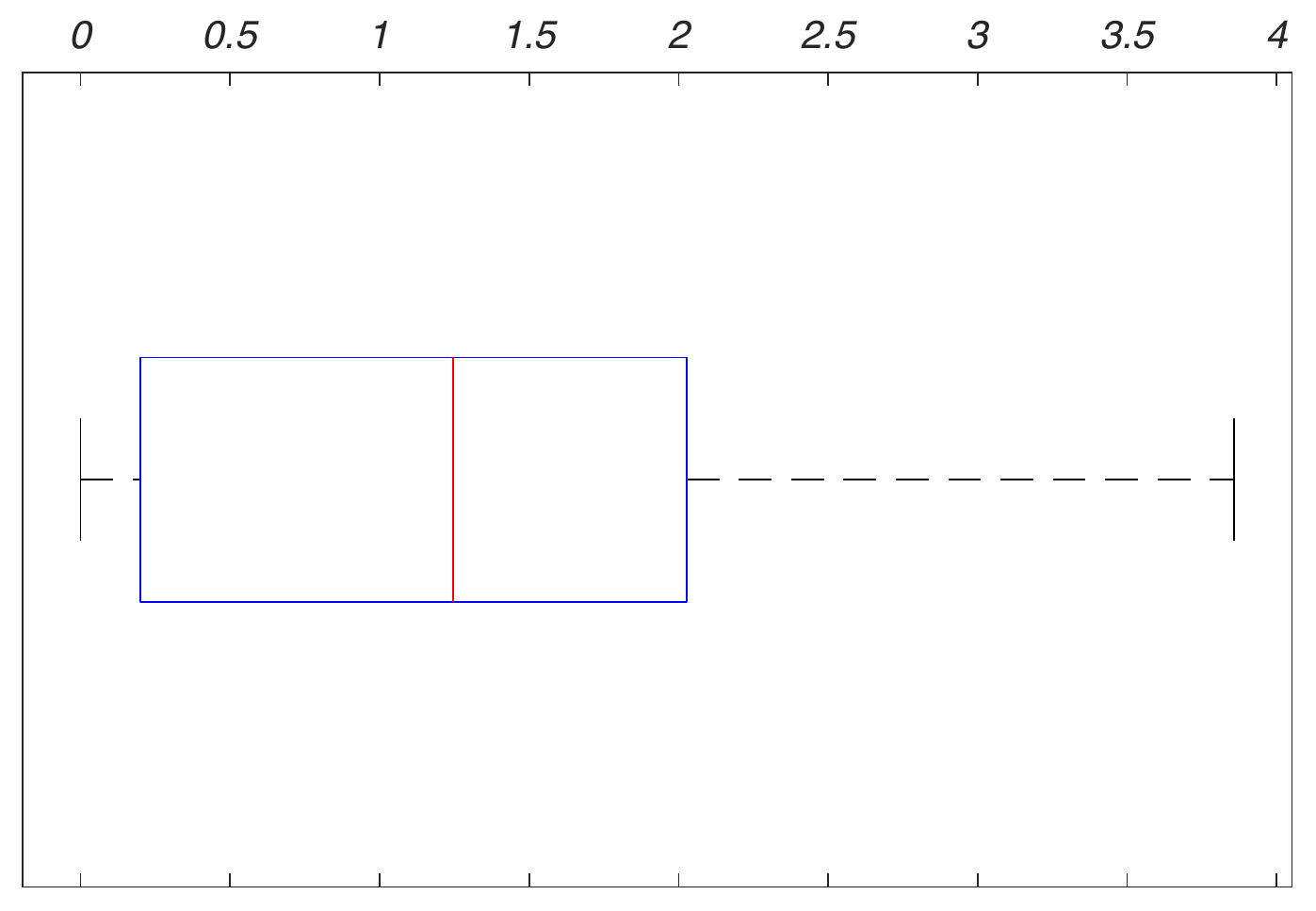}
\caption*{\ref{fig:DistributionDifference}a $T=4$, $S=4$, $C = 12$, $W=3$ and $C_w = 4$ (for $w=1,2,3$).}
\end{minipage}
\qquad \qquad
\begin{minipage}{7.5cm}
\centering
\includegraphics[width=\textwidth]{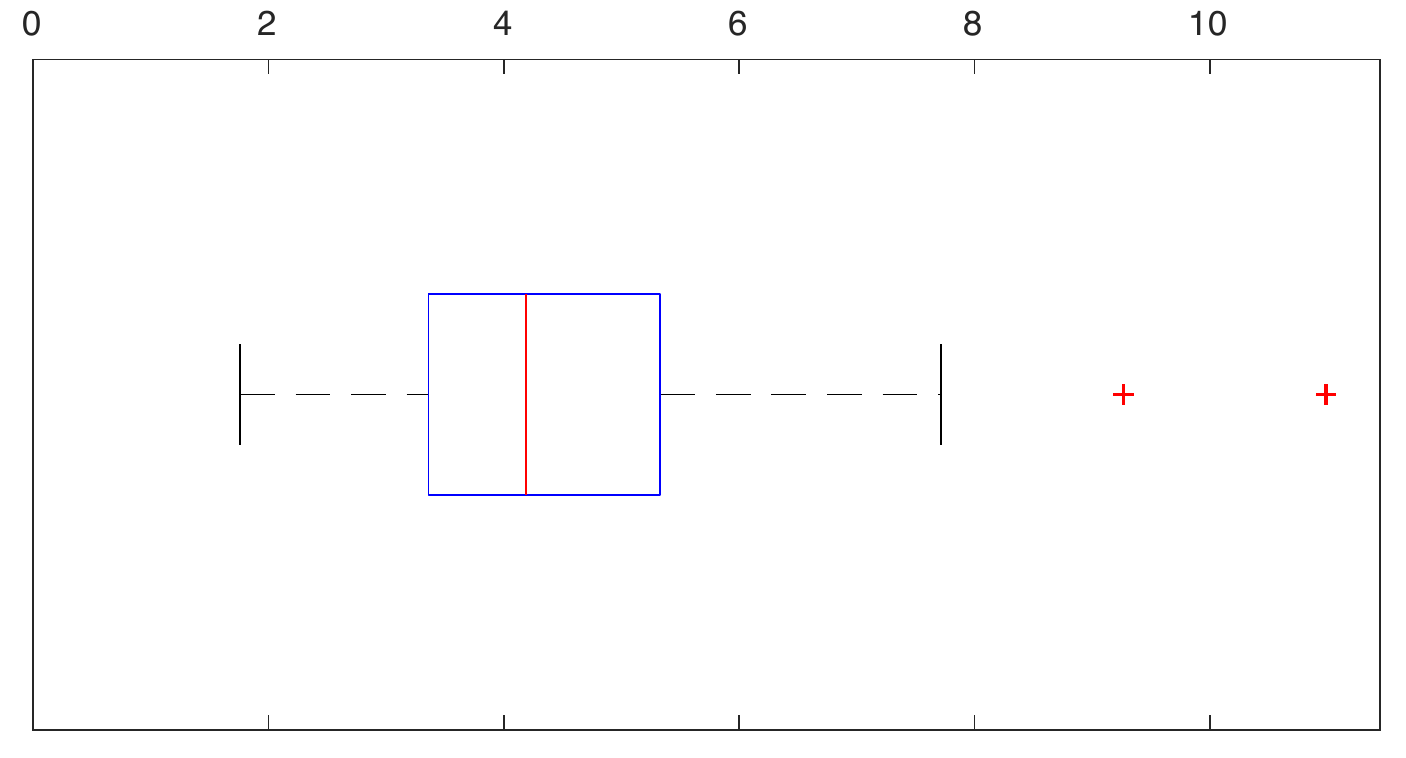}
\caption*{\ref{fig:DistributionDifference}b $T=4$, $S=4$, $C = 12$, $W=2$ and $C_w = 6$ (for $w=1,2$).}
\end{minipage}
\qquad
\caption{Distributions of \% difference between the batch and the online models from 100 randomly generated instances.}
\label{fig:DistributionDifference}
\end{figure}

We also consider 100 instances for which $T=4$, $S=4$, $C = 12$, but now $W=2$ and each batch has $C_w = 6$ (for $w=1,2$). Figure \ref{fig:DistributionDifference}b show that this relative difference appears to increase when the batch size increases. Indeed, the average difference is about $4.251\%$ (batch: $6.751$, online: $,7.038$, difference: $0.287$) with 25 instances having a difference of more than $5.3\%$.

\newpage
\section{Computational Experiments Tables}

\begin{table}[h!]
  \centering
  \ra{1}
    \resizebox{\textwidth}{!}{
\begin{tabular}{l l l l l l l l l l l l l l l l l}\toprule
 & &   & &   & &  \multicolumn{3}{l}{Lower bounds} & & PBFS & & \multicolumn{5}{l}{Heuristics} \\
 \cline{7-9} \cline{13-17} $S$ & & $T$ & & $C$ & & $b$ & $b_1$ & $b_2$ & &   & & EG & EM & ERI & L & Rand. \\
 \hline 5 & & 3 & & 8 & & 1.64 & 1.66 & 1.66 & & 1.70 & & \textbf{1.70} & \textbf{1.70} & \textbf{1.70} & 1.82 & 2.34 \\
		 & & 4 & & 10 & & 2.88 & 2.96 & 2.99 & & 3.11 & & \textbf{3.11} & 3.13 & 3.14 & 3.51 & 4.62 \\
		 & & 5 & & 13 & & 4.61 & 4.88 & 5.01 & & 5.32 & & 5.40 & \textbf{5.38} & 5.57 & 6.17 & 8.00 \\	
		 & & 6 & & 15 & & 6.28 & 6.64 & 7.06 & & 7.59 & & 7.85 & \textbf{7.81} & 8.09 & 9.41 & 12.35 \\
\hline 6 & & 3 & & 9 & & 1.68 & 1.69 & 1.69 & & 1.74 & & 1.76 & \textbf{1.74} & \textbf{1.74} & 1.84 & 2.43 \\
		 & & 4 & & 12 & & 3.54 & 3.61 & 3.63 & & 3.68 & & 3.69 & \textbf{3.68} & 3.68 & 4.11 & 5.59 \\
		 & & 5 & & 15 & & 5.37 & 5.57 & 5.68 & & 5.91 & & 5.97 & \textbf{5.94} & 6.01 & 7.21 & 9.55 \\
		 & & 6 & & 18 & & 7.19 & 7.52 & 7.68 & & 8.23 & & 8.38 & \textbf{8.29} & 8.63 & 10.05 & 13.53 \\
\hline 7 & & 3 & & 11 & & 2.82 & 2.86 & 2.88 & & 2.88 & & \textbf{2.88} & \textbf{2.88} & 2.89 & 2.96 & 4.12 \\
		 & & 4 & & 14 & & 3.97 & 4.06 & 4.10 & & 4.16 & & \textbf{4.17} & \textbf{4.17} & 4.20 & 4.65 & 6.47 \\
		 & & 5 & & 18 & & 6.49 & 6.65 & 6.74 & & 6.97 & & 7.05 & \textbf{7.00} & 7.07 & 8.46 & 11.27 \\
		 & & 6 & & 21 & & 8.82 & 9.21 & 9.51 & & -    & & 10.40 & \textbf{10.35} & 10.76 & 12.47 & 17.69 \\
\hline 8 & & 3 & & 12 & & 2.29 & 2.30 & 2.30 & & 2.31 & & \textbf{2.31} & \textbf{2.31} & \textbf{2.31} & 2.43 & 3.23 \\
		 & & 4 & & 16 & & 4.68 & 4.73 & 4.75 & & 4.82 & & \textbf{4.83} & \textbf{4.83} & \textbf{4.83} & 5.41 & 7.48 \\
		 & & 5 & & 20 & & 7.20 & 7.42 & 7.54 & & 7.85 & & 7.96 & \textbf{7.93} & 8.06 & 9.32 & 13.44 \\
		 & & 6 & & 24 & & 9.52 & 9.85 & 10.09 & & - & & 11.10 & \textbf{10.99} & 11.34 & 13.29 & 19.28 \\
\hline 9 & & 3 & & 14 & & 2.98 & 2.98 & 2.98 & & 3.00 & & \textbf{3.00} & 3.01 & \textbf{3.00} & 3.19 & 4.54 \\
		 & & 4 & & 18 & & 5.63 & 5.71 & 5.71 & & 5.73 & & \textbf{5.73} & \textbf{5.73} & \textbf{5.73} & 6.52 & 9.29 \\
		 & & 5 & & 23 & & 8.58 & 8.69 & 8.77 & & -    & & 9.05 & \textbf{9.02} & 9.12 & 11.16 & 15.57 \\
		 & & 6 & & 27 & & 10.38 & 10.78 & 10.98 & & - & & 11.59 & \textbf{11.58} & 11.76 & 14.62 & 20.93 \\
\hline 10 & & 3 & & 15 & & 3.18 & 3.18 & 3.18 & & 3.19 & & \textbf{3.19} & 3.20 & 3.20 & 3.27 & 4.75 \\
		  & & 4 & & 20 & & 6.20 & 6.23 & 6.23 & & 6.28 & & 6.30 & \textbf{6.28} & \textbf{6.28} & 6.98 & 10.41 \\
		  & & 5 & & 25 & & 9.10 & 9.37 & 9.39 & & -    & & \textbf{9.60} & \textbf{9.60} & 9.73 & 11.38 & 16.64 \\
		  & & 6 & & 30 & & 11.91 & 12.28 & 12.44 & & - & & 13.01 & \textbf{12.92} & 13.15 & 15.93 & 23.40 \\ \hline
\end{tabular}}
\caption{Results of experiment 1: Performance of PBFS, heuristics and tightness of lower bounds for a fill rate of 50 percent in the Batch Model, in the case of small batches. Bold numbers highlight the best heuristic for a given problem size.}
\label{tab:UBLB_50_1}
\end{table}

\begin{table}[h!]
  \centering
  \ra{1}
    \resizebox{\textwidth}{!}{
\begin{tabular}{l l l l l l l l l l l l l l l l l}\toprule
 & &   & &   & &  \multicolumn{3}{l}{Lower bounds} & & PBFS & & \multicolumn{5}{l}{Heuristics} \\
 \cline{7-9} \cline{13-17} $S$ & & $T$ & & $C$ & & $b$ & $b_1$ & $b_2$ & &   & & EG & EM & ERI & L & Rand. \\
 \hline 5 & & 3 & & 10 & & 2.83 & 2.97 & 3.01 & & 3.08 & & \textbf{3.08} & \textbf{3.08} & 3.11 & 3.33 & 4.16 \\
 		  & & 4 & & 13 & & 4.69 & 4.97 & 5.09 & & 5.58 & & 5.69 & \textbf{5.64} & 5.75 & 6.50 & 8.22 \\
 		  & & 5 & & 17 & & 7.58 & 8.52 & 8.92 & & -    & & \textbf{10.44} & 10.48 & 11.04 & 12.23 & 15.27 \\
		  & & 6 & & 20 & & 9.69 & 10.79 & 11.46 & & -  & & \textbf{14.53} & 14.65 & 15.77 & 18.47 & 22.93 \\
 \hline 6 & & 3 & & 12 & & 3.60 & 3.70 & 3.72 & & 3.89 & & \textbf{3.89} & 3.90 & 3.90 & 4.32 & 5.54 \\
 		  & & 4 & & 16 & & 6.20 & 6.59 & 6.78 & & 7.28 & & \textbf{7.41} & 7.45 & 7.61 & 8.55 & 11.29 \\
 		  & & 5 & & 20 & & 8.28 & 8.85 & 9.16 & & -    & & 10.39 & \textbf{10.38} & 10.80 & 12.85 & 16.43 \\
 		  & & 6 & & 24 & & 11.67 & 12.22 & 12.54 & & - & & 15.19 & \textbf{15.17} & 16.14 & 19.33 & 25.03 \\
 \hline 7 & & 3 & & 14 & & 3.85 & 3.89 & 3.91 & & 3.97 & & \textbf{3.98} & \textbf{3.98} & 4.02 & 4.40 & 6.05 \\
 		  & & 4 & & 19 & & 6.25 & 6.60 & 6.86 & & 7.29 & & 7.37 & \textbf{7.36} & 7.54 & 8.89 & 11.60 \\
 		  & & 5 & & 23 & & 9.72 & 10.24 & 10.55 & & -  & & 11.75 & \textbf{11.71} & 12.30 & 14.92 & 19.65 \\
 		  & & 6 & & 28 & & 13.52 & 14.41 & 14.93 & & - & & 17.72 & \textbf{17.63} & 18.81 & 23.10 & 30.78 \\
 \hline 8 & & 3 & & 12 & & 4.47 & 4.57 & 4.61 & & 4.66 & & \textbf{4.66} & \textbf{4.66} & 4.68 & 5.14 & 6.94 \\
 		  & & 4 & & 21 & & 7.62 & 7.85 & 7.98 & & 8.29 & & \textbf{8.33} & 8.35 & 8.43 & 9.76 & 13.26 \\
 		  & & 5 & & 27 & & 11.61 & 12.08 & 12.52 & & - & & 13.56 & \textbf{13.47} & 14.10 & 17.15 & 23.11 \\
 		  & & 6 & & 32 & & 15.60 & 16.39 & 16.78 & & - & & \textbf{19.28} & 19.51 & 20.85 & 25.89 & 34.66 \\
 \hline 9 & & 3 & & 18 & & 4.81 & 4.96 & 4.99 & & 5.10 & & \textbf{5.10} & 5.12 & 5.14 & 5.66 & 7.81 \\
 		  & & 4 & & 24 & & 8.98 & 9.18 & 9.30 & & 9.58 & & 9.63 & \textbf{9.61} & 9.76 & 11.66 & 16.00 \\
 		  & & 5 & & 30 & & 13.16 & 13.90 & 14.29 & & - & & \textbf{15.65} & 15.79 & 16.75 & 20.03 & 27.41 \\
 		  & & 6 & & 36 & & 16.77 & 17.36 & 17.83 & & - & & \textbf{20.38} & 20.40 & 21.86 & 28.12 & 38.12 \\
\hline 10 & & 3 & & 20 & & 5.21 & 5.21 & 5.21 & & 5.27 & & \textbf{5.28} & \textbf{5.28} & \textbf{5.28} & 5.79 & 7.86 \\
		  & & 4 & & 27 & & 9.18 & 9.54 & 9.71 & & -    & & \textbf{10.27} & 10.29 & 10.37 & 12.16 & 16.91 \\
		  & & 5 & & 34 & & 14.46 & 14.88 & 15.16 & & - & & \textbf{16.13} & 16.19 & 16.69 & 21.06 & 29.03 \\
		  & & 6 & & 40 & & 19.55 & 20.24 & 20.66 & & - & & 23.33 & \textbf{23.20} & 24.46 & 32.11 & 44.07 \\
 \hline
\end{tabular}}
\caption{Results of experiment 1: Performance of PBFS, heuristics and tightness of lower bounds for a fill rate of 67 percent in the Batch Model, in the case of small batches. Bold numbers highlight the best heuristic for a given problem size.}
\label{tab:UBLB_67_1}
\end{table}

\begin{table}[h!]
  \centering
  \ra{1}
    \resizebox{\textwidth}{!}{
\begin{tabular}{l l l l l l l l l l l l l l l l l}\toprule
 & &   & &   & &  \multicolumn{3}{l}{Lower bounds} & & PBFSA & & \multicolumn{5}{l}{Heuristics} \\
\cline{7-9} \cline{13-17} $S$ & & $T$ & & $C$ & & $b$ & $b_1$ & $b_2$ & &   & & EG & EM & ERI & L & Rand. \\
\hline 5 & & 3 & & 8 & & 1.68 & 1.68 & 1.68 & & 1.76 & & 1.77 & \textbf{1.76} & \textbf{1.76} & 1.87 & 2.38 \\
		 & & 4 & & 10 & & 2.96 & 3.02 & 3.05 & & 3.33 & & \textbf{3.36} & 3.37 & 3.38 & 3.67 & 4.85 \\
		 & & 5 & & 13 & & 4.58 & 4.82 & 4.93 & & 5.50 & & 5.69 & \textbf{5.65} & 5.74 & 6.33 & 8.22 \\
		 & & 6 & & 15 & & 6.31 & 6.72 & 6.96 & & 7.81 & & 8.28 & \textbf{8.03} & 8.35 & 9.49 & 12.40 \\	 								
\hline 6 & & 3 & & 9 & & 1.66 & 1.67 & 1.67 & & 1.73 & & 1.75 & \textbf{1.74} & \textbf{1.74} & 1.82 & 2.43 \\
		 & & 4 & & 12 & & 3.61 & 3.71 & 3.73 & & 3.93 & & \textbf{3.99} & \textbf{3.99} & 4.00 & 4.34 & 5.82 \\
		 & & 5 & & 15 & & 5.38 & 5.57 & 5.64 & & 6.11 & & \textbf{6.23} & \textbf{6.23} & 6.31 & 7.16 & 9.65 \\
		 & & 6 & & 18 & & 7.01 & 7.26 & 7.44 & & -    & & 8.49 & \textbf{8.36} & 8.61 & 9.92 & 13.45 \\					
\hline 7 & & 3 & & 11 & & 2.76 & 2.79 & 2.79 & & 2.85 & & 2.84 & 2.84 & \textbf{2.83} & 2.95 & 4.06 \\
		 & & 4 & & 14 & & 4.02 & 4.12 & 4.15 & & 4.24 & & 4.31 & \textbf{4.29} & 4.31 & 4.73 & 6.52 \\
		 & & 5 & & 18 & & 6.29 & 6.39 & 6.43 & & 6.77 & & 7.00 & \textbf{6.92} & 7.01 & 8.20 & 11.08 \\
		 & & 6 & & 21 & & 8.69 & 9.12 & 9.35 & & -    & & 10.60 & \textbf{10.52} & 10.91 & 12.61 & 17.79 \\						
\hline 8 & & 3 & & 12 & & 2.30 & 2.31 & 2.31 & & 2.31 & & \textbf{2.31} & 2.32 & 2.32 & 2.37 & 3.19 \\
		 & & 4 & & 16 & & 4.61 & 4.62 & 4.63 & & 4.71 & & \textbf{4.74} & \textbf{4.74} & 4.75 & 5.25 & 7.40 \\
		 & & 5 & & 20 & & 7.31 & 7.46 & 7.52 & & -    & & \textbf{8.01} & \textbf{8.01} & 8.09 & 9.33 & 13.25 \\
		 & & 6 & & 24 & & 9.65 & 9.95 & 10.12 & & -   & & \textbf{11.37} & \textbf{11.37} & 11.67 & 13.44 & 19.51 \\ 								
\hline 9 & & 3 & & 14 & & 2.93 & 2.93 & 2.93 & & 2.95 & & \textbf{2.96} & \textbf{2.96} & \textbf{2.96} & 3.15 & 4.48 \\
		 & & 4 & & 18 & & 5.56 & 5.58 & 5.59 & & 5.69 & & 5.74 & \textbf{5.70} & \textbf{5.70} & 6.33 & 9.07 \\
		 & & 5 & & 23 & & 8.49 & 8.64 & 8.73 & & -    & & 9.16 & \textbf{9.12} & 9.16 & 10.99 & 15.39 \\
		 & & 6 & & 27 & & 10.38 & 10.69 & 10.90 & & - & & 11.77 & \textbf{11.75} & 11.95 & 14.65 & 20.95 \\						
\hline 10 & & 3 & & 15 & & 3.15 & 3.16 & 3.16 & & 3.15 & & \textbf{3.17} & \textbf{3.17} & \textbf{3.17} & 3.25 & 4.72 \\
		  & & 4 & & 20 & & 6180 & 6.20 & 6.21 & & 6.28 & & 6.35 & \textbf{6.34} & \textbf{6.34} & 6.92 & 10.27 \\
		  & & 5 & & 25 & & 9.13 & 9.31 & 9.36 & & -    & & 9.66 & \textbf{9.63} & 9.68 & 11.44 & 16.73 \\
		  & & 6 & & 30 & & 12.09 & 12.35 & 12.51 & & - & & 13.38 & \textbf{13.23} & 13.42 & 16.33 & 23.80 \\ \hline								
\end{tabular}}
\caption{Results of experiment 2: Performance of PBFSA, heuristics and tightness of lower bounds for a fill rate of 50 percent in the Batch Model with larger batches. Bold numbers highlight the best heuristic for a given problem size.}
\label{tab:UBLB_50_2}
\end{table}
\begin{table}[h!]
  \centering
  \ra{1}
    \resizebox{\textwidth}{!}{
\begin{tabular}{l l l l l l l l l l l l l l l l l}\toprule
 & &   & &   & &  \multicolumn{3}{l}{Lower bounds} & & PBFSA & & \multicolumn{5}{l}{Heuristics} \\
\cline{7-9} \cline{13-17} $S$ & & $T$ & & $C$ & & $b$ & $b_1$ & $b_2$ & &   & & EG & EM & ERI & L & Rand. \\
\hline 5 & & 3 & & 10 & & 2.78 & 2.87 & 2.90 & & 3.07 & & \textbf{3.08} & \textbf{3.08} & 3.09 & 3.36 & 4.18 \\
		 & & 4 & & 13 & & 4.71 & 4.92 & 5.00 & & 5.70 & & 5.81 & \textbf{5.80} & 5.89 & 6.60 & 8.23 \\
 		 & & 5 & & 17 & & 7.53 & 8.17 & 8.44 & & -    & & 10.38 & \textbf{10.32} & 10.68 & 11.96 & 14.9815 \\
 		 & & 6 & & 20 & & 9.69 & 10.56 & 11.12 & & -  & & 15.00 & \textbf{14.91} & 16.01 & 18.22 & 22.79 \\
\hline 6 & & 3 & & 12 & & 3.56 & 3.63 & 3.64 & & 3.90 & & \textbf{3.90} & 3.91 & 3.92 & 4.28 & 5.55 \\
		 & & 4 & & 16 & & 6.12 & 6.48 & 6.61 & & 7.17 & & 7.36 & \textbf{7.41} & 7.48 & 8.44 & 11.02 \\				
 		 & & 5 & & 20 & & 8.33 & 8.72 & 8.90 & & -    & & 10.36 & \textbf{10.30} & 10.57 & 12.54 & 16.15 \\
 		 & & 6 & & 24 & & 11.77 & 12.41 & 12.80 & & - & & 15.94 & \textbf{15.91} & 16.86 & 19.83 & 25.45 \\ 
\hline 7 & & 3 & & 14 & & 3.95 & 4.01 & 4.02 & & 4.14 & & 4.17 & \textbf{4.15} & \textbf{4.15} & 4.56 & 6.11 \\
		 & & 4 & & 19 & & 6.27 & 6.56 & 6.77 & & -    & & 7.36 & \textbf{7.35} & 7.52 & 8.66 & 11.46 \\
 		 & & 5 & & 23 & & 9.81 & 10.27 & 10.54 & & -  & & \textbf{12.08} & 12.10 & 12.49 & 14.78 & 19.67 \\
 		 & & 6 & & 28 & & 13.64 & 14.47 & 14.91 & & - & & 18.36 & \textbf{18.26} & 19.32 & 23.10 & 31.03 \\
\hline 8 & & 3 & & 12 & & 4.65 & 4.74 & 4.76 & & 4.85 & & 4.86 & \textbf{4.85} & 4.88 & 5.34 & 7.14 \\
		 & & 4 & & 21 & & 7.58 & 7.86 & 7.99 & & -    & & \textbf{8.38} & 8.42 & 8.50 & 9.82 & 13.20 \\
 		 & & 5 & & 27 & & 11.46 & 11.98 & 12.28 & & - & & 13.73 & \textbf{13.60} & 14.11 & 17.00 & 22.84 \\
 		 & & 6 & & 32 & & 15.45 & 16.28 & 16.72 & & - & & 19.94 & \textbf{19.83} & 21.21 & 25.73 & 34.70 \\
\hline 9 & & 3 & & 18 & & 4.85 & 4.98 & 5.02 & & 5.14 & & \textbf{5.17} & 5.20 & 5.20 & 5.67 & 7.77 \\
		 & & 4 & & 24 & & 8.82 & 9.00 & 9.11 & & -    & & 9.66 & \textbf{9.60} & 9.72 & 11.52 & 15.70 \\
		 & & 5 & & 30 & & 13.15 & 13.84 & 14.18 & & - & & \textbf{15.91} & 15.96 & 16.82 & 20.16 & 27.35 \\
		 & & 6 & & 36 & & 16.85 & 17.39 & 17.80 & & - & & 20.99 & \textbf{20.83} & 21.97 & 28.05 & 38.04 \\
\hline 10 & & 3 & & 20 & & 5.19 & 5.21 & 5.22 & & 5.31 & & \textbf{5.31} & \textbf{5.31} & \textbf{5.31} & 5.79 & 7.92 \\
		  & & 4 & & 27 & & 9.40 & 9.66 & 9.82 & & -    & & 10.47 & \textbf{10.46} & 10.54 & 12.25 & 16.96 \\
		  & & 5 & & 34 & & 14.44 & 14.83 & 15.07 & & - & & \textbf{16.29} & \textbf{16.29} & 16.62 & 21.12 & 28.86 \\
		  & & 6 & & 40 & & 19.49 & 20.24 & 20.66 & & - & & 23.83 & \textbf{23.65} & 24.96 & 32.04 & 44.12 \\
 \hline
\end{tabular}}
\caption{Results of experiment 2: Performance of PBFSA, heuristics and tightness of lower bounds for a fill rate of 67 percent in the Batch Model with larger batches. Bold numbers highlight the best heuristic for a given problem size.}
\label{tab:UBLB_67_2}
\end{table}

\begin{table}[h!]
  \centering
  \ra{1}
    \resizebox{\textwidth}{!}{
\begin{tabular}{l l l l l l l l l l l l l l l l l}\toprule
 & &   & &   & &  \multicolumn{3}{l}{Lower bounds} & & PBFS & & \multicolumn{5}{l}{Heuristics} \\
 \cline{7-9} \cline{13-17} $S$ & & $T$ & & $C$ & & $b$ & $b_1$ & $b_2$ & &   & & EG & EM & ERI & L & Rand. \\
 \hline 5 & & 3 & & 8 & & 1.64 & 1.66 & 1.66 & & 1.70 & & \textbf{1.70} & \textbf{1.70} & 1.71 (1.71) & 1.82 & 2.34 (2.34) \\
		     & & 4 & & 10 & & 2.88 & 2.96 & 2.99 & & 3.11 & & \textbf{3.11} & 3.13 & 3.14 (3.20) & 3.51 & 4.62 (4.62) \\
		     & & 5 & & 13 & & 4.61 & 4.88 & 5.01 & & 5.32 & & \textbf{5.38} & \textbf{5.38} & 5.57 (5.58) & 6.16 & 8.00 (8.00) \\
		     & & 6 & & 15 & & 6.28 & 6.64 & 7.06 & & 7.59 & & 7.85 & \textbf{7.80} & 8.08 (8.29) & 9.41 & 12.36 (12.35) \\
  \hline 6 & & 3 & & 9 & & 1.68 & 1.69 & 1.69 & & 1.74 & & 1.76 & \textbf{1.74} & \textbf{1.74} (1.75) & 1.84 & 2.43 (2.43) \\
  			 & & 4 & & 12 & & 3.54 & 3.61 & 3.63 & & 3.68 & & 3.69 & \textbf{3.68} & \textbf{3.68} (3.75) & 4.11 & 5.59 (5.59) \\
  			 & & 5 & & 15 & & 5.37 & 5.57 & 5.68 & & 5.91 & & 5.96 & \textbf{5.94} & 6.00 (6.18) & 7.21 & 9.54 (9.54) \\
  			 & & 6 & & 18 & & 7.19 & 7.52 & 7.68 & & 8.23 & & 8.38 & \textbf{8.29} & 8.62 (8.77) & 10.05 & 13.53 (13.53) \\
  \hline 7 & & 3 & & 11 & & 2.82 & 2.86 & 2.88 & & 2.88 & & \textbf{2.88} & \textbf{2.88} & 2.89 (2.88) & 2.96 & 4.11 (4.11) \\
  			 & & 4 & & 14 & & 3.97 & 4.06 & 4.1 & & 4.16 & & \textbf{4.17} & \textbf{4.17} & 4.21 (4.20*) & 4.66 & 6.47 (6.03*) \\
  			 & & 5 & & 18 & & 6.49 & 6.65 & 6.74 & & 6.97 & & 7.04 & \textbf{7.00} & 7.07 (7.18) & 8.45 & 11.27 (11.27) \\
  			 & & 6 & & 21 & & 8.82 & 9.21 & 9.51 & & - & & 10.40 & \textbf{10.35} & 10.76 (10.98) & 12.46 & 17.69 (17.69) \\
  \hline 8 & & 3 & & 12 & & 2.29 & 2.3 & 2.3 & & 2.31 & & \textbf{2.31} & \textbf{2.31} & \textbf{2.31} (2.32) & 2.43 & 3.23 (3.23) \\
  			 & & 4 & & 16 & & 4.68 & 4.73 & 4.75 & & 4.82 & & \textbf{4.83} & \textbf{4.83} & \textbf{4.83} (4.88) & 5.41 & 7.49 (7.49) \\
  			 & & 5 & & 20 & & 7.20 & 7.42 & 7.54 & & 7.85 & & 7.97 & \textbf{7.94} & 8.07 (8.27) & 9.32 & 13.44 (13.45) \\
  			 & & 6 & & 24 & & 9.52 & 9.85 & 10.09 & & - & & 11.10 & \textbf{10.98} & 11.34 (11.61) & 13.29 & 19.29 (19.29) \\
  \hline 9 & & 3 & & 14 & & 2.98 & 2.98 & 2.98 & & 3.00 & & \textbf{3.00} & \textbf{3.00} & \textbf{3.00} (3.00) & 3.19 & 4.54 (4.54) \\
  			 & & 4 & & 18 & & 5.63 & 5.71 & 5.71 & & 5.73 & & \textbf{5.73} & \textbf{5.73} & \textbf{5.73} (5.80) & 6.52 & 9.29 (9.29) \\
  			 & & 5 & & 23 & & 8.58 & 8.69 & 8.77 & & - & & 9.05 & \textbf{9.02} & 9.12 (9.36) & 11.16 & 15.56 (15.57) \\
  			 & & 6 & & 27 & & 10.38 & 10.78 & 10.98 & & - & & 11.59 & \textbf{11.58} & 11.76 (12.09) & 14.62 & 20.94 (20.93) \\
  \hline 10 & & 3 & & 15 & & 3.18 & 3.18 & 3.18 & & 3.19 & & \textbf{3.19} & 3.20 & 3.20 (3.20) & 3.27 & 4.75 (4.75) \\
  			 & & 4 & & 20 & & 6.20 & 6.23 & 6.23 & & 6.28 & & 6.30 & \textbf{6.27} & 6.28 (6.33) & 6.98 & 10.41 (10.41) \\
  			 & & 5 & & 25 & & 9.10 & 9.37 & 9.39 & & - & & 9.61 & \textbf{9.60} & 9.73 (9.80) & 11.38 & 16.64 (16.63) \\
  			 & & 6 & & 30 & & 11.91 & 12.28 & 12.44 & & - & & 13.01 & \textbf{12.92} & 13.15 (13.51) & 15.92 & 23.41 (23.41) \\
 \hline
\end{tabular}}
\caption{Results of experiment 3: Performance of heuristics and tightness of lower bounds for a fill rate of 50 percent in the online model with small batches. Bold numbers highlight the best heuristic for a given problem size. Numbers in parenthesis are taken from \cite{Ku}.}
\label{tab:UBLB_50_3}
\end{table}

\begin{table}[h!]
  \centering
  \ra{1}
    \resizebox{\textwidth}{!}{
\begin{tabular}{l l l l l l l l l l l l l l l l l}\toprule
 & &   & &   & &  \multicolumn{3}{l}{Lower bounds} & & PBFS & & \multicolumn{5}{l}{Heuristics} \\
 \cline{7-9} \cline{13-17} $S$ & & $T$ & & $C$ & & $b$ & $b_1$ & $b_2$ & &   & & EG & EM & ERI & L & Rand. \\
 \hline 5 & & 3 & & 10 & & 2.83 & 2.97 & 3.01 & & 3.08 & & \textbf{3.08} & \textbf{3.08} & 3.12 (3.10) & 3.33 & 4.16 (4.16) \\
		     & & 4 & & 13 & & 4.69 & 4.97 & 5.09 & & 5.58 & & 5.68 & \textbf{5.64} & 5.75 (5.80) & 6.50 & 8.22 (8.22) \\
		     & & 5 & & 17 & & 7.58 & 8.52 & 8.92 & & - & & \textbf{10.45} & 10.48 & 11.04 (11.15) & 12.24 & 15.28 (15.28) \\
		     & & 6 & & 20 & & 9.69 & 10.79 & 11.46 & & - & & \textbf{14.53} & 14.65 & 15.77 (16.14) & 18.46 & 22.93 (22.93) \\
  \hline 6 & & 3 & & 12 & & 3.6 & 3.7 & 3.72 & & 3.89 & & \textbf{3.89} & 3.90 & 3.90 (3.92) & 4.32 & 5.53 (5.53) \\
  			 & & 4 & & 16 & & 6.2 & 6.59 & 6.78 & & 7.28 & & \textbf{7.41} & 7.45 & 7.61 (7.68) & 8.54 & 11.29 (11.28) \\
  			 & & 5 & & 20 & & 8.28 & 8.85 & 9.16 & & - & & \textbf{10.38} & \textbf{10.38} & 10.80 (10.97) & 12.85 & 16.42 (16.42) \\
  			 & & 6 & & 24 & & 11.67 & 12.22 & 12.54 & & - & & \textbf{15.17} & \textbf{15.17} & 16.14 (16.65) & 19.33 & 25.04 (25.03) \\
  \hline 7 & & 3 & & 14 & & 3.85 & 3.89 & 3.91 & & 3.97 & & \textbf{3.98} & \textbf{3.98} & 4.02 (4.01) & 4.40 & 6.05 (6.05) \\
  			 & & 4 & & 19 & & 6.25 & 6.6 & 6.86 & & 7.29 & & 7.37 & \textbf{7.36} & 7.54 (7.68) & 8.89 & 11.60 (11.61) \\
  			 & & 5 & & 23 & & 9.72 & 10.24 & 10.55 & & - & & 11.76 & \textbf{11.71} & 12.30 (12.64) & 14.92 & 19.66 (19.65) \\
  			 & & 6 & & 28 & & 13.52 & 14.41 & 14.93 & & - & & 17.70 & \textbf{17.64} & 18.82 (19.49) & 23.10 & 30.77 (30.79) \\
  \hline 8 & & 3 & & 12 & & 4.47 & 4.57 & 4.61 & & 4.66 & & \textbf{4.65} & 4.66 & 4.68 (4.7) & 5.14 & 6.94 (6.94) \\
  			 & & 4 & & 21 & & 7.62 & 7.85 & 7.98 & & 8.29 & & \textbf{8.32} & 8.35 & 8.43 (8.5) & 9.75 & 13.26 (13.25) \\
  			 & & 5 & & 27 & & 11.61 & 12.08 & 12.52 & & - & & 13.56 & \textbf{13.47} & 14.10 (14.44) & 17.14 & 23.11 (23.12) \\
  			 & & 6 & & 32 & & 15.6 & 16.39 & 16.78 & & - & & \textbf{19.27} & 19.51 & 20.85 (21.72) & 25.89 & 34.64 (34.63) \\
  \hline 9 & & 3 & & 18 & & 4.81 & 4.96 & 4.99 & & 5.10 & & \textbf{5.10} & 5.12 & 5.14 (5.19) & 5.66 & 7.80 (7.80) \\
  			 & & 4 & & 24 & & 8.98 & 9.18 & 9.3 & & 9.58 & & 9.63 & \textbf{9.61} & 9.76 (9.92) & 11.66 & 16.01 (16.00) \\
  			 & & 5 & & 30 & & 13.16 & 13.9 & 14.29 & & - & & \textbf{15.65} & 15.79 & 16.75 (16.97) & 20.03 & 27.38 (27.39) \\
  			 & & 6 & & 36 & & 16.77 & 17.36 & 17.83 & & - & & \textbf{20.38} & 20.40 & 21.87 (22.73) & 28.13 & 38.11 (38.14) \\
  \hline 10 & & 3 & & 20 & & 5.21 & 5.21 & 5.21 & & 5.27 & & \textbf{5.28} & \textbf{5.28} & \textbf{5.28} (5.30) & 5.79 & 7.85 (7.86) \\  
  			& & 4 & & 27 & & 9.18 & 9.54 & 9.71 & & - & & \textbf{10.27} & 10.29 & 10.37 (10.50) & 12.15 & 16.92 (16.91) \\
  			& & 5 & & 34 & & 14.46 & 14.88 & 15.16 & & - & & \textbf{16.13} & 16.19 & 16.69 (17.23) & 21.07 & 29.03 (29.03) \\
  			& & 6 & & 40 & & 19.55 & 20.24 & 20.66 & & - & & 23.33 & \textbf{23.20} & 24.46 (25.58) & 32.11 & 44.08 (44.07) \\
 \hline
\end{tabular}}
\caption{Results of experiment 3: Performance of heuristics and tightness of lower bounds for a fill rate of 67 percent in the online model with small batches. Bold numbers highlight the best heuristic for a given problem size. Numbers in parenthesis are taken from \cite{Ku}.}
\label{tab:UBLB_67_3}
\end{table}
 
\end{document}